\documentclass[10pt,aps,nofootinbib,prx,superscriptaddress]{revtex4-2}
\usepackage[ruled,vlined]{algorithm2e}
\usepackage{xcolor}
\definecolor{prxlink}{RGB}{46,48,146}
\usepackage[colorlinks=true,
            linkcolor=prxlink,
            urlcolor=prxlink,
            citecolor=prxlink,
            anchorcolor=prxlink]{hyperref}
\usepackage{amsfonts,amssymb,amsthm,bm,graphicx,mathtools,physics,qcircuit,scalerel,tikz,upgreek}
\usepackage{cleveref,enumitem}
\usepackage[T1]{fontenc}

\theoremstyle{plain}
\newtheorem{theorem}{Theorem}

\newtheorem{result}{Result}
\newtheorem{proposition}[theorem]{Proposition}
\newtheorem{lemma}[theorem]{Lemma}
\newtheorem{fact}[theorem]{Fact}
\newtheorem{corollary}[theorem]{Corollary}
\theoremstyle{definition}
\newtheorem{definition}[theorem]{Definition}

\theoremstyle{remark}
\newtheorem{remark}[theorem]{Remark}

\newcommand{\ci}{\ensuremath{\mathrm{i}}}

\newcommand{\rsbk}{{\mathrm{RSB}_k}}
\newcommand{\loc}{\mathrm{loc}}
\newcommand{\pauli}{\mathrm{Pauli}}
\newcommand{\EE}[1]{\mathbb{E}\left[#1\right]}
\newcommand{\E}{\mathbb{E}}
\newcommand{\KL}{\mathrm{KL}}

\newcommand{\cA}{\mathcal{A}}
\newcommand{\cB}{\mathcal{B}}
\newcommand{\cC}{\mathcal{C}}
\newcommand{\cD}{\mathcal{D}}
\newcommand{\cE}{\mathcal{E}}

\newcommand{\cH}{\mathcal{H}}

\newcommand{\cL}{\mathcal{L}}

\newcommand{\cN}{\mathcal{N}}
\newcommand{\cO}{\mathcal{O}}
\newcommand{\cP}{\mathcal{P}}

\newcommand{\cS}{\mathcal{S}}
\newcommand{\cT}{\mathcal{T}}
\newcommand{\cU}{\mathcal{U}}

\newcommand{\lr}[1]{\left(#1\right)}

\begin{document}
\title{Average-case quantum complexity from glassiness}

\author{Alexander Zlokapa}
\email{azlokapa@mit.edu}
\affiliation{Center for Theoretical Physics, MIT, 77 Massachusetts Ave., Cambridge, MA 02139, USA}
\author{Bobak T.\ Kiani}
\email{b.kiani@bowdoin.edu}
\affiliation{Department of Computer Science, Bowdoin College, 255 Maine Street, Brunswick, Maine 04011, USA}
\author{Eric R.\ Anschuetz}
\email{eans@caltech.edu}
\affiliation{Institute for Quantum Information and Matter, Caltech, 1200 E. California Blvd., Pasadena, CA 91125, USA}
\affiliation{Walter Burke Institute for Theoretical Physics, Caltech, 1200 E. California Blvd., Pasadena, CA 91125, USA}

\begin{abstract}
    We provide a framework for average-case quantum complexity by showing that glassiness obstructs a natural family of quantum algorithms. Glassiness---a phenomenon in physics characterized by a rough free energy landscape---is known to imply hardness for stable classical algorithms; for example, constant-time Langevin dynamics or message-passing fail for random $k$-SAT and max-cut problems in a glassy parameter regime. We present comparable results in the quantum setting.
    \begin{itemize}[rightmargin=7em]
        \item \emph{Quantum optimal transport implications of glassiness.} We show that the standard notion of quantum glassiness in physics implies that the Gibbs state is decomposed into clusters extensively separated in quantum Wasserstein distance. We prove this implies lower bounds on the quantum Wasserstein complexity of channels from non-glassy to glassy states.
        \item \emph{Structural argument for hardness.} We define \emph{stable quantum algorithms} in terms of Lipschitz temperature dependence. We prove that constant-time local Lindbladian evolution and shallow variational algorithms are stable and hence fail to capture all clusters of the Gibbs state, yielding a geometrically interpretable algorithmic obstruction. Unlike mixing time lower bounds, our results hold even when starting from the maximally mixed state.
    \end{itemize}
    Our results apply to non-commuting, non-stoquastic quantum Hamiltonians that satisfy a formal definition of glassiness. To show glassiness, we rely on the replica trick, a widely used but non-rigorous method in statistical physics. We find that the ensemble of all 3-local Pauli strings with independent Gaussian coefficients is average-case hard. As a complementary result, we also provide strong analytical evidence that the general $p$-local Pauli ensemble is \emph{not} glassy for sufficiently large constant $p$; this is distinct from classical (Ising $p$-spin, glassy phase for any $p$) and fermionic (SYK, never glassy) analogues.\\
    \\
    \textit{This is a preliminary draft of the paper posted to arXiv due to the mandatory QIP 2026 conference deadline. We intend to provide further numerical analyses in a v2 posting.}
\end{abstract}

\maketitle

\section{Introduction}

We study the algorithmic implications of \emph{glassiness} in the quantum setting. A phenomenon characterized by slow dynamics and a disordered energy landscape, glassiness is observed in many natural (classical and quantum) Hamiltonians~\cite{mezard1987spin,charbonneau2023spin}: examples include random $k$-SAT~\cite{kirkpatrick1994critical,monasson1997statistical,monasson1999determining,mezard2002analytic}, classical and quantum Heisenberg models~\cite{bray1981metastable,georges2001quantum,christos2022spin,kavokine2024exact}, and the classical and quantum Sherrington-Kirkpatrick model~\cite{thouless1977solution,almeida1978,yamamoto1987perturbation,kopec1989instabilities}.

Glassy classical Hamiltonians are believed to obstruct efficient classical algorithms in both sampling and optimization tasks. Arguments for this \emph{average-case complexity} in the non-oracle model are formalized via, e.g., the overlap gap property~\cite{achlioptas2006solution,mezard2005clustering,ogp}, low-degree likelihood ratio~\cite{barak2019nearly,hopkins2017bayesian,hopkins2017power,hopkins2018statistical}, or transport disorder chaos~\cite{el2022sampling,el2025sampling,el2025shattering}.
Similar average-case results are largely absent in the quantum setting because the classical proof techniques are difficult to generalize. 
Among the many challenges in generalizing arguments to the quantum setting, quantum Hamiltonians do not have explicitly known eigenbases, leading to a sign problem that prevents the application of common probabilistic tools.

Our main results overcome this issue and formalize a connection between quantum glassiness and quantum algorithms, providing a natural approach to a theory of average-case quantum complexity. This allows us to address several open questions~\cite{chen2024sparse,basso2024optimizing,gamarnik2024slowmixingquantumgibbs,swingle2023bosonicmodelquantumholography,ramkumar2024mixing,anschuetz2025efficientlearningimpliesquantum}, such as the following.
\begin{itemize}
    \item[Q1:] Is it hard to efficiently prepare the Gibbs state (at constant temperature) of random $p$-local Hamiltonians, the quantum Heisenberg model (of which quantum max-cut is a special case), or other non-commuting non-stoquastic Hamiltonians?
    \item[Q2:] Can lower bounds be shown for the runtime of Lindbladian Gibbs samplers to prepare the Gibbs state starting from the maximally mixed state? (In comparison, mixing time lower bounds from Markov theory only guarantee lower bounds for Lindbladians applied to a worst-case initial state.)
    \item[Q3:] Do the physics and algorithmic hardness of the random $p$-local Pauli string Hamiltonian ensemble more closely resemble its classical (Ising $p$-spin) or fermionic (SYK) analogues?
\end{itemize}
We answer Q1 in the affirmative with evidence comparable to analogous average-case hardness results in the classical literature. Specifically, we define a family of ``stable'' quantum algorithms that generalizes classical notions of stability and show that these stable algorithms fail to prepare the Gibbs state of Hamiltonian ensembles in Q1. We prove that Q2 is encompassed by our definition of stability. We resolve Q3 by computing the one-step and full replica symmetry breaking solutions of the $p$-local Hamiltonian ensemble using methods from statistical physics.

To describe our results in more detail, we introduce a definition of glassiness that coincides with the notion of a quantum spin glass in the physics literature. We define it in terms of the following bilinear form between two $n$-qubit states:
\begin{equation}\label{eq:pauliform}
    \langle \bm \rho, \bm \rho'\rangle_\pauli := \sum_{r=1}^n \sum_{\mu \in \{X,Y,Z\}} \Tr[\sigma_r^\mu \bm \rho] \Tr[\sigma_r^\mu \bm \rho'],
\end{equation}
where $\sigma_r^\mu$ denotes Pauli matrix $\mu$ acting on qubit $r$. This induces a seminorm $\|\bm \rho - \bm \rho'\|_\pauli^2 = \langle \bm \rho, \bm \rho\rangle_\pauli + \langle \bm \rho', \bm \rho'\rangle_\pauli - 2\langle \bm \rho, \bm \rho'\rangle_\pauli$\footnote{It is a seminorm because there may exist inputs $\bm X \neq 0$ for which $\|\bm X\|_\pauli^2 = 0$.}. In the physics literature, glassiness is typically characterized by replica symmetry breaking (RSB) or shattering~\cite{edwards1975theory,parisi1979infinite,monasson1995structural,franz1998effective}, which are special cases of the following definition. Informally, both are characterized by a clustering of the Gibbs state according to the Pauli seminorm.

\begin{definition}[Glassiness]\label{def:introglass}
    Let $(\cH_n)_n$ be a sequence of $n$-qubit Hamiltonian ensembles. The ensemble is \emph{glassy} at inverse temperature $\beta$ if with probability $1-o(1)$, the Gibbs state $\bm \rho = e^{-\beta H}/\Tr[e^{-\beta H}]$ for random instance $H \sim \cH_n$ satisfies
    \begin{align}\label{eq:firstdecomp}
        \frac{1}{2}\norm{\bm \rho - \sum_{i=1}^{m(n)} c_i \bm \rho_i}_1 = o(1), \qquad c_i > 0 \; \forall \, i \in [m(n)], \qquad \sum_{i=1}^{m(n)} c_i = 1 - o(1),
    \end{align}
    where $m(n) \geq 2$ for sufficiently large $n$, and $\bm \rho_i$ are $n$-qubit states satisfying the following properties.
    \begin{enumerate}
        \item Separate clusters. There exists some constant $q > 0$ such that for sufficiently large $n$,
        \begin{align}
            \forall \, i, j \in [m(n)] \text{ s.t. } i \neq j, \; \frac{1}{n} \norm{\bm \rho_i - \bm \rho_j}_\pauli^2 \geq q.
        \end{align}
        \item Equal clusters. Each $\bm \rho_i$ appears with asymptotically similar weight, i.e.,
        \begin{equation}
            \forall \,i \in [m(n)], \; c_i = \exp[-\Theta(\log m(n))].
        \end{equation}
    \end{enumerate}
    If \eqref{eq:firstdecomp} holds up to error $\epsilon$, we refer to clusters $\cC = \{(c_i, \bm \rho_i)\}_{i=1}^{m(n)}$ as witnessing the $(\epsilon, q, m(n))$-decomposition of the Gibbs state. If no such $\cC$ satisfying these properties exists, we say $\bm{\rho}$ is \emph{non-glassy}.
\end{definition}

\begin{remark}\label{rem:glass}
    In \Cref{sec:glasses}, we show that this definition describes replica symmetry breaking and shattering for both classical Hamiltonians that are local in $\sigma^Z$ operators and quantum Hamiltonians local in Pauli strings. We also elaborate on the relation of the parameters $q$ and $m$ to quantities such as the TAP complexity~\cite{Thouless01031977}.
\end{remark}

Crucially, this notion of clustering is independent of the Hamiltonian eigenbasis. In comparison, the slow mixing lower bounds of~\cite{gamarnik2024slowmixingquantumgibbs,rakovszky2024bottlenecksquantumchannelsfinite} on Lindbladian Gibbs samplers~\cite{chen2023efficientexactnoncommutativequantum} required a basis that approximately block diagonalized the Hamiltonian; each block corresponded to a cluster under a distance metric defined by Lindbladian jump operators. Similarly, the recent work of~\cite{anschuetz2025efficientlearningimpliesquantum} chose a random basis and examined if the resulting classical probability distribution was clustered. Because of this limitation of an unknown basis, prior work was unable to prove tight lower bounds for natural models such as dense, random $p$-local Hamiltonians. Our \Cref{def:introglass} avoids imposing a basis; it instead corresponds to the notion of clustering found in physics, and thus we expect it to produce tight lower bounds for many natural models.

At a technical level, we show algorithmic implications of glassiness by relating it to quantum optimal transport. We show that glassiness implies lower bounds on the quantum Wasserstein complexity (WC)~\cite{li2022wasserstein} of any quantum channel that takes a non-glassy state to a glassy state, where a channel $\bm \cE$ from $n$-qubit to $n$-qubit states has WC:
\begin{align}
    \mathrm{WC}(\bm \cE) := \max_{\bm \rho} \norm{\bm \rho - \bm \cE(\bm \rho)}_{W_1},
\end{align}
where the maximization is over $n$-qubit states $\bm \rho$ and $\norm{\cdot}_{W_1}$ denotes the quantum Wasserstein distance~\cite{9420734}.

\begin{theorem}[Extensive Wasserstein complexity across the glass transition, informal]\label{thm:wass}
    Let $\cH_n$ be a Hamiltonian ensemble that is glassy at inverse temperature $\beta$ and non-glassy at $\beta'$. For $H \sim \cH_n$, denote the corresponding Gibbs states by $\bm \rho_\beta$ and $\bm \rho_{\beta'}$. Then with probability $1-o(1)$, there exists constant $\epsilon_* > 0$ such that for all $0 < \epsilon < \epsilon_*$, any quantum channel $\bm\varPhi$ that satisfies
    \begin{align}
        \frac{1}{2}\norm{\bm \rho_\beta - \bm\varPhi(\bm \rho_{\beta'})}_1 \leq \epsilon
    \end{align}
    must also satisfy
    \begin{align}
        \mathrm{WC}(\bm\varPhi) \geq \alpha n
    \end{align}
    for constant $\alpha$ independent of $\epsilon$.
\end{theorem}

Note that this approach is distinct from classical proofs relating glassiness to optimal transport (e.g., transport temperature/disorder chaos~\cite{el2022sampling,el2025sampling,el2025shattering}). For non-commuting Hamiltonians, one cannot directly diagonalize $\bm \rho_\beta, \bm \rho_\beta'$ and apply probabilistic methods since the eigenbasis is not a priori known, and changes with the disorder. This also makes bounding the quantum Wasserstein distance between glassy and non-glassy states difficult, prompting our use of WC instead. Similarly, we define stable quantum algorithms in terms of WC. Our definition is similar to the disorder-Lipschitz criterion previously used to define stable algorithms~\cite{el2022sampling,el2025sampling,anschuetz2025efficientlearningimpliesquantum}.

\begin{definition}[Stable quantum algorithm, informal]\label{def:temp_stab_inf}
    Let $(\bm \cA_n)_n$ be a sequence of quantum algorithms and $H_n \sim \cH_n$ be sampled from a sequence of Hamiltonian ensembles $(\cH_n)_n$, where the algorithm $\cA_n(H_n,\beta)$ maps an $n$-qubit Hamiltonian and inverse temperature to an $n$-qubit quantum state. The sequence $(\bm \cA_n)_n$ is \emph{temperature stable} with respect to the Hamiltonian ensemble if with probability $1-o(1)$, for any temperatures $\beta,\beta' > 0$, there exists a constant $L > 0$ such that
    \begin{align}
        \inf_{\bm\varPhi \,:\, \norm{\bm\varPhi(\bm \cA_n(H_n, \beta)) - \bm \cA_n(H_n, \beta')}_1 = o(1)} \mathrm{WC}(\bm\varPhi) \leq Ln|\beta-\beta'|,
    \end{align}
    where the infimum is evaluated over all channels $\bm\varPhi$ that take $n$-qubit states to $n$-qubit states.
\end{definition}

We emphasize that stable algorithms include classically intractable algorithms: we do \emph{not} require that the algorithm itself has low Wasserstein complexity. For instance, quantum algorithms such as Decoded Quantum Interferometry (DQI) and phase estimation are known to satisfy similar definitions of stability~\cite{anschuetz2025decodedquantuminterferometryrequires,anschuetz2025efficientlearningimpliesquantum}.

Classically, algorithms such as local Langevin dynamics are stable up to constant time~\cite{ben2018spectral,el2025shattering}. 
We show here that the quantum Lindbladian analogue is also stable, as well as parameterized shallow circuits such as ``hardware efficient'' variational algorithms~\cite{kandala2017hardware}.

\begin{theorem}[Examples of stable quantum algorithms, informal]\label{thm:algs}
    For any initial state $\bm \rho$, evolution under the Lindbladian $e^{\bm{\cL}_\beta t}[\bm \rho]$ for any time $t=O(1)$ is stable, where
    \begin{equation}\label{eq:introlind}
        \bm{\mathcal{L}}_\beta\left[\bm{\rho}\right]=-\ci\left[\bm{\hat H},\bm{\rho}\right]+\sum_i\gamma_i\left(\beta\right)\left(\bm{A}_i\bm{\rho}\bm{A}_i^\dagger-\frac{1}{2}\left\{\bm{A}_i^\dagger\bm{A}_i,\bm{\rho}\right\}\right)
    \end{equation}
    such that $\bm{\hat H}$ is Hermitian, $\gamma_i>0$ are Lipschitz functions of $\beta$, and $\{\bm A_i\}$ are a set of $O(n)$ temperature-independent constant-local jump operators with $\norm{\bm A_i} = O(1)$. Moreover, constant-depth algorithms of the form:
    \begin{equation}
        \bm{\mathcal{A}}\left(\beta\right) := \bigcirc_\ell \left(\mathcal{U}_\ell \circ \exp\left(\bm{\mathcal{L}}_{\ell,\beta}\right)\right)\left[\bm{\rho}_0\right],
    \end{equation}
    are stable, where $\cU_\ell$ are constant-depth parameterized unitaries composed of 2-local gates on a sparse, constant-degree hypergraph, and $\bm\cL_{\ell,\beta}$ have the form of \eqref{eq:introlind}. Note that in both of the above algorithms, $\cU_\ell$ and $\bm A_i$ may depend on the Hamiltonian disorder and are only required to be independent of the temperature.
\end{theorem}

At the glass transition temperature $1/\beta$, the state undergoes a phase transition from non-glassy to glassy, which we show results in an $\Omega(n)$ change in Wasserstein distance. The WC of a stable algorithm can be made arbitrarily small across this phase transition by choosing sufficiently small constant $|\beta'-\beta|$, violating the glassy WC bound. This gives the following result.

\begin{theorem}[Glassiness obstructs stable algorithms]\label{thm:introlb}
    Let $\bm \cA$ be a temperature stable algorithm and $\cH$ be a Hamiltonian ensemble that is non-glassy at temperatures $\beta < \beta_*$ and glassy at temperatures $\beta > \beta_*$ for some constant $\beta_* > 0$. Then there exist constants $\beta, \delta > 0$ such that with high probability
    \begin{align}
        \frac{1}{2} \norm{\bm \cA(H, \beta) - \bm \rho_\beta}_1 \geq \delta,
    \end{align}
    where $\bm \rho_\beta = e^{-\beta H}/\Tr[e^{-\beta H}]$.
\end{theorem}

The state produced by a stable algorithm is not only far in trace distance from $\bm \rho_\beta$ but is also structurally distinct, containing the wrong number of clusters as measured by $\norm{\cdot}_\pauli$ (as in \Cref{def:introglass}). We state this more carefully as follows (proven in \Cref{lem:observable_from_Pauli,lem:wc_preserves_shattering}).

\begin{remark}[Topological barrier]
    We prove a stronger (and topologically interpretable) sense of separation between $\bm{\cA}(H,\beta)$ and $\rho_\beta$ than described by trace distance. Assume that the algorithm fails such that $\bm \cA(H, \beta)$ is not glassy at some $\beta > \beta_*$. Then $\bm \cA(H, \beta)$ can be close to at most one of the clusters $\bm \rho_i$ in the decomposition $\bm \rho_\beta = \sum_i c_i \bm \rho_i$, i.e., $\norm{\bm \cA(H, \beta) - \bm \rho_i}_\pauli^2 = o(n)$. The algorithm's output is locally distinguishable from all other clusters: for every $\bm \rho_j \neq \bm \rho_i$, there is a bounded one-local observable $X = \sum_{r=1}^n X_r$ for which $\Tr[X (\bm \cA(H, \beta)-\bm \rho_j)] = \Omega(1)$.
\end{remark}

\begin{figure}
    \centering
    \includegraphics[width=0.9\linewidth]{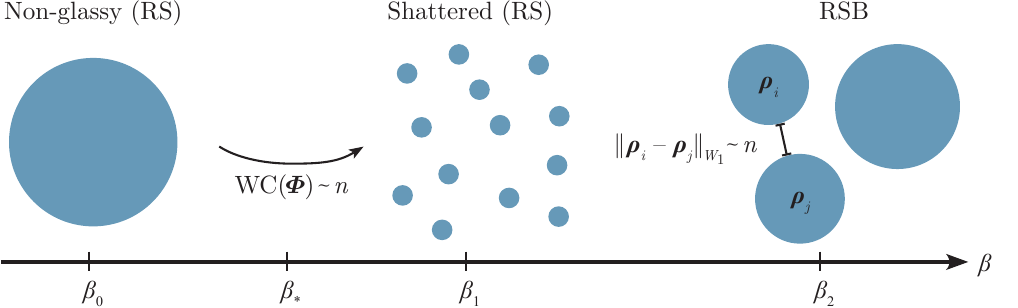}
    \caption{Cartoon of the glass decomposition $\bm \rho_\beta = \sum_i c_i \bm \rho_i$ of a Gibbs state that is non-glassy, shattered, and 1RSB. At a high temperature $\beta_0$, the Gibbs state is replica symmetric and non-glassy, and is dominated by one ``cluster'' as measured by $\norm{\cdot}_\pauli$ \eqref{eq:pauliform}. At a cooler temperatures $\beta_1$, the Gibbs state may shatter into a glassy phase defined by exponentially many clusters $\bm \rho_i$, which we show are separated by $\Omega(n)$ in quantum Wasserstein distance. Although the state remains replica symmetric in the shattered phase (due to each cluster being exponentially small), we show that any quantum channel satisfying $\bm \varPhi(\bm \rho_{\beta_0}) \approx \bm \rho_{\beta_1}$ must have Wasserstein complexity $\Omega(n)$. At an even cooler temperature $\beta_2$, the Gibbs state is decomposed into a constant number of clusters due to (static) replica symmetry breaking. Since the transition from a non-glassy to a glassy phase occurs abruptly at some constant $\beta_*$, any algorithm that is Lipschitz in $\beta$ will fail at the glass transition.}
    \label{fig:glass}
\end{figure}

Previous lower bounds for Gibbs state preparation via Lindbladian dynamics addressed only worst-case initial states~\cite{gamarnik2024slowmixingquantumgibbs,rakovszky2024bottlenecksquantumchannelsfinite}, whereas we prove hardness from any initial state (including the maximally mixed state). We also note that~\cite{ding2025lower} uses WC to show no-fast-forwarding results for specific Lindbladians (rather than the task of Gibbs state preparation).

This topological description of the algorithmic obstruction in terms of quantum optimal transport gives a new route to analyzing algorithms. 
Past work relating glassiness to hardness for quantum algorithms considered only \emph{classical} Hamiltonians\footnote{Or reduced the hardness of quantum Hamiltonians to that of classical Hamiltonians~\cite{anschuetz2025efficientlearningimpliesquantum}.}. This line of work examined optimization of these classical systems and mostly studied the quantum adiabatic algorithm~\cite{jorg2010first,young2010first,bapst2013quantum} or QAOA~\cite{basso2022performance,chou2021limitations}. In comparison, our results establish hardness for classical spin glasses across a more generic family of quantum algorithms that represents the quantum generalization of stable classical algorithms such as Langevin dynamics. More importantly, our results also encompass quantum non-commuting, non-stoquastic Hamiltonians. To emphasize this fact, we study the following ensemble of random $p$-local Pauli Hamiltonians.

\begin{definition}[Random $p$-local Hamiltonian ensemble]\label{def:pspin}
    Let $\cP_{n,p}$ be the set of all Pauli strings of weight $p$ acting on $n$ qubits. The ensemble of random $p$-local Hamiltonians is defined as the distribution of
    \begin{align}
        H_{n,p} = \sum_{P \in \cP_{n,p}} J_P P, \qquad J_P \sim_\mathrm{iid} \cN\lr{0, \frac{J^2(p-1)!}{n^{p-1}}}.
    \end{align}
\end{definition}

This ensemble has been previously studied using tools from high-dimensional probability and random matrix theory~\cite{chen2023efficientexactnoncommutativequantum,erdHos2014phase,anschuetz2025bounds}. However, these direct approaches to studying spectral properties have largely obtained loose results; for example, \cite{chen2023efficientexactnoncommutativequantum} only studied the non-local case at $p=n$. We pursue here the approach began by~\cite{swingle2023bosonicmodelquantumholography} using the replica trick to probe the glassiness of the $p$-local ensemble at arbitrary constant $p$. The replica trick~\cite{mezard1987spin} is a method from statistical physics that is non-rigorous due to an analytic continuation that relies on the re-ordering of limits; nonetheless, it is widely used for quantum models (e.g., see Kitaev's analysis of the SYK model~\cite{kitaev2015talks}) due to the difficulty of formally controlling the path integral in quantum mechanics, and in practice is predictive of the physics of such models.

\begin{result}[Glassiness of random $p$-local Hamiltonian ensemble]\label{res:glass}
    For $p=3$, there exists a constant temperature $\beta_*$ such that the random $p$-local Hamiltonian ensemble is non-glassy for all $\beta < \beta_*$ and glassy for all $\beta > \beta_*$. For sufficiently large constant $p$, the random $p$-local Hamiltonian ensemble is non-glassy at all $\beta > 0$ sub-exponential in $p$.
\end{result}

\begin{corollary}[Hardness of 3-local Hamiltonians]
    Let $\bm \rho_\beta$ denote the Gibbs state of the random 3-local $n$-qubit Hamiltonian ensemble of \Cref{def:pspin}. There exists a constant inverse temperature $\beta > 0$ and a constant $\delta > 0$ independent of $n$ such that stable quantum algorithms w.h.p. output a state $\bm \cA(H, \beta)$ satisfying $\frac{1}{2}\norm{\bm \rho_\beta - \bm \cA(H, \beta)}_1 \geq \delta$.
\end{corollary}

To show our glassiness result at $p=3$, we analytically examine the stability of the replica symmetric solution to obtain bounds on the glass transition temperature. We confirm these bounds by solving the replica saddle point equations numerically via a continuous-time Monte Carlo algorithm; note that these numerics take place in the $n\to\infty$ limit of system size due to the replica trick. Prior work~\cite{swingle2023bosonicmodelquantumholography} predicted that the $p=3$ model has a glassy phase by examining the ground states of random Hamiltonian instances on $n \leq 17$ qubits. Our result thus upholds this prediction in the thermodynamic limit at constant $\beta$. Moreover, we also derive the full RSB saddle point equations, which can be numerically solved to completely characterize the physics of random 3-local Hamiltonians.

To show our non-glassy result at large $p \geq p_*$, we analyze the replica saddle point equations via a sharpened log-Sobolev inequality and prove that no solution exists outside the replica symmetric solution. This method only addresses temperatures satisfying $\beta J = \exp[o(p)]$. Beyond this result mentioned in \Cref{res:glass}, we estimate the value of the smallest non-glassy $p$ using our log-Sobolev bound. This suggests that non-glassy physics begins at some $4 \lesssim p_* \lesssim 8$. The finite-$n$ ground state numerics of~\cite{swingle2023bosonicmodelquantumholography} suggest that even at extensive $\beta$, $p=5$ is non-glassy, closely matching our result. In combination, we regard these two results as strong evidence that the ensemble has no glassy phase even for fairly small $p$. This yields a surprising phase diagram compared to those of the fermionic and classical analogues.
\begin{itemize}
    \item Classical ensemble (Ising $p$-spin)~\cite{montanari2003nature,ferrari2012two,talagrand2000rigorous,gamarnik2023shatteringisingpurepspin,alaoui2024near}: glassy phase exists for any constant $p$.
    \item Fermionic ensemble (SYK)~\cite{sachdev1993gapless,kitaev2015talks,maldacena2016remarks}: no glassy phase at any constant $p$.
    \item Sparsified variant of $p$-local Pauli ensemble~\cite{anschuetz2025efficientlearningimpliesquantum}: glassy phase exists for any constant $p$.
    \item $p$-local Pauli ensemble (\textbf{this work}): glassy phase exists for small constant $p$; non-glassy at all sufficiently large constant $p$.
\end{itemize}

Evaluating replica solutions for models on spin-1/2 operators is analytically difficult, although the saddle point equations can be solved exactly via numerical integration (see, e.g.,~\cite{kavokine2024exact}). Besides finite-$n$ numerics, \cite{swingle2023bosonicmodelquantumholography} provided an analytical replica symmetric (RS) solution in the $p\to\infty$ limit that assumed vanishing off-diagonals in the replica ansatz. Here, we derive the 1RSB (and full RSB) saddle point equations for any $p$. We analytically control the 1RSB solution with large $p$ by a sharpened log-Sobolev inequality~\cite{fathi2014quantitative}, which allows us to rule out glassiness. To identify the glass phase at $p=3$, we find that it suffices to numerically evaluate the RS saddle point equations with generic off-diagonal elements in the replica ansatz. We also provide an analytical estimate of the glass transition temperature following an expansion similar to~\cite{bray1980replica}.

We summarize open questions that we believe can be answered given the framework established in this paper.
\begin{itemize}
    \item \textbf{Additional algorithms and Hamiltonians.} Our main results show sufficient conditions for a Hamiltonian ensemble to be hard for stable quantum algorithms. Although we provide new examples of stable algorithms and glassy Hamiltonians, we expect many other previously studied algorithms to be stable and Hamiltonians to be glassy. Examples may include systems on lattices~\cite{rieger1994zero,guo1994quantum,laumann2008cavity,morita2005gauge,itoi2024gauge}, Gibbs samplers based on quantum phase estimation~\cite{jiang2024quantummetropolissamplingweak}, and quasi-local Lindbladian dynamics~\cite{chen2023efficientexactnoncommutativequantum}.
    \item \textbf{Mixing time lower bounds for non-stoquastic Hamiltonians.} Our main results lower-bound the time required for Lindbladian evolution with local jumps to converge to the Gibbs state from an arbitrary starting state. Although there are known sufficient conditions for exponential-time mixing lower bounds (i.e., from a \emph{worst-case} starting state) to hold for Gibbs samplers on generic non-commuting, non-stoquastic Hamiltonians~\cite{gamarnik2024slowmixingquantumgibbs,rakovszky2024bottlenecksquantumchannelsfinite}, the only examples where the analysis is tractable are either stoquastic or code-like. We expect that the definition of glassiness provided in this work implies slow mixing, similarly to the classical case (e.g., see Corollary 2.9 of~\cite{el2025shattering}).
    \item \textbf{Lower bounds for optimization.} We present average-case lower bounds for the task of Gibbs sampling. Another natural task is optimization, where the goal is to construct a state below some energy. A quantum version of the \emph{overlap gap property} (OGP)~\cite{ogp} can be defined in terms of the Pauli overlap we use. We anticipate that lower bounds on optimization can be shown for Hamiltonians satisfying such a quantum OGP, and that spin-glass methods can be used to provide examples of systems with a quantum OGP.
    \item \textbf{New upper bounds.} Although our main results concern lower bounds, we identify two paths towards improving our understanding of quantum upper bounds. First, our work can be used to identify promising Hamiltonian ensembles for quantum algorithms. We expect that non-glassy systems are quantumly easy; this belief is supported by the Schwinger-Keldysh computation in quantum disordered systems~\cite{cugliandolo1993analytical,cugliandolo1998quantum,cugliandolo1999real,biroli2001quantum,kennett2001time,biroli2002out,cugliandolo2006dissipative}, which yields rapid mixing when coupled to a bath. As an example, while the results of~\cite{chen2024sparse} showed that random $n$-local Pauli Hamiltonians are quantumly easy, our results suggest that random constant-local Pauli Hamiltonians may also be quantumly easy, giving a possible avenue for more natural settings for quantum advantage. Second, by giving insight into why quantum algorithms fail, our work suggests key features for successful quantum algorithms. The geometric nature of our lower bounds provides additional perspective on analyzing quantum algorithms via optimal transport.
\end{itemize}
We also remark on two additional open questions relevant to this work that remain difficult to address but would be valuable to resolve. The first question asks for a rigorous proof of quantum glassiness. All of the examples of glasses shown above are established via non-rigorous methods in statistical physics based on the replica trick; producing rigorous proofs of classical glasses required decades of effort~\cite{guerra2002thermodynamic,guerra2003broken,talagrand2003spin,talagrand2006parisi,panchenko2013sherrington}, and thus we expect the quantum setting to also be difficult. The second question asks to extend average-case lower bounds for sampling beyond stable algorithms, to capture, e.g., Lindbladian evolution for polynomial time starting from the maximally mixed state. Classically, it also remains open if lower bounds can be shown for Langevin dynamics beyond constant time~\cite{el2025shattering}, and thus we also expect this open question to be challenging.

\section{Glasses and quantum optimal transport}
\label{sec:glasses}

In this section, we formally relate replica symmetry (RS), replica symmetry breaking (RSB), and shattering to quantum optimal transport, showing the formal version of \Cref{thm:wass} in the introduction. We begin in \Cref{sec:paulioverlap} by precisely defining RS, RSB and shattering in terms of the Pauli bilinear form of \eqref{eq:pauliform}. 
We will see that RSB and shattering are special cases of the glassy decomposition of \Cref{def:introglass}. 
We then show in \Cref{sec:transport} that glassiness implies a lower bound on the Wasserstein complexity of quantum channels between non-glassy and glassy states, which will ultimately yield our desired algorithmic obstruction.

Below, we use $\sigma_r^\mu$ to denote Pauli matrix $\mu \in \{X,Y,Z\}$ acting on qubit $r$. We use $\cS_n$ to denote the set of $n$-qubit states.

\subsection{Glassiness from Pauli overlaps}\label{sec:paulioverlap}
We first define a replica symmetric state, which may or may not be glassy as per \Cref{def:introglass}.

\begin{definition}[Replica symmetry]\label{def:rs}
    Let $(\cH_n)_n$ be a sequence of $n$-qubit Hamiltonian ensembles. The ensemble is \emph{replica symmetric} at inverse temperature $\beta$ if with probability $1-o(1)$ over draws $H \sim \cH_n$ , and for every $1 \leq m(n) \leq \exp[\Theta(n)]$ admitting the below decomposition, the Gibbs state $\bm \rho = e^{-\beta H}/\Tr[e^{-\beta H}]$ satisfies
    \begin{align}
        \frac{1}{2}\norm{\bm \rho - \sum_{i=1}^{m(n)} c_i \bm \rho_i}_1 = o(1), \qquad c_i > 0 \; \forall \, i \in [m(n)], \qquad \sum_{i=1}^{m(n)} c_i = 1 - o(1),
    \end{align}
    where $\bm \rho_i$ are $n$-qubit states satisfying the following properties:
    \begin{enumerate}
        \item Separate clusters: There exists constant $q > 0$ such that for sufficiently large $n$,
        \begin{align}
            \forall \, i, j \in [m(n)] \text{ s.t. } i \neq j, \; \frac{1}{n} \norm{\bm \rho_i - \bm \rho_j}_\pauli^2 \geq q.
        \end{align}
        \item Equal weight clusters: Each $\bm \rho_i$ appears with asymptotically similar weight,
        \begin{align}
            \forall \,i \in [m(n)], \; c_i = \exp[-\Theta(\log (m(n))].
        \end{align}
        \item Concentration to single overlap: Let $\cD$ be the probability distribution that draws $\bm \rho_i$ with probability $c_i/\cN$ for $\cN = \sum_i c_i$. Then there exists constant $q^* \in \mathbb R$ such that
        \begin{align}
            \mathrm{Pr}_{\bm \rho_i, \bm \rho_j \sim_\mathrm{iid}\cD}\left[\left|\frac{1}{n}\norm{\bm \rho_i - \bm \rho_j}_\pauli^2 - q^*\right| \leq \delta\right] = 1 - o(1)
        \end{align}
        for all $\delta > 0$.
    \end{enumerate}
\end{definition}

In the classical setting, replica symmetry has an operational meaning: with high probability, the Hamming distance between two spin configurations drawn independently from the Gibbs distribution concentrates to a single value. In contrast, replica symmetry \emph{breaking} requires the Hamming distance of pairs to have nonzero probability over at least two distinct values when pairs are drawn i.i.d. from the Gibbs distribution. Hence, no Gibbs distribution can be simultaneously RSB and RS. In the quantum setting, we replace Hamming distance with $\norm{\cdot}_\pauli^2$. In the following definition, we make it explicitly clear that RSB and shattering are special cases of glassiness (\Cref{def:introglass}).

\begin{definition}[Shattering and RSB]\label{def:shatterrsb}
    A Gibbs state satisfies \emph{replica symmetry breaking} if it can be decomposed as \Cref{def:introglass} with $m(n) = \Theta(1)$. (Recall that per \Cref{def:introglass}, we require $m(n) \geq 2$ for sufficiently large $n$.) A Gibbs state is \emph{shattered} if it can be decomposed as \Cref{def:introglass} with $m(n) = \exp[\Theta(n)]$.
\end{definition}
\begin{remark}[Comparison to physics convention]
    The number of clusters $m(n)$ in the decomposition $\bm \rho = \sum_{i=1}^{m(n)} c_i \bm \rho_i$ does not necessarily correspond precisely to the true ``TAP states'' referred to in the physics literature~\cite{monasson1995structural,franz1998effective,mezard1987spin,thouless1977solution,biroli2001quantum}. For example, to satisfy the definition of RSB above, it suffices to group different TAP states into a single $\bm \rho_i$.
\end{remark}

Note that shattering does not necessarily imply replica symmetry breaking. Indeed, a shattered state can be RS. Since $m= \exp[\Theta(n)]$, distinct clusters ($i \neq j$) are drawn with probability $1-o(1)$ in the distribution $\cD$ of \Cref{def:rs}; if the Pauli overlap between distinct pairs of clusters concentrates to a single value, then the state will satisfy the RS definition.

The definitions above encompass classical spin glasses that are local in only $\sigma^Z$ operators and non-commuting quantum spin glasses that are local in Pauli strings. Both of these classes of Hamiltonians feature a \emph{glass transition} from non-glassy RS to either shattered or RSB at a certain critical temperature.

\begin{definition}[Glass transition]
    Let $(\cH_n)_n$ be a sequence of $n$-qubit Hamiltonian ensembles. The ensemble is said to have a \emph{glass transition} if there exist constant temperatures $\beta_- < \beta_* < \beta_+$ such that the ensemble is replica symmetric but not glassy for all $\beta \in (\beta_-, \beta_*)$ and glassy (\Cref{def:introglass}) for all $\beta \in (\beta_*, \beta_+)$.
\end{definition}

Historically, the glass transition was first identified using the \emph{replica trick}, which evaluates in the $n\to\infty$ limit the quenched free energy $-\EE{\log Z}/\beta$ for partition function $Z = \Tr[e^{-\beta H}]$ by computing $\EE{Z^s}$ for integer values of $s$ and performing an analytic continuation of $s \to 0$ \cite{kirkpatrick1987p}. 
The replica trick was confirmed numerically to correctly predict the physics of many models, such as the Ising $p$-spin model (and the special case of the Sherrington-Kirkpatrick model, which sets $p=2$).
This computation is non-rigorous due to its treatment of limits and assumption of an ansatz that corresponds precisely to the above decomposition into clusters.

\begin{remark}[Glass transition of the Ising $p$-spin Hamiltonian ensemble~\cite{montanari2003nature,ferrari2012two,talagrand2000rigorous,gamarnik2023shatteringisingpurepspin,alaoui2024near}]\label{rem:ising}
    The ensemble of Ising $p$-spin Hamiltonians is given by the distribution of
    \begin{align}
        H_n = \sum_{1 \leq i_1 < \cdots < i_p \leq n} J_{i_1\cdots i_p} \sigma^Z_{i_1} \cdots \sigma^Z_{i_p}
    \end{align}
    for i.i.d. coefficients $J_{i_1\cdots i_p} \sim \cN(0, J^2 p! / 2N^{p-1})$ and $\Gamma > 0$.
    For all constant $p \geq 2$, this Hamiltonian ensemble has a glass transition.
\end{remark}

Breakthrough results in mathematics computed the free energy (and the correctness of the ansatz) rigorously for the Ising $p$-spin and other classical Hamiltonians several decades after the replica trick was first proposed in the physics literature \cite{guerra2002thermodynamic,guerra2003broken,talagrand2003spin,talagrand2006parisi,panchenko2013sherrington}. A rigorous proof has not verified the replica trick for quantum Hamiltonians\footnote{We make an exception here for the case of the transverse field Ising model which has been rigorously analyzed \cite{leschke2021existence,manai2025parisi}.}, but it is widely used in the physics literature and has been extensively tested numerically~\cite{sachdev1993gapless,chakrabarti1996quantum,maldacena2016remarks,charbonneau2023spin,kavokine2024exact}. These methods recover glass transitions in stoquastic ensembles such as the classical Ising $p$-spin ensemble with a weak transverse field, as well as in more quantum ensembles such as the quantum Heisenberg model, of which quantum max-cut is a special case.

\begin{remark}[Glass transition of the quantum Heisenberg Hamiltonian ensemble~\cite{kavokine2024exact}]
    The ensemble of random quantum Heisenberg Hamiltonians is defined as the distribution of
    \begin{align}
        H_n = \sum_{1 \leq i < j \leq n} J_{ij} \, (X_i X_j + Y_i Y_j + Z_i Z_j)
    \end{align}
    for i.i.d coefficients $J_{ij} \sim \cN(0, J^2)$. This Hamiltonian ensemble has a glass transition.
\end{remark}

There are many other variants of quantum spin glasses. For example, many quantum spin glasses are defined on bounded-degree graphs; these include both classical spin glasses on regular graphs or lattices placed in a transverse field~\cite{rieger1994zero,guo1994quantum,laumann2008cavity} and non-stoquastic models such as quantum gauge glasses~\cite{morita2005gauge} and the quantum XYZ model~\cite{itoi2024gauge} in finite dimensions. Some non-stoquastic quantum Hamiltonian ensembles are replica symmetric and non-glassy at any constant temperature. One such example is the SYK model~\cite{sachdev1993gapless,kitaev2015talks,maldacena2016remarks}, which replaces the Ising spins in \Cref{rem:ising} with Majorana fermions.

In \Cref{sec:plocal}, we will use the replica trick to show that another ensemble has a glass transition: the random Pauli $p$-local Hamiltonian ensemble. We show it is glassy at $p=3$ and that it is not glassy in the $p\to\infty$ limit.

Note that our definitions of RS/RSB/shattering do not properly capture Hamiltonians that are local in different operators, e.g., the spherical $p$-spin model or the SYK model. However, for many applications in quantum computing, and for the models discussed in this work, our definitions suffice. We also remark that in \Cref{sec:plocal}, we will further distinguish 1RSB, $k$RSB, and full RSB; we also relate the above definitions to the evaluation of the path integral in the replica trick and to TAP states.

Before we proceed with further proofs, we remark that the decomposition into clusters (e.g., in \Cref{def:introglass}) is robust to small errors in trace distance. This will be useful later.
\begin{lemma}\label{lem:shattering_robust}
    Let $\bm{\rho}$ be a $\left(\epsilon,q,M\right)$-decomposed state as given by \Cref{def:introglass}. Then any state $\bm{\tilde{\rho}}$ satisfying
    \begin{equation}
        \frac{1}{2}\left\lVert\bm{\rho}-\bm{\tilde{\rho}}\right\rVert_1\leq\delta
    \end{equation}
    for some $0\leq\delta\leq 1-\epsilon$ is $\left(\epsilon+\delta,q,M\right)$-decomposed.
\end{lemma}
\begin{proof}
    Let $\left\{\left(c_a,\bm{\rho}_a\right)\right\}_{a=1}^M$ be a clustering witnessing the $\left(\epsilon,q,M\right)$-decomposition of $\bm{\rho}$. We have by the triangle inequality that
    \begin{equation}
        \begin{aligned}
            \frac{1}{2}\left\lVert\bm{\tilde{\rho}}-\sum_{a=1}^M c_a \bm{\rho}_a\right\rVert_1&\leq\frac{1}{2}\left\lVert\bm{\tilde{\rho}}-\bm{\rho}\right\rVert_1+\frac{1}{2}\left\lVert\bm{\rho} -\sum_{a=1}^Mc_a \bm{\rho}_a\right\rVert_1 \leq\delta+\epsilon.
        \end{aligned}
    \end{equation}
    Therefore, $\left\{\left(c_a,\bm{\rho}_a\right)\right\}_{a=1}^M$ witnesses the $\left(\epsilon+\delta,q,M\right)$-decomposition of $\bm{\rho}$.
\end{proof}

\subsection{Quantum optimal transport implications of glassiness}\label{sec:transport}
Our main results for algorithmic obstructions rely on relating notions of glassiness to the quantum Wasserstein distance~\cite{9420734}, defined between two $n$-qubit quantum states $\bm \rho, \bm \sigma$ as:
\begin{align}\label{eq:wass}
    \norm{\bm \rho - \bm \sigma}_{W_1} &= \frac{1}{2} \min_{\{X_i\}} \left\{\sum_i \norm{\bm X_i}_1 \,:\, \bm \rho - \bm \sigma=  \sum_i \bm X_i, \, \bm X_i \in \cO_n, \, \Tr_i[\bm X_i] = 0 \, \forall \, i\right\},
\end{align}
where $\cO_n$ is the set of traceless Hermitian $n$-qubit operators. We will need the following properties of the quantum $W_1$ norm.

\begin{fact}[Tensorization, Proposition 4 of~\cite{9420734}]\label{fac:tensor}
    For any $(m+n)$-qubit states $\bm \rho, \bm \sigma$,
    \begin{align}
        \norm{\bm \rho-\bm\sigma}_{W_1} \geq \norm{\bm\rho_{1\dots m} - \bm\sigma_{1\dots m}}_{W_1} + \norm{\bm\rho_{m+1\dots m+n} - \bm\sigma_{m+1\dots m+n}}_{W_1},
    \end{align}
    where for $S \subseteq [n]$, $\bm\rho_S \coloneqq \Tr_{[n]\setminus S}(\bm \rho)$ denotes the marginal state on qubits in $S$.
\end{fact}

\begin{fact}[Discarding qubits, Proposition 5 of~\cite{9420734}]\label{fac:meas}
    Let $S \in [n]$ and let $X$ be a traceless Hermitian operator on $n$ qubits such that $\Tr_S X = 0$. Then
    \begin{align}
        \norm{X}_{W_1} \leq |S| \frac{3}{4} \norm{X}_1.
    \end{align}
\end{fact}

This allows us to immediately relate the Pauli bilinear form to $W_1$ via the following result.

\begin{lemma}[Wasserstein distance lower bound from overlap] \label{lem:was_dist_lower_bound}
    Given states $\bm{\rho}_a, \bm{\rho}_b \in \mathcal{S}_n$ on $n$ qubits, where
    \begin{equation}
        \begin{split}
            \text{(overlap) }& \quad \frac1n \langle \bm \rho_a , \bm \rho_b\rangle_\pauli = q_0 \\
            \text{(induced norm) }& \quad \frac1n \langle \bm \rho_c , \bm \rho_c\rangle_\pauli =  q_c  \text{ for } c \in \{a, b\},
        \end{split}
    \end{equation}
    then $\|\bm{\rho}_a - \bm{\rho}_b\|_{W_1} \geq \frac{n}{4}\left(q_a + q_b - 2q_0\right)$.
\end{lemma}
\begin{proof}
    Given a state $\bm{\rho} \in \mathcal{S}_n$, denote by $\bm{\rho}^i = \Tr_{[n]\setminus \{i\}}[\bm{\rho}] \in \mathcal{S}_1$ the reduced density matrix on the $i$-th qubit. Then by \Cref{fac:tensor},
    \begin{equation}
    \begin{split}
        \|\bm{\rho}_a - \bm{\rho}_b\|_{W_1} &\geq \sum_{i = 1}^n \|\bm{\rho}_a^i - \bm{\rho}_b^i\|_{W_1}      \\
        &= \frac12 \sum_{i = 1}^n \|\bm{\rho}_a^i - \bm{\rho}_b^i\|_{1}.  
    \end{split}
    \end{equation}
    In the second line above we use the fact that the Wasserstein distance equals the trace distance for one qubit. For one qubit states $\bm{\rho}, \bm{\rho}' \in \mathcal{S}_1$, the 1-Schatten norm explicitly equals: 
    \begin{equation}
        \| \bm{\rho} - \bm{\rho}'\|_1  =  \sqrt{
        \Big(\Tr[\bm{\rho}\,\sigma_x] - \Tr[\bm{\rho}'\,\sigma_x]\Big)^2
        + \Big(\Tr[\bm{\rho}\,\sigma_y] - \Tr[\bm{\rho}'\,\sigma_y]\Big)^2
        + \Big(\Tr[\bm{\rho}\,\sigma_z] - \Tr[\bm{\rho}'\,\sigma_z]\Big)^2
        }.
    \end{equation}
     Thus,
     \begin{equation}
    \begin{split}
        \|\bm{\rho}_a - \bm{\rho}_b\|_{W_1} &\geq   \sum_{i = 1}^n \frac12 \|\bm{\rho}_a^i - \bm{\rho}_b^i\|_{1} \\
        &\geq \sum_{i = 1}^n \frac14 \|\bm{\rho}_a^i - \bm{\rho}_b^i\|_{1}^2 \\
        &= \frac14 \sum_{i = 1}^n  \sum_{r \in \{x, y, z\}} \left( \Tr[\bm{\rho}_a^i \sigma_r] - \Tr[\bm{\rho}_b^i \sigma_r] \right)^2 \\
        &= \frac14 \sum_{i = 1}^n  \sum_{r \in \{x, y, z\}} \left( \Tr[\bm{\rho}_a \sigma_r^i] - \Tr[\bm{\rho}_b \sigma_r^i] \right)^2 \\
        &= \frac{n}{4}\left(q_a + q_b - 2q_0\right).
    \end{split}
    \end{equation}
\end{proof}

\begin{corollary}[Wasserstein distance from $\|\cdot\|_\pauli$] \label{cor:Was_from_Pauli}
    Given states $\bm \rho, \bm \rho'$ on $n$ qubits, it holds that
    \begin{equation}
        \|\bm \rho - \bm \rho'\|_{W_1} \geq \frac{1}{4} \|\bm \rho - \bm \rho'\|_\pauli^2.
    \end{equation}
\end{corollary}
\begin{proof}
    The proof follows by noting that $n(q_a + q_b - 2 q_0)$ in \Cref{lem:was_dist_lower_bound} is equal to $\|\bm \rho - \bm \rho'\|_\pauli^2$.
\end{proof}

\Cref{cor:Was_from_Pauli} only bounds the Wasserstein distance between clusters, not necessarily between the non-glassy and glassy Gibbs states. We show here that glassiness has a more operational implication in quantum optimal transport. Specifically, we prove a lower bound on the minimum Wasserstein complexity of a quantum channel that goes from a non-glassy to a glassy state.
We write $\operatorname{WC}$ to denote the \emph{Wasserstein complexity}~\cite{li2022wasserstein}, defined as
\begin{equation}
\operatorname{WC}\left(\bm{\varPhi}\right)=\sup_{\bm{\rho}}\left\lVert\bm{\rho}-\bm{\varPhi}\left(\bm{\rho}\right)\right\rVert_{W_1},
\end{equation}
where $\left\lVert\cdot\right\rVert_{W_1}$ is the quantum Wasserstein distance of \eqref{eq:wass}.

Not only does the Pauli bilinear form lower bound the Wasserstein distance, there is also a one local operator which can be used to quantify separation $\|\bm \rho_i - \bm \rho_j\|_{\rm Pauli}$ between two states $\bm \rho_i$ and $\bm \rho_j$. Differences in expectations $\Tr[X(\bm \rho_i-\bm \rho_j)]$ grow linearly with the distance $\|\bm \rho_i - \bm \rho_j\|_{\rm Pauli}^2$ as we show below.

\begin{lemma}\label{lem:observable_from_Pauli}
    Given states $\bm \rho_i, \bm \rho_j$ at distance $\|\bm \rho_i - \bm \rho_j\|_{\rm Pauli}^2$, there exists a one-local observable $X = \sum_i X_i$ made of observables $X_1, \dots, X_n$ acting on qubits $1, \dots, n$ respectively normalized as $\|X_i\|\leq 1$ such that $\Tr[X\bm \rho_i] - \Tr[X\bm \rho_j] \geq \frac{1}{2} \|\bm \rho_i - \bm \rho_j\|_{\rm Pauli}^2$.
\end{lemma}
\begin{proof}
Let $\bm \rho_i^r = \Tr_{[n]\setminus \{r\}}[\bm \rho_i]$ denote the marginal state of $\bm \rho_i$ on qubit $r$ (and respectively $\bm \rho_j^r$ for $\bm \rho_j$).
Denote the eigendecomposition of $\bm\rho_i^r - \bm\rho_j^r = e_r v_r v_r^\dagger - e_r w_r w_r^\dagger$ and define $X_r$ as the operator
    \begin{equation}
        X_r \coloneqq v_r v_r^\dagger - w_r w_r^\dagger,
    \end{equation}
    which has norm $\|X_r\|=1$.
    Note that
    \begin{equation}
        \begin{split}
            \sum_r \Tr[X_r (\bm \rho_i - \bm \rho_j)] & = \sum_r \|\bm \rho_i^r - \bm \rho_j^r\|_1 \\
            &= \sum_{r} \sqrt{ \sum_\mu \Tr[\sigma_r^\mu (\bm \rho_i^r - \bm \rho_j^r)]^2 } \\
            &\geq \frac{1}{2} \sum_{r,\mu} \Tr[\sigma_r^\mu (\bm \rho_i^r - \bm \rho_j^r)]^2 \\
            &= \frac{1}{2}\|\bm \rho_i - \bm \rho_j\|_{\rm Pauli}^2.
        \end{split}
    \end{equation}
    In the inequality above we use the fact that $\sqrt{x} \geq x/2$ for all $0 \leq x \leq 4$.
\end{proof}

We now show that low-Wasserstein complexity quantum channels preserves the decomposition of quantum states into clusters.
\begin{lemma}[Clustering is preserved]\label{lem:wc_preserves_shattering}
    Assume $\bm{\rho}$ is $\left(\epsilon,q,M\right)$-decomposable with clustering $\left\{\left(c_a,\bm{\rho}_a\right)\right\}_{a=1}^M$. Consider a state $\bm{\sigma}$ satisfying
    \begin{equation}
        \bm{\sigma}=\bm{\varPhi}\left(\bm{\rho}\right)
    \end{equation}
    for some quantum channel $\bm{\varPhi}$. If
    \begin{equation}
        \operatorname{WC}\left(\bm{\varPhi}\right)\leq\frac{q}{144}n,
    \end{equation}
    then $\bm{\sigma}$ is $\left(\epsilon,4q/9,M\right)$-decomposable.
\end{lemma}
\begin{proof}
    Let $\left\{\left(c_a,\bm{\rho}_a\right)\right\}_a$ be a clustering witnessing the $\left(\epsilon,q,M\right)$-decomposition of $\bm{\rho}$. Define $\bm{\sigma}_a:=\bm{\varPhi}\left(\bm{\rho}_a\right)$. We have by the triangle inequality and \Cref{cor:Was_from_Pauli} that for any $a\neq b$,
    \begin{equation}
        \begin{aligned}
            \left\lVert\bm{\sigma}_a-\bm{\sigma}_b\right\rVert_\pauli&\geq\left\lVert\bm{\rho}_a-\bm{\rho}_b\right\rVert_\pauli-\left\lVert\bm{\sigma}_a-\bm{\rho}_a\right\rVert_\pauli-\left\lVert\bm{\sigma}_b-\bm{\rho}_b\right\rVert_\pauli\\
            &\geq\left\lVert\bm{\rho}_a-\bm{\rho}_b\right\rVert_\pauli-2\sqrt{\left\lVert\bm{\sigma}_a-\bm{\rho}_a\right\rVert_{W_1}}-2\sqrt{\left\lVert\bm{\sigma}_b-\bm{\rho}_b\right\rVert_{W_1}}\\
            &\geq\left\lVert\bm{\rho}_a-\bm{\rho}_b\right\rVert_\pauli-4\sqrt{\operatorname{WC}\left(\bm{\varPhi}\right)}\\
            &\geq\frac{2\sqrt{q n}}{3}.
        \end{aligned}
    \end{equation}
    Define
    \begin{equation}
        \bm{\varDelta}:=\bm{\rho}-\sum_{a=1}^M c_a\bm{\rho}_i,
    \end{equation}
    which by the definition of $\left(\epsilon,q,M\right)$-decomposition has trace norm at most $2\epsilon$. As quantum channels are contractive in trace norm (even for nonpositive arguments)~\cite[Theorem~2.1]{10.1063/1.2218675},
    \begin{equation}
        \left\lVert\bm{\varPhi}\left(\bm{\varDelta}\right)\right\rVert_1\leq 2\epsilon.
    \end{equation}
    We therefore have that
    \begin{equation}
        \begin{aligned}
            &\sum_{a=1}^M c_a\bm{\sigma}_a=\sum_{a=1}^M c_a\bm{\varPhi}\left(\bm{\rho}_a\right)=\bm{\varPhi}\left(\sum_{a=1}^M c_a\bm{\rho}_a\right)=\bm{\varPhi}\left(\bm{\rho}-\bm{\varDelta}\right)=\bm{\sigma}-\bm{\varPhi}\left(\bm{\varDelta}\right)\\
            \implies&\frac{1}{2}\left\lVert\bm{\sigma}-\sum_{a=1}^M c_a\bm{\sigma}_a\right\rVert_1\leq\frac{1}{2}\left\lVert\bm{\varPhi}\left(\bm{\varDelta}\right)\right\rVert_1\leq\epsilon.
        \end{aligned}
    \end{equation}
    Finally, each $c_a=\exp[-\Theta(\log M)]$ since $\left\{\left(c_a,\bm{\rho}_a\right)\right\}_a$ witnessed the decomposition of $\bm{\rho}$.
    Therefore, $\bm{\sigma}$ is $\left(\epsilon,4q/9,M\right)$-decomposable.
\end{proof}

\Cref{lem:wc_preserves_shattering} taken in combination with \Cref{lem:shattering_robust} have the following immediate implication.

\begin{theorem}[Low-complexity channels cannot witness transition]\label{thm:no_trans_wit}
    Let $\bm{\rho}$ be a state which is glassy and $\left(\epsilon,q,M\right)$-decomposable, and $\bm{\sigma}$ which is not glassy. There exists no pair of states $\left(\bm{\tilde{\rho}},\bm{\tilde{\sigma}}\right)$ and quantum channel $\bm{\varPhi}$ with the following simultaneous properties:
    \begin{align}
        \left\lVert\bm{\rho}-\bm{\tilde{\rho}}\right\rVert_1&\leq o\left(1\right)\\
        \left\lVert\bm{\sigma}-\bm{\tilde{\sigma}}\right\rVert_1&\leq o\left(1\right)\\
        \bm{\tilde{\sigma}}&=\bm{\varPhi}\left(\bm{\tilde{\rho}}\right)\\
        \operatorname{WC}\left(\bm{\varPhi}\right)&\leq\frac{q}{144}n.
    \end{align}
\end{theorem}
\begin{proof}
    Assume such a pair of states $\left(\bm{\tilde{\rho}},\bm{\tilde{\sigma}}\right)$ and quantum channel $\bm{\varPhi}$ existed which satisfied these four properties.  Note that $\bm{\tilde{\rho}}$ is $\left(\epsilon+o\left(1\right),q,M\right)$-decomposable by \Cref{lem:shattering_robust}. In particular, by \Cref{lem:wc_preserves_shattering}, $\bm{\tilde{\sigma}}$ is $\left(\epsilon+o\left(1\right),4q/9,M\right)$-decomposable. Finally, again by \Cref{lem:shattering_robust} this implies that $\bm{\sigma}$ is glassy, yielding a contradiction.
\end{proof}

\Cref{thm:no_trans_wit} is the formal version of the result presented as \Cref{thm:wass} in the introduction; it shows the main contradiction required for an algorithmic obstruction.

\section{Algorithmic obstruction}
\label{sec:algs}

We now demonstrate that the presence of an RSB or shattering phase transition obstructs preparation of the Gibbs state for a class of algorithms we call \emph{temperature stable}. Our proof of obstruction relies on the relationship between glassiness and quantum optimal transport given in \Cref{sec:glasses}. We will then prove that several common algorithms are temperature stable.

\subsection{Glassiness obstructs stable quantum algorithms}

Roughly, we consider an algorithm temperature stable if the quantum channel with minimal quantum Wasserstein complexity that takes the algorithm's output at one temperature to its output at another temperature is Lipschitz. This is formally defined as follows.

\begin{definition}[$\left(\epsilon,f,L\right)$-Temperature Stability]\label{def:temp_stab}
    Let $\bm{\mathcal{A}}:\mathbb{R}^D\times\left[0,\infty\right)\to\mathcal{S}_n^{\text{m}}$ be a quantum algorithm which takes a problem instance $\bm{X}\in\mathbb{R}^D$ and an inverse temperature $\beta\in\left(0,\infty\right)$ and outputs a quantum state. The algorithm is said to be \emph{$\left(\epsilon,f,L,\leq\right)$-temperature stable} in a set $\mathcal{B}\subseteq\left(0,\infty\right)$ if, for any $\bm{X}\in\mathbb{R}^D$ and $\beta'\leq\beta\in\mathcal{B}$,
    \begin{equation}\label{eq:wc_comp_stab_def}
        \inf_{\substack{\bm{\varPhi}\in\mathcal{C}_n\\\text{s.t. }\norm{\bm{\varPhi}\left(\bm{\mathcal{A}}\left(\bm{X},\beta\right)\right)-\bm{\mathcal{A}}\left(\bm{X},\beta'\right)}_1 \leq \epsilon }}\operatorname{WC}\left(\bm{\varPhi}\right)\leq f+L\left\lvert\beta-\beta'\right\rvert.
    \end{equation}
    Similarly, we say the algorithm is \emph{$\left(\epsilon,f,L,\geq\right)$-temperature stable} in a set $\mathcal{B}\subseteq\left(0,\infty\right)$ if Eq.~\eqref{eq:wc_comp_stab_def} holds for any $\beta'\geq\beta\in\mathcal{B}$. Finally, we say that the algorithm is \emph{$\left(\epsilon,f,L\right)$-temperature stable} in a set $\mathcal{B}\subseteq\left(0,\infty\right)$ if it is both $\left(\epsilon,f,L,\leq\right)$- and $\left(\epsilon,f,L,\geq\right)$-temperature stable.
\end{definition}

We show that the glass transition obstructs algorithms that are temperature stable for appropriate choices of $\epsilon, f, L$. In the following subsections, we will give examples of algorithms satisfying these parameter choices.

\begin{theorem}[Temperature stable algorithms are obstructed by the glass transition]
    Let $(\bm \cA_n)_n$ be a sequence of $(\epsilon, f, L)$-temperature stable quantum algorithms in a set $\cB \subseteq (0, \infty)$. Let $(\cH_n)_n$ be a sequence of $n$-qubit Hamiltonian ensembles with a glass transition at temperature $\beta_*$ in neighborhood $\cB$. Let $\bm X$ denote the problem instance corresponding to a Hamiltonian $H \sim \cH_n$. If $\epsilon = o(1)$, $f = o(n)$ and $L = O(n)$, then there exists some $\beta \in \cB$ such that the thermal state $\bm \rho_\beta$ satisfies
    \begin{align}
        \norm{\cA(\bm X, \beta) - \bm \rho_\beta}_1 = \Omega(1)
    \end{align}
    with probability $1-o(1)$.
\end{theorem}
\begin{proof}
    Consider some $\beta,\beta'\in\mathcal{B}$ satisfying $\beta'<\beta_{\mathrm{d}}<\beta$, $\bm{\rho}_\beta$ is glassy, and:
    \begin{equation}
        |\beta - \beta'|<\frac{q}{144C}.
    \end{equation}
    By the stability of $\bm{\mathcal{A}}$, there exists a channel $\bm{\varPhi}$ of Wasserstein complexity $\frac{q}{144}n(1+o(1))$ satisfying
    \begin{equation}
        \norm{\bm{\mathcal{A}}\left(\bm{X},\beta'\right) - \bm{\varPhi}\left(\bm{\mathcal{A}}\left(\bm{X},\beta\right)\right)}_1 = o(1).
    \end{equation}
    By \Cref{thm:no_trans_wit}, it is therefore either the case that
    \begin{equation}
        \left\lVert\bm{\mathcal{A}}\left(\bm{X},\beta\right)-\bm{\rho}_\beta\right\rVert_1\geq\operatorname{\Omega}\left(1\right)
    \end{equation}
    or
    \begin{equation}
        \left\lVert\bm{\mathcal{A}}\left(\bm{X},\beta'\right)-\bm{\rho}_{\beta'}\right\rVert_1\geq\operatorname{\Omega}\left(1\right),
    \end{equation}
    proving the desired result.
\end{proof}

\subsection{Local continuous-time Lindbladian evolution}

We show that Lindbladian evolution with local Lindblad operators is temperature stable. More precisely, under the convention of each Lindblad operator having norm 1, we show that Lindbladian evolution is stable for any constant time. This constraint of constant time is shared by the analogous classical result for the stability of Langevin dynamics~\cite{10.1137/22M150263X,10.1002/cpa.22197}.

Lindbladian evolution is governed by a superoperator $\bm{\mathcal{L}}$ known as a \emph{Lindbladian}~\cite{lindblad1976generators}, with action:
\begin{equation}
    \bm{\mathcal{L}}\left[\bm{\rho}\right]=-\ci\left[\bm{H},\bm{\rho}\right]+\sum_i\gamma_i\left(\bm{A}_i\bm{\rho}\bm{A}_i^\dagger-\frac{1}{2}\left\{\bm{A}_i^\dagger\bm{A}_i,\bm{\rho}\right\}\right)
\end{equation}
where $\bm{H}$ is Hermitian and the $\gamma_i>0$. Often, the choice of \emph{jump operators} $\left\{\bm{A}_i\right\}$ is independent of temperature but may depend on Hamiltonian disorder; we consider the setting where there are $O(n)$ such operators and each is local. The only temperature dependence is in the coefficients $\gamma_i$; these can be thought of as generalizing the choice of ``acceptance probability'' in classical simulated annealing~\cite{chen2023efficientexactnoncommutativequantum}. For this reason, we write:
\begin{equation}\label{eq:thermal_lindblad}
    \bm{\mathcal{L}}_\beta\left[\bm{\rho}\right]=-\ci\left[\bm{H},\bm{\rho}\right]+\sum_i\gamma_i\left(\beta\right)\left(\bm{A}_i\bm{\rho}\bm{A}_i^\dagger-\frac{1}{2}\left\{\bm{A}_i^\dagger\bm{A}_i,\bm{\rho}\right\}\right).
\end{equation}
We will specifically be interested in the continuous time evolution:
\begin{equation}\label{eq:cont_time_lind}
    \bm{\mathcal{A}}\left(\beta;t\right):=\exp\left(t\bm{\mathcal{L}}_\beta\right)\left[\bm{\rho}_0\right],
\end{equation}
where $\bm{\rho}_0$ is some fixed initial state and the dependence of $\bm{\mathcal{A}}$ on $\bm{H}$ and the choice of $\bm{A}_i$ left implicit.

We show that $\bm{\mathcal{A}}\left(\beta;t\right)$ is temperature stable according to \Cref{def:temp_stab} over the set of $\beta\in\mathbb{Q}$.
\begin{theorem}[Lindblad dynamics are temperature-stable]
    Consider the quantum algorithm $\bm{\mathcal{A}}\left(\beta;t\right)$ as in Eq.~\eqref{eq:cont_time_lind}. Fix any choice of $0<\beta_{\mathrm{min}}<\beta_{\mathrm{max}}\in\mathbb{R}$, and assume the $\gamma_i$ are $\lambda_i$-Lipschitz on this domain, i.e., for all $i$ there exists $\lambda_i\geq 0$ such that for all $\beta,\beta'\in\left[\beta_{\mathrm{min}},\beta_{\mathrm{max}}\right]$:
    \begin{equation}
        \left\lvert\gamma_i\left(\beta\right)-\gamma_i\left(\beta'\right)\right\rvert\leq\lambda_i\left\lvert\beta-\beta'\right\rvert.
    \end{equation}
    Then, $\bm{\mathcal{A}}\left(\beta;t\right)$ is $\left(0,L\left(t\right)\right)$-temperature stable over $\beta\in\mathbb{Q}\cap\left[\beta_{\mathrm{min}},\beta_{\mathrm{max}}\right]$, where:
    \begin{equation}
        L\left(t\right)=\frac{3t}{2}\sum_i\lambda_i\left\lvert\operatorname{supp}\left(\bm A_i\right)\right\rvert\left\lVert\bm{A}_i\right\rVert_{\mathrm{op}}^2.
    \end{equation}
\end{theorem}
\begin{proof}
    We first write $\bm{\mathcal{A}}$ in a more convenient form. Fix $\beta^-<\beta^+$ where $\beta^\pm\in\mathbb{Q}$. Note that by Trotterization:
    \begin{equation}
        \exp\left(t\bm{\mathcal{L}}_{\beta^\pm}\right)=\lim_{m\to\infty}\exp\left(t\bm{\mathcal{L}}_0+t\sum_{\tau=1}^{\beta^\pm m}\left(\bm{\mathcal{L}}_{\frac{\tau}{m}}-\bm{\mathcal{L}}_{\frac{\tau-1}{m}}\right)\right)=\lim_{m\to\infty}\prod_{\tau=\beta^\pm m}^1\exp\left(t\bm{\mathcal{L}}_{\frac{\tau}{m}}-t\bm{\mathcal{L}}_{\frac{\tau-1}{m}}\right)\exp\left(\frac{t}{m}\bm{\mathcal{L}}_0\bm{1}\left\{\tau\leq m\right\}\right),
    \end{equation}
    where the limit is taken in a way such that both $\beta^+ m\in\mathbb{N}$ and $\beta^- m\in\mathbb{N}$. We can then define the algorithm:
    \begin{equation}
        \bm{\mathcal{A}}\left(\beta;t,m\right):=\left(\prod_{\tau=\beta m}^1\exp\left(t\bm{\mathcal{L}}_{\frac{\tau}{m}}-t\bm{\mathcal{L}}_{\frac{\tau-1}{m}}\right)\exp\left(\frac{t}{m}\bm{\mathcal{L}}_0\bm{1}\left\{\tau\leq m\right\}\right)\right)\left[\bm{\rho}_0\right],
    \end{equation}
    such that:
    \begin{equation}
        \bm{\mathcal{A}}\left(\beta;t\right)=\lim_{m\to\infty}\bm{\mathcal{A}}\left(\beta;t,m\right).
    \end{equation}
    
    As the Wasserstein complexity is convex~\cite{li2022wasserstein}, $\left(f,L\right)$-temperature stability is robust to errors in $\bm{\mathcal{A}}$. For this reason, to demonstrate the stability of $\bm{\mathcal{A}}\left(\beta;t\right)$ over rational $\beta$ it suffices to show that $\bm{\mathcal{A}}\left(\beta;t,m\right)$ is stable in $\left\{\beta^\pm\right\}$ for all sufficiently large $m$. Noting that:
    \begin{equation}
        \bm{\mathcal{A}}\left(\beta^-;t,m\right)=\bm{\varPhi}_t\left(\bm{\mathcal{A}}\left(\beta^+;t,m\right)\right),
    \end{equation}
    where
    \begin{equation}
        \bm{\varPhi}_t=\prod_{\tau=\beta^- m}^{\beta^+ m}\exp\left(t\bm{\mathcal{L}}_{\frac{\tau-1}{m}}-t\bm{\mathcal{L}}_{\frac{\tau}{m}}\right),
    \end{equation}
    we therefore need only bound:
    \begin{equation}
        \lim_{m\to\infty}\operatorname{WC}\left(\bm{\varPhi}_t\right)=\lim_{m\to\infty}\sup_{\bm{\rho}}\left\lVert\bm{\rho}-\bm{\varPhi}_t\left(\bm{\rho}\right)\right\rVert_{W_1}
    \end{equation}
    to demonstrate $\left(f,L,\geq\right)$-temperature stability. Similarly, as:
    \begin{equation}
        \bm{\mathcal{A}}\left(\beta^+;t,m\right)=\bm{\varPhi}_{-t}\left(\bm{\mathcal{A}}\left(\beta^-;t,m\right)\right),
    \end{equation}
    we therefore need only bound:
    \begin{equation}
        \lim_{m\to\infty}\operatorname{WC}\left(\bm{\varPhi}_{-t}\right)=\lim_{m\to\infty}\sup_{\bm{\rho}}\left\lVert\bm{\rho}-\bm{\varPhi}_{-t}\left(\bm{\rho}\right)\right\rVert_{W_1}
    \end{equation}
    to demonstrate $\left(f,L,\leq\right)$-temperature stability.
    
    By the subadditivity of the Wasserstein complexity under concatenation~\cite{li2022wasserstein},
    \begin{equation}
        \lim_{m\to\infty}\operatorname{WC}\left(\bm{\varPhi}_{\pm t}\right)\leq\lim_{m\to\infty}\sum_{\tau=\beta^- m}^{\beta^+ m}\operatorname{WC}\left(\exp\left(\pm t\bm{\mathcal{L}}_{\frac{\tau-1}{m}}\mp t\bm{\mathcal{L}}_{\frac{\tau}{m}}\right)\right).
    \end{equation}
    As only the linear term of the exponential is non-vanishing in this limit,
    \begin{equation}
        \begin{aligned}
            \lim_{m\to\infty}\operatorname{WC}\left(\bm{\varPhi}_{\pm t}\right)&\leq t\lim_{m\to\infty}\sum_{\tau=\beta^- m}^{\beta^+ m}\sup_{\bm{\rho}}\left\lVert\left(\bm{\mathcal{L}}_{\frac{\tau}{m}}-\bm{\mathcal{L}}_{\frac{\tau-1}{m}}\right)\left[\bm{\rho}\right]\right\rVert_{W_1}\\
            &\leq t\left(\beta^+-\beta^-\right)\lim_{m\to\infty}m\sup_{\beta\in\left[\beta^-,\beta^+\right]}\sup_{\bm{\rho}}\left\lVert\left(\bm{\mathcal{L}}_\beta-\bm{\mathcal{L}}_{\beta-m^{-1}}\right)\left[\bm{\rho}\right]\right\rVert_{W_1}.
        \end{aligned}
    \end{equation}
    
    We now bound the Wasserstein norm of this final term. We have:
    \begin{equation}
        \left(\bm{\mathcal{L}}_\beta-\bm{\mathcal{L}}_{\beta-m^{-1}}\right)\left[\bm{\rho}\right]=\sum_i\left(\gamma_i\left(\beta\right)-\gamma_i\left(\beta-m^{-1}\right)\right)\left(\bm{A}_i\bm{\rho}\bm{A}_i^\dagger-\frac{1}{2}\left\{\bm{A}_i^\dagger\bm{A}_i,\bm{\rho}\right\}\right).
    \end{equation}
    Now, for each $i$ define $S_i:=\operatorname{supp}\left(\bm{A}_i\right)$, i.e., the set of qubits on which $\bm{A}_i$ acts nontrivially. From the cyclicity of the trace:
    \begin{equation}
        \Tr_{S_i}\left(\bm{A}_i\bm{\rho}\bm{A}_i^\dagger-\frac{1}{2}\left\{\bm{A}_i^\dagger\bm{A}_i,\bm{\rho}\right\}\right)=\bm{0}.
    \end{equation}
    In particular, by \Cref{fac:meas},
    \begin{equation}
        \left\lVert\bm{A}_i\bm{\rho}\bm{A}_i^\dagger-\frac{1}{2}\left\{\bm{A}_i^\dagger\bm{A}_i,\bm{\rho}\right\}\right\rVert_{W_1}\leq\frac{3}{2}\left\lvert S_i\right\rvert\left\lVert\bm{A}_i\right\rVert_{\mathrm{op}}^2.
    \end{equation}
    Therefore, by the triangle inequality,
    \begin{equation}
        \begin{aligned}
            \lim_{m\to\infty}\operatorname{WC}\left(\bm{\varPhi}_{\pm}\right)&\leq\frac{3t}{2}\left(\beta^+-\beta^-\right)\lim_{m\to\infty}m\sup_{\beta\in\left[\beta^-,\beta^+\right]}\sum_i\left\lvert\gamma_i\left(\beta\right)-\gamma_i\left(\beta-m^{-1}\right)\right\rvert\left\lvert S_i\right\rvert\left\lVert\bm{A}_i\right\rVert_{\mathrm{op}}^2\\
            &\leq\frac{3t}{2}\left(\beta^+-\beta^-\right)\sum_i\lambda_i\left\lvert S_i\right\rvert\left\lVert\bm{A}_i\right\rVert_{\mathrm{op}}^2.
        \end{aligned}
    \end{equation}
    Taken together, this proves the desired result.
\end{proof}

\subsection{Shallow quantum algorithms}

We can generalize the Lindbladian result to parameterized shallow circuits composed of such Lindbladians. Namely, we write:
\begin{equation}\label{eq:param_shallow_circ}
    \bm{\mathcal{A}}\left(\bm{X},\beta;\bm{\theta}\right):=\bigcirc_{\ell=p}^1\left(\mathcal{U}_\ell\left(\bm{X},\bm{\theta}_\ell\right)\circ\exp\left(\bm{\mathcal{L}}_{\ell,\beta}\right)\right)\left[\bm{\rho}_0\right],
\end{equation}
where $\bm{\rho}_0$ is some fixed initial state, $\mathcal{U}_\ell\left(\bm{X},\bm{\theta}\right)$ some parameterized family of unitary channels which potentially depend on the disorder $\bm{X}$, and
\begin{equation}
    \bm{\mathcal{L}}_{\ell,\beta}\left[\bm{\rho}\right]=\sum_i\gamma_{\ell,i}\left(\beta\right)\left(\bm{A}_i\bm{\rho}\bm{A}_i^\dagger-\frac{1}{2}\left\{\bm{A}_i^\dagger\bm{A}_i,\bm{\rho}\right\}\right).
\end{equation}
In what follows, we will assume that $\mathcal{U}\left(\bm{X},\bm{\theta}_\ell\right)$ can be implemented as a depth-$K$ unitary composed of $d$-local gates acting on a sparse hypergraph of degree $\mathfrak{d}$. This construction includes so-called ``hardware efficient'' variational algorithms~\cite{kandala2017hardware} on quantum computers with topology given by a sparse graph. We also assume for simplicity and concreteness that the $\bm{A}_i$ are $1$-local and mutually commuting, e.g., the Lindbladian evolution is generated by coupling with a bath using an interaction of the form $\sum_i\bm{O}_i\otimes\bm{B}_i$ with the $\bm{O}_i$ mutually commuting. Note that our impossibility result is independent of how the parameters $\bm{\theta}$ are set---in particular, our results in some ways generalize the known impossibility of efficiently training shallow hardware efficient variational quantum algorithms using gradient descent~\cite{anschuetz2022critical,anschuetzkiani2022,anschuetz2025unified}.

Before proving our stability result, we first cite a known result bounding the superoperator norm $\left\lVert\bm{\varLambda}\right\rVert_{W_1\to W_1}$ induced by the Wasserstein norm via the light cone of $\bm{\varLambda}$~\cite{9420734}.
\begin{proposition}[Superoperator norm light cone bound,~{\cite[Proposition~13]{9420734}}]\label{prop:w1_supop_bound}
    Let $\bm{\varLambda}$ a quantum channel on $n$-qubit states. Define subsets $\mathcal{I}_i\subseteq\left[n\right]$ for $i\in\left[n\right]$ with the following property: if $\Tr_{\left\{i\right\}}\left(\bm{X}\right)=\bm{0}$,
    \begin{equation}
        \Tr_{\mathcal{I}_i}\left(\bm{\varLambda}\left(\bm{X}\right)\right)=\bm{0}.
    \end{equation}
    Then,
    \begin{equation}
        \left\lVert\bm{\varLambda}\right\rVert_{W_1\to W_1}\leq\frac{3}{2}\max_{i\in\left[n\right]}\left\lvert\mathcal{I}_i\right\rvert.
    \end{equation}
\end{proposition}
We now prove our stability result. To save on notation, we will often refer to the Wasserstein complexity or superoperator norm of a unitary operator; this should be taken to mean the associated measure with respect to the unitary channel induced by the unitary operator.
\begin{theorem}[Shallow quantum algorithms are temperature-stable]
    Consider the quantum algorithm $\bm{\mathcal{A}}_{\bm{\theta}}\left(\bm{X},\beta\right)$ as in Eq.~\eqref{eq:param_shallow_circ} and fix $\beta^-<\beta^+$. Assume for all $\bm{X}$ and $\bm{\alpha}$ that $\mathcal{U}_\ell\left(\bm{X},\bm{\theta}_\ell\right)$ is a depth-$K$ unitary composed of $d$-local gates acting on a sparse hypergraph of degree $\mathfrak{d}$. Assume further that the $\bm{\gamma}_\ell$ are $\lambda$-Lipschitz, i.e., for all $\ell\in\left[p\right]$ and $\beta,\beta'\in\left(0,\infty\right)$:
    \begin{equation}
        \left\lVert\bm{\gamma}_\ell\left(\beta\right)-\bm{\gamma}_\ell\left(\beta'\right)\right\rVert_1\leq\lambda\left\lvert\beta-\beta'\right\rvert.
    \end{equation}
    Finally, assume that the jump operators $\left\{\bm{A}_i\right\}_i$ are all $1$-local, mutually commuting, and have operator norm at most $1$. Then $\bm{\mathcal{A}}_{\bm{\theta}}\left(\bm{X},\beta\right)$ is $\left(0,L\right)$-stable on the domain $\beta\in\left[\beta^-,\beta^+\right]$, where:
    \begin{equation}
        L=6^p\left(d\mathfrak{d}\right)^{Kp}\lambda+\operatorname{O}\left(\left\lvert\beta^+-\beta^-\right\rvert\right).
    \end{equation}
\end{theorem}
\begin{proof}
    We fix the disorder $\bm{X}$ and parameters $\bm{\theta}$. We define:
    \begin{equation}
        \bm{\varPhi}_{\beta,\ell}:=\mathcal{U}_\ell\left(\bm{X},\bm{\theta}_\ell\right)\circ\exp\left(\bm{\mathcal{L}}_{\ell,\beta}\right).
    \end{equation}
    Note this has inverse:
    \begin{equation}
        \bm{\varPhi}_{\beta,\ell}^{-1}=\exp\left(-t_\ell\bm{\mathcal{L}}_{\ell,\beta}\right)\circ\mathcal{U}_\ell\left(\bm{X},\bm{\theta}_\ell\right)^\dagger.
    \end{equation}
    We are interested in calculating the Wasserstein complexity of the unitary channel:
    \begin{equation}
        \bm{\varLambda}_{\beta,\beta',p}:=\left(\bigcirc_{\ell=p}^1\bm{\varPhi}_{\beta',\ell}\right)\circ\left(\bigcirc_{\ell=1}^p \bm{\varPhi}_{\beta,\ell}^{-1}\right);
    \end{equation}
    as:
    \begin{equation}
        \bm{\mathcal{A}}\left(\bm{X},\beta';\bm{\theta}\right)=\bm{\varLambda}_{\beta,\beta',p}\left[\bm{\mathcal{A}}\left(\bm{X},\beta;\bm{\theta}\right)\right],
    \end{equation}
    upper bounding the Wasserstein complexity of $\bm{\varLambda}_{\beta,\beta',p}$ upper-bounds the stability parameter of the $\bm{\mathcal{A}}$.

    We bound the Wasserstein complexity of $\bm{\varLambda}_{\beta,\beta',p}$ inductively in $p$. First, note that trivially:
    \begin{equation}
        \operatorname{WC}\left(\bm{\varLambda}_{\beta,\beta',0}\right)=0.
    \end{equation}
    Now, assume the inductive hypothesis:
    \begin{equation}
        \operatorname{WC}\left(\bm{\varLambda}_{\beta,\beta',p-1}\right)\leq 6^{p-1}\left(d\mathfrak{d}\right)^{K\left(p-1\right)}\lambda\left\lvert\beta-\beta'\right\rvert+\operatorname{O}\left(\left\lvert\beta-\beta'\right\rvert^2\right).
    \end{equation}
    For any state $\bm{\rho}$, defining:
    \begin{equation}
        \bm{\sigma}:=\bm{\varPhi}_{\beta',p}^{-1}\left[\bm{\rho}\right],
    \end{equation}
    we have:
    \begin{equation}
        \begin{aligned}
            \left\lVert\bm{\rho}-\bm{\varLambda}_{\beta,\beta',p}\left[\bm{\rho}\right]\right\rVert_{W_1}&=\left\lVert\bm{\varPhi}_{\beta',p}\left[\bm{\sigma}\right]-\bm{\varLambda}_{\beta,\beta',p}\left[\bm{\varPhi}_{\beta',p}\left[\bm{\sigma}\right]\right]\right\rVert_{W_1}\\
            &\leq\left\lVert\bm{\varPhi}_{\beta',p}^{-1}\right\rVert_{W_1\to W_1}\left\lVert\bm{\sigma}-\left(\bm{\varLambda}_{\beta,\beta',p}\circ\bm{\varPhi}_{\beta,p}\circ\bm{\varPhi}_{\beta',p}^{-1}\right)\left[\bm{\sigma}\right]\right\rVert_{W_1}\\
            &\leq\left\lVert\bm{\varPhi}_{\beta',p}^{-1}\right\rVert_{W_1\to W_1}\operatorname{WC}\left(\bm{\varLambda}_{\beta,\beta',p}\circ\bm{\varPhi}_{\beta,p}\circ\bm{\varPhi}_{\beta',p}^{-1}\right)\\
            &\leq\left\lVert\bm{\varPhi}_{\beta',p}^{-1}\right\rVert_{W_1\to W_1}\left(\operatorname{WC}\left(\bm{\varLambda}_{\beta,\beta',p}\right)+\operatorname{WC}\left(\bm{\varPhi}_{\beta,p}\circ\bm{\varPhi}_{\beta',p}^{-1}\right)\right)\\
            &\leq\left\lVert\bm{\varPhi}_{\beta',p}^{-1}\right\rVert_{W_1\to W_1}\left(6^{p-1}\left(d\mathfrak{d}\right)^{K\left(p-1\right)}\lambda\left\lvert\beta-\beta'\right\rvert+\operatorname{WC}\left(\bm{\varPhi}_{\beta,p}\circ\bm{\varPhi}_{\beta',p}^{-1}\right)\right).
        \end{aligned}
    \end{equation}
    Now, since the disorder and parameters are fixed note that:
    \begin{equation}
        \begin{aligned}
            \bm{\varPhi}_{\beta,p}^{-1}\bm{\varPhi}_{\beta',p}&=\exp\left(-\bm{\mathcal{L}}_{p,\beta}\right)\exp\left(\bm{\mathcal{L}}_{p,\beta'}\right)\\
            &=\exp\left(\sum_i\left(\gamma_{\ell,i}\left(\beta'\right)-\gamma_{\ell,i}\left(\beta\right)\right)\left(\bm{A}_i\bm{\rho}\bm{A}_i^\dagger-\frac{1}{2}\left\{\bm{A}_i^\dagger\bm{A}_i,\bm{\rho}\right\}\right)\right)\\
            &=\bigcirc_i\exp\left(\left(\gamma_{\ell,i}\left(\beta'\right)-\gamma_{\ell,i}\left(\beta\right)\right)\left(\bm{A}_i\bm{\rho}\bm{A}_i^\dagger-\frac{1}{2}\left\{\bm{A}_i^\dagger\bm{A}_i,\bm{\rho}\right\}\right)\right),
        \end{aligned}
    \end{equation}
    with the final two lines following as the $\bm{A}_i$ are assumed to be mutually commuting. By the operator norm bound on $\bm{A}_i$, each channel in this composition can only change the trace distance of a state it acts on by at most:
    \begin{equation}
        \left\lVert\left(\gamma_{\ell,i}\left(\beta'\right)-\gamma_{\ell,i}\left(\beta\right)\right)\left(\bm{A}_i\bm{\rho}\bm{A}_i^\dagger-\frac{1}{2}\left\{\bm{A}_i^\dagger\bm{A}_i,\bm{\rho}\right\}\right)\right\rVert_{1\to 1}\leq 2\left\lvert\gamma_{\ell,i}\left(\beta'\right)-\gamma_{\ell,i}\left(\beta\right)\right\rvert
    \end{equation}
    up to higher-order terms in $\left\lvert\gamma_{\ell,i}\left(\beta'\right)-\gamma_{\ell,i}\left(\beta\right)\right\rvert$. As the trace distance equals the Wasserstein distance on a single qubit, we bound the Wasserstein complexity as:
    \begin{equation}
        \begin{aligned}
            \operatorname{WC}\left(\bm{\varPhi}_{\beta,p}^{-1}\bm{\varPhi}_{\beta',p}\right)&\leq 2\sum_i\left\lvert\gamma_i\left(\beta'\right)-\gamma_i\left(\beta\right)\right\rvert+\operatorname{O}\left(\sum_i\left\lvert\gamma_i\left(\beta'\right)-\gamma_i\left(\beta\right)\right\rvert^2\right)\\
            &\leq 2\lambda\left\lvert\beta-\beta'\right\rvert+\operatorname{O}\left(\left\lvert\beta-\beta'\right\rvert^2\right).
        \end{aligned}
    \end{equation}
    Similarly, as $\bm{\mathcal{L}}_{\beta'}$ is $1$-local and $\mathcal{U}_\ell\left(\bm{X},\bm{\theta}_\ell\right)$ has a maximal light cone cardinality of $\left(d\mathfrak{d}\right)^K$, by \Cref{prop:w1_supop_bound} we have:
    \begin{equation}
        \left\lVert\bm{\varPhi}_{\beta',p}^{-1}\right\rVert_{W_1\to W_1}\leq\frac{3}{2}\left(d\mathfrak{d}\right)^K.
    \end{equation}
    We therefore have that:
    \begin{equation}
        \begin{aligned}
            \left\lVert\bm{\rho}-\bm{\varLambda}_{\beta,\beta',p}\left[\bm{\rho}\right]\right\rVert_{W_1}&\leq\frac{3}{2}\left(d\mathfrak{d}\right)^K\left(6^{p-1}\left(d\mathfrak{d}\right)^{K\left(p-1\right)}\lambda\left\lvert\beta-\beta'\right\rvert+2\lambda\left\lvert\beta-\beta'\right\rvert+\operatorname{O}\left(\left\lvert\beta-\beta'\right\rvert^2\right)\right)\\
            &\leq\frac{3}{2}\left(d\mathfrak{d}\right)^K\left(3\times 6^{p-1}\left(d\mathfrak{d}\right)^{K\left(p-1\right)}\left\lVert\bm{\theta}\right\rVert_\infty\lambda\left\lvert\beta-\beta'\right\rvert+\operatorname{O}\left(\left\lvert\beta-\beta'\right\rvert^2\right)\right)\\
            &\leq 6^p\left(d\mathfrak{d}\right)^{Kp}\lambda\left\lvert\beta-\beta'\right\rvert+\operatorname{O}\left(\left\lvert\beta-\beta'\right\rvert^2\right)
        \end{aligned}
    \end{equation}
    As $\bm{\rho}$ was arbitrary,
    \begin{equation}
        \operatorname{WC}\left(\bm{\varLambda}_{\beta,\beta',p}\right)\leq 6^p\left(d\mathfrak{d}\right)^{Kp}\lambda\left\lvert\beta-\beta'\right\rvert+\operatorname{O}\left(\left\lvert\beta-\beta'\right\rvert^2\right).
    \end{equation}
\end{proof}

\section{Glassiness of random \texorpdfstring{$p$}\ -local Hamiltonians}
\label{sec:plocal}

In this section, we show that the random $p$-local Hamiltonian ensemble has a glass transition. We recall that the ensemble is defined over all $p$-local Pauli strings with Gaussian coefficients, i.e.,
\begin{align}
    H_{n,p} = \sum_{P \in \cP_{n,p}} J_P P, \qquad J_P \sim_\mathrm{iid} \cN\lr{0, \frac{J^2(p-1)!}{n^{p-1}}}.
\end{align}
To show a glass transition, we will use the replica trick to compute the quenched free energy: for partition function $Z_\beta = \Tr[e^{-\beta H}]$, we compute
\begin{align}
    F^{\mathrm{quenched}}_\beta(\bm \rho_\beta) = -\frac{1}{\beta} \EE{\log Z_\beta}.
\end{align}
Since the Gibbs state $\bm \rho_\beta = e^{-\beta H} / Z_\beta$ is the state that minimizes the free energy
\begin{align}
    F_\beta(\bm \rho) = \Tr[H\bm \rho] - \frac{1}{\beta} S(\bm \rho),
\end{align}
one can determine the structural properties of the Gibbs state (i.e., RS or shattering or 1RSB or full RSB) that hold in expectation by minimizing the quenched free energy. As is standard in the physics literature, we use the quenched computation to claim that the glassiness property also holds with high probability over individual Hamiltonians drawn from the ensemble.

We attempt below to give a somewhat self-contained description of the replica method, and we provide additional details in the appendices. In \Cref{sec:path}, we review the path integral over spin coherent states used throughout our calculations; the derivation is known in the literature and is similarly presented in~\cite{swingle2023bosonicmodelquantumholography} and described in more detail in \Cref{app:path}, but it is necessary to describe our work. In \Cref{sec:replica}, we review the replica trick at a high level, focusing on its application in quantum spin glasses. In \Cref{sec:saddles}, we report the RS, 1RSB, and full RSB saddle point equations of the $p$-local Pauli ensemble. This is our first technical contribution of the section, and we provide derivations in \Cref{app:1rsb} and \Cref{app:frsb} for these results. In \Cref{sec:p3}, we show that the 3-local Pauli ensemble has a glass phase by numerically solving the replica saddle point equations, and we provide an analytical estimate of its temperature based on an expansion similar to~\cite{bray1980replica}. In \Cref{sec:1rsb}, we show that the $p$-local ensemble for sufficiently large $p$ has no glass phase via an analysis of the 1RSB saddle point equations, proving that no nontrivial solution exists via a sharpened log-Sobolev inequality.

\subsection{Path integral}\label{sec:path}
Before describing the replica trick, we comment on how one computes the simpler quantity of $\E Z_\beta$. This quantity suffices to evaluate the \emph{annealed} free energy,
\begin{align}
    F^{\mathrm{annealed}}_\beta(\bm \rho_\beta) = -\frac{1}{\beta} \log \EE{Z_\beta}.
\end{align}
In general, the annealed free energy may be different from the quenched free energy. Classically, this indicates the presence of a glass phase~\cite{mezard1987spin}; in the quantum case, however, a model may be non-glassy and yet have different quenched and annealed free energies~\cite{baldwin2020quenched}. The method we use to compute $\E Z_\beta$ also resembles the method used to compute $\E Z_\beta^s$ for any integer $s \geq 1$, which is required in the replica trick (as we shall see below).

In the classical case, one evaluates $\E Z_\beta$ by summing over eigenstates. For stoquastic quantum Hamiltonians, it can be evaluated by quantum Monte Carlo (or the Feynman-Kac path integral), which evaluates a probabilistic process of bitflips over classical bitstrings. We are interested here in non-stoquastic quantum Hamiltonians, which prevent a probabilistic analysis due to the sign problem. Hence, we use the saddle point method over a Feynman path integral. This is detailed in a self-contained manner in \Cref{app:path}, but we offer a brief treatment here. The annealed free energy of this model has already been studied by these techniques in~\cite{swingle2023bosonicmodelquantumholography}, so we refrain from evaluating the saddle point equations and discussing the annealed free energy here.

Since the Hamiltonians we study are expressed over Pauli strings, we choose the overcomplete basis of spin coherent states to compute $Z = \Tr[e^{-\beta H}]$ (and its replicated version) as a path integral. For $\theta \in [0, \pi)$ and $\phi \in [0,2\pi)$, a \emph{spin coherent state} is defined as
\begin{align}
    \ket{\Omega} = \begin{pmatrix}\cos \frac{\theta}{2} \\ e^{i\phi}\sin \frac{\theta}{2}\end{pmatrix}.
\end{align}
These have the convenient property that for any $\mu$, the Pauli matrix can be written as
\begin{align}
    \bm \sigma_\mu = \frac{1}{2\pi}\int d\Omega\, \ketbra{\Omega} s_\mu
\end{align}
for spin coherent paths
\begin{align}
    s_X(\theta,\phi) = 3\sin\theta\cos\phi, \quad s_Y(\theta,\phi) = 3\sin\theta\sin\phi, \quad s_Z(\theta,\phi) = 3\cos\theta.
\end{align}
We will write write all path integrals in terms of spin coherent states with $\ket{\Omega_r^{(i)}}$, where $r \in [N]$ labels a qubit and $i \in [\beta/\Delta\tau]$ labels an imaginary time slice. In the continuum limit of $\Delta\tau\to0$, we denote the path of Pauli $\mu$ acting at time $\tau$ by $s_\mu(\Omega(\tau))$. We refer to \Cref{app:path} for a somewhat pedagogical introduction to computing path integrals over spin coherent states. Briefly, we rewrite $Z_\beta$ as
\begin{align}
    Z_\beta = \Tr e^{-\beta \sum_I J_I \sigma_{I_1} \cdots \sigma_{I_p}} = \int \cD \Omega \,\exp\left[-\int_0^\beta d\tau \sum_I J_I s_{I_1}(\tau)\cdots s_{I_p}(\tau)\right],
\end{align}
integrating over paths via $\cD \Omega$, where $I_j$ contains both the qubit index and Pauli type. Since each $J_I$ is drawn independently from a Gaussian distribution, we apply the Gaussian MGF $\EE{e^{tR} = e^{t^2/2}}$ for $R \sim \cN(0,1)$. This yields a path integral over two times, $\tau_1$ and $\tau_2$:
\begin{align}
    \E Z_\beta = \int \cD^n \Omega \exp\left[\frac{J^2}{2p N^{p-1}} \int_0^\beta \int_0^\beta d\tau_1 \,d\tau_2 \lr{\sum_{r=1}^n \sum_{\mu\in\{x,y,z\}} s_{r\mu}^a(\tau_1) s_{r\mu}^b(\tau_2)}^p\right].
\end{align}
To evaluate this integral, we introduce variables $G, \Sigma$ by inserting a Dirac delta function
\begin{align}
    \delta\lr{G(\tau_1, \tau_2) - \frac{1}{n}\sum_{r,\mu} s_{r\mu}(\tau_1) s_{r\mu}(\tau_2)}
\end{align}
and enforcing it via a Lagrange multiplier $\Sigma(\tau_1,\tau_2)$. Integrating over all functions $G(\tau_1,\tau_2), \Sigma(\tau_1,\tau_2)$ can be done rigorously by identifying a kernel and showing that its spectrum has a finite number of eigenvalues~\cite{kitaev2015talks,maldacena2016remarks}. Here, we will follow the physics convention of leaving the measure over $G, \Sigma$ unspecified, obtaining the path integral
\begin{align}
    \E Z_\beta &= \int \cD^n \Omega\, \cD G\, \cD\Sigma\, \exp\Bigg\{\int_0^\beta \int_0^\beta d\tau_1 \,d\tau_2 \Bigg(\frac{J^2 n}{2p} G^p(\tau_1, \tau_2) \nonumber \\
    &\qquad \qquad \qquad \qquad \qquad \quad - \frac{n}{2} \Sigma(\tau_1, \tau_2) \left[G(\tau_1, \tau_2) - \frac{1}{n} \sum_{r,\mu} s_{r\mu}(\tau_1) s_{r\mu}(\tau_2)\right]\Bigg)\Bigg\}.
\end{align}
The final step is to apply the Hubbard-Stratonovich transformation
\begin{align}
    \exp[\frac{1}{2}x^T M x] = \lr{1/\sqrt{(2\pi)^n\det(M)}}\int d^n h \exp[-\frac{1}{2}h^T M^{-1} h + h^T x]
\end{align}
for Hermitian PSD matrix $M$. Letting $\cD^n h$ denote the Gaussian measure relating $h_{r\mu}(\tau_1), h_{s\nu}(\tau_2)$ by covariance $\delta_{rs}\delta_{\mu\nu}\Sigma(\tau_1,\tau_2)$, we find that
\begin{align}
    \EE{Z} &= \int \cD^n \Omega\, \cD G\, \cD\Sigma\, \cD^n h\, \exp\Bigg\{\int_0^\beta \int_0^\beta d\tau_1 \,d\tau_2 \Bigg(\frac{J^2 n}{2p} G^p(\tau_1, \tau_2) - \frac{n}{2}\Sigma(\tau_1, \tau_2) G(\tau_1, \tau_2)\Bigg) \nonumber \\
    &\qquad \qquad \qquad \qquad \qquad \qquad \quad  + \sum_{r,\mu} \int_0^\beta d\tau\, h_{r\mu}(\tau) s_{r\mu}(\tau)\Bigg\}.
\end{align}
This integral can now be evaluated in the $n\to\infty$ limit by the saddle point approximation, where we set derivatives with respect to $G$ and $\Sigma$ to zero. This yields saddle point equations
\begin{align}\label{eq:saddleannealed}
    J^2 G(\tau_1,\tau_2)^{p-1} = \Sigma(\tau_1, \tau_2), \qquad G(\tau_1, \tau_2) = \left\langle \cT \sum_\mu \sigma^\mu(\tau_1)\sigma^\mu(\tau_2) \right\rangle_\Sigma,
\end{align}
where $\cT$ denotes the time-ordering operator and $\langle \cdot \rangle_\Sigma$ is the thermal expectation of a single-qubit Hamiltonian with time-dependent fields sampled according to covariance $\Sigma(\tau_1,\tau_2)$. See \Cref{app:path} for a derivation of these saddle point equations.

\subsection{Replica trick}\label{sec:replica}

To compute and minimize the quenched free energy $\E \log Z_\beta$, we follow the prescription of the replica trick: we evaluate the partition function over $s \geq 1$ copies of the system and then analytically continue $s \to 0$ to obtain
\begin{align}\label{eq:replica}
    \E \log Z_\beta = \lim_{s\to 0} \frac{1}{s}\lr{\E Z_\beta^s - 1}.
\end{align}
Performing an analytic continuation can often be made rigorous; however, because we will evaluate $\E Z_\beta^s$ via the saddle point method by taking $n\to\infty$ first, the quantity $\E Z_\beta^s$ cannot be formally controlled. We must commute limits of $n$ and $s$, which is difficult to control formally.

When $s$ copies are introduced, the replicated partition function becomes $\E Z_\beta^s = e^{s X}$ for
\begin{align}\label{eq:action}
    X &= \int_0^\beta\int_0^\beta d\tau_1\, d\tau_2\, \lr{\frac{J^2}{2p} \sum_{a,b=1}^{s}G_{ab}^p(\tau_1,\tau_2) - \frac{1}{2} \sum_{a,b=1}^{s} \Sigma_{ab}(\tau_1,\tau_2) G_{ab}(\tau_1, \tau_2)}\nonumber\\
    &\quad + \log \mathbb{E}_h \int \cD\Omega\, \exp[\sum_{a=1}^{s} \sum_\mu \int_0^\beta d\tau\, h_\mu^a(\tau) s_\mu(\Omega^a(\tau))].
\end{align}
The saddle point equations (\eqref{eq:saddleannealed} in the annealed case) become
\begin{align}\label{eq:repsaddle}
    J^2 G_{ab}(\tau_1,\tau_2)^{p-1} = \Sigma_{ab}(\tau_1, \tau_2), \qquad G_{ab}(\tau_1, \tau_2) = \left\langle \cT \sum_\mu \sigma^\mu_a(\tau_1)\sigma^\mu_b(\tau_2) \right\rangle_\Sigma
\end{align}
for indices $a, b\in [s]$, and $\Sigma_{ab}(\tau_1,\tau_2)$ also carries replica indices. Here, we use notation
\begin{align}\label{eq:sigmaexp}
    \langle \xi \rangle_\Sigma = \frac{\E_h \Tr(\cT\, \xi \exp[\sum_a \sum_\mu \int_0^\beta d\tau\, h_\mu^a(\tau) \sigma_\mu^a(\tau)])}{\mathbb{E}_h \Tr(\cT\, \exp[\sum_a \sum_\mu \int_0^\beta d\tau\, h_\mu^a(\tau) \sigma_\mu^a(\tau)])},
\end{align}
where the expectation is evaluated over $h_\mu^a$ distributed as a Gaussian process with covariance $\delta_{\mu\nu}\Sigma_{ab}(\tau_1,\tau_2)$. To solve these saddle point equations at arbitrary $s$, we must simplify the form of $G$. Several properties can directly be seen to hold: when $a \neq b$, $G_{ab}(\tau_1,\tau_2) = G_{ab}$ is time-independent; and when replica indices coincide, $G_{aa}(\tau_1,\tau_2) = G(\tau_1-\tau_2)$ due to translation invariance. The form of the overlaps when $a\neq b$ is then specified by assuming an ansatz corresponding to either replica symmetry or replica symmetry breaking.

The simplest choice is the RS ansatz, which takes $G_{ab} = q_0$ for all $a \neq b$:{
  \setlength{\arraycolsep}{4pt}
  \renewcommand{\arraystretch}{1.1}
  \begin{align}\label{eq:rspic}
    \underbrace{%
      \begin{pmatrix}
        G(\tau_1-\tau_2) & q_0 & \cdots & q_0\\
        q_0 & G(\tau_1-\tau_2) & \cdots & q_0\\
        \vdots & \vdots & \ddots & \vdots\\
        q_0 & q_0 & \cdots & G(\tau_1-\tau_2)
      \end{pmatrix}
    }_{s\times s}
  \end{align}
}%
In~\cite{swingle2023bosonicmodelquantumholography}, the $p$-local ensemble was analyzed in the $p\to\infty$ assuming $q_0=0$. In \Cref{sec:p3}, we will analyze the model in the RS ansatz with generic $q_0$ and find a phase transition from $q_0=0$ to $q_0 \neq 0$, indicating glassiness.

To show non-glassiness for sufficiently large $p$, we will analyze the 1RSB solution. The 1RSB ansatz $s \times s$ matrix $G_{ab}(\tau_1,\tau_2)$ to have elements indexed by multi-indices $a=(a_0,a_1)$ for $a_0 \in [m], a_1 \in [s/m]$ and has the form
\begin{align}\
    G_{ab}(\tau_1,\tau_2) = G(\tau_1-\tau_2)\delta_{ab} + q_0(1-\delta_{a_1b_1}) + q_1(1-\delta_{a_0b_0})\delta_{a_1b_1}.
\end{align}
Pictorially, the 1RSB ansatz corresponds to setting the matrix $G(\tau_1,\tau_2)$ to be{
  \setlength{\arraycolsep}{4pt}
  \renewcommand{\arraystretch}{1.1}
  \begin{align}\label{eq:1rsbpic}
    \underbrace{%
      \begin{pmatrix}
        \underbrace{%
          \begin{smallmatrix}
            G(\tau_1-\tau_2)&q_1&\cdots&q_1\\
            q_1&G(\tau_1-\tau_2)&\cdots&q_1\\
            \vdots&\vdots&\ddots&\vdots\\
            q_1&q_1&\cdots&G(\tau_1-\tau_2)
          \end{smallmatrix}
        }_{m\times m}
        &  
        & q_0
      \\[1.5ex]
         & \ddots & 
      \\
        q_0
        & 
        & \underbrace{%
          \begin{smallmatrix}
            G(\tau_1-\tau_2)&q_1&\cdots&q_1\\
            q_1&G(\tau_1-\tau_2)&\cdots&q_1\\
            \vdots&\vdots&\ddots&\vdots\\
            q_1&q_1&\cdots&G(\tau_1-\tau_2)
          \end{smallmatrix}
        }_{m\times m}
      \end{pmatrix}
    }_{s\times s}
  \end{align}
}%
where $q_0$ indicates that all matrix elements outside the block diagonal are filled with $q_0$. For a classical spin glass, the diagonal would be time-independent and equal to 1. The matrix $G$ would be directly interpretable as an overlap matrix with the following operational meaning. Each row and column index labels a ``typical'' spin configuration; drawing two spin configurations independently from the Gibbs distribution corresponds to drawing a uniformly random row and column index; the matrix element $G_{ab} \in \{1, q_1, q_0\}$ corresponds to the resulting overlap. Note that this overlap matrix only holds with probability $1-o(1)$, that $q_0, q_1$ are the values that the overlaps concentrate to, and that the size $m$ of the interior blocks is also a random value that the system concentrates to.

In a quantum spin glass, the values $q_0, q_1$ no longer correspond to overlaps between spin configurations. Instead, the path integral is formally decomposed over spin coherent states, and the overlaps are between paths. Although the paths are nonphysical, the form of the ansatz (and the choice of spin coherent states) relates $q_0, q_1$ to the decomposition in our definition of glassiness (\Cref{def:introglass}). A state decomposed as
\begin{align}\label{eq:glassdecomp}
    \bm \rho \approx \sum_i c_i \bm \rho_i
\end{align}
has
\begin{align}
    \langle \bm\rho_i, \bm\rho_j \rangle_\pauli = q_0 (1-\delta_{ij}) + q_1 \delta_{ij}.
\end{align}
Thus, the value $q$ in \Cref{def:introglass} is given by $2(q_1-q_0)$. The values $q_0, q_1$ can also be interpreted as overlaps between paths in the path integral representation of the partition function due to the delta function
\begin{align}
    \delta\lr{G_{ab}(\tau_1, \tau_2) - \frac{1}{n}\sum_{r,\mu} s_{r\mu}^a(\tau_1) s_{r\mu}^b(\tau_2)}.
\end{align}
However, the interpretation as a decomposition of the Gibbs state is more useful for our purposes since the paths are nonphysical quantities. In contrast, results evaluating the dynamics of quantum spin glasses when weakly coupled to a bath show that the two-point function $\langle \sum_\mu \sigma^\mu(\tau_1) \sigma^\mu(\tau_2) \rangle$ converges to values $q_1$ then $q_0$ at late times, providing the necessary physical grounds for our decomposition~\cite{cugliandolo1993analytical,cugliandolo1998quantum,cugliandolo1999real,biroli2001quantum,kennett2001time,biroli2002out,cugliandolo2006dissipative,lang2024replica}.

When evaluating $\E \log Z_\beta$ with the 1RSB ansatz, one must choose the value of $m$ to take. When performing the analytic continuation $s \to 0$, the value will lie in $m \in (0, 1]$, with larger $m$ corresponding to smaller coefficients $c_i$ in the cluster decomposition \eqref{eq:glassdecomp}. In static 1RSB, one solves a saddle point equation $\partial_m X = 0$. As discussed extensively for the case of quantum spin glasses in~\cite{baldwin2023revisiting}, one \emph{maximizes} the free energy with respect to $m$ (due to the limit $s \to 0$) in order to obtain the correct quenched free energy. A solution of $m \in (0,1]$ corresponds to the Gibbs state being dominated by a constant number of clusters (i.e., in the notation of \Cref{def:shatterrsb}, constant $m \geq 2$).

In shattering (also known as dynamical 1RSB), the saddle point equation $\partial_m X = 0$ is unsatisfied even at the extremal value of $m=1$. To check for shattering, one uses the formalism of the Thouless-Anderson-Palmer (TAP) free energy~\cite{monasson1995structural,franz1998effective,mezard1987spin,thouless1977solution,biroli2001quantum}. Rather than solving the saddle point equation, $m=1$ is fixed and the saddle point equation is reinterpreted as a quantity known as the TAP complexity, which counts the number of ``TAP states'' $\bm \rho_i$ in the decomposition \eqref{eq:glassdecomp}. To determine if the system is glassy, the stability of the 1RSB ansatz is examined via the marginality condition, which checks the value of the so-called replicon eigenvalue. 

In this work, we will rule out glassiness for sufficiently large $p$ in \Cref{sec:1rsb} by analyzing the 1RSB saddle point equations. Since dynamical 1RSB is always followed by static 1RSB (as $m$ eventually decreases from 1), we avoid checking the marginality condition and will simply rule out static 1RSB.

Before going through these results, we comment that the glassiness of the $p=3$ model can be further analyzed by computing its 1RSB and full RSB solutions. The full RSB solution is obtained by a recursive ansatz that repeatedly places block matrices within the block matrices of the 1RSB ansatz; it can be written as the $k\to\infty$ limit of the general $k$RSB ansatz. We derive the full RSB saddle point equations in \Cref{app:frsb} and note that they can be numerically solved in a manner similar to~\cite{kavokine2024exact}. For the purpose of identifying a glass transition, however, we only require the analysis of \Cref{sec:p3}.

\subsection{Saddle point equations}\label{sec:saddles}
We report here the saddle point equations obtained from the RS, 1RSB and full RSB solutions. Due to the spin-1/2 operators, the saddle point equations cannot be analytically solved in closed form but can be numerically solved (see, e.g.,~\cite{swingle2023bosonicmodelquantumholography,kavokine2024exact}). We emphasize that these numerical results analytically take the $n\to\infty$ limit and are solving the quenched disorder at any given $p, \beta$ and thus reflect glassiness; they do not capture finite-$n$ or annealed physics. In \Cref{sec:p3} and \Cref{sec:1rsb} we both numerically solve these saddle point equations and analytically analyze them to show either the presence or absence of glassiness (depending on $p,\beta$).

The simplest saddle point equations are in the RS ansatz. Taking
\begin{align}
    G_{ab}(\tau_1, \tau_2) &= G(\tau_1-\tau_2)\delta_{ab} + q_0(1-\delta_{ab})\\
    \Sigma_{ab}(\tau_1, \tau_2) &= \Sigma(\tau_1-\tau_2)\delta_{ab} + \hat q_0(1-\delta_{ab}),
\end{align}
the action of \eqref{eq:action} becomes
\begin{align}
    X &= \frac{s}{2}\int_0^\beta\int_0^\beta d\tau_1\, d\tau_2\, \lr{\frac{J^2 G(\tau_1-\tau_2)^p}{p} - \Sigma(\tau_1-\tau_2) G(\tau_1-\tau_2)} + \frac{\beta^2 s(s-1)}{2}\lr{\frac{J^2 q_0^p}{p} - \hat q_0 q_0}\nonumber\\
    &\quad + \log \mathbb{E}_h \int \cD\Omega\, \exp[\sum_{a=1}^{s} \sum_\mu \int_0^\beta d\tau\, h_\mu^a(\tau) s_\mu(\Omega^a(\tau))].
\end{align}
The saddle point equations of \eqref{eq:repsaddle} then become
\begin{align}\label{eq:rssaddle}
    J^2 G(\tau_1-\tau_2)^{p-1} = \Sigma(\tau_1-\tau_2), \quad J^2 q_0^{q-1} = \hat q_0, \quad G(\tau_1-\tau_2) = \left\langle\sum_\mu \sigma_\mu^1(\tau_1) \sigma_\mu^1(\tau_2)\right\rangle, \quad q_0 = \left\langle\sum_\mu \sigma_\mu^1(\tau_1) \sigma_\mu^2(\tau_2)\right\rangle.
\end{align}
In \Cref{app:1rsb}, we derive the 1RSB saddle point equations. For our analysis, it is more useful to write them in terms of random variables
\begin{align}
    z_\mu \sim \cN(0, \delta_{\mu\nu}J^2 q_1^{p-1}), \qquad \xi_\mu(\tau) \sim \cN(0, \delta_{\mu\nu}(\Sigma(\tau_1-\tau_2) - J^2 q_1^{p-1}))
\end{align}
and quantities
\begin{align}
    \zeta(z) &= \E_{\xi_\mu(\tau)} \Tr( \cT \exp[\sum_\mu \int_0^\beta d\tau\, (z_\mu + \xi_\mu(\tau)) \sigma_\mu(\tau)] )\\
    a_\nu(z;\tau_1, \dots, \tau_k) &= \frac{\E_{\xi_\mu(\tau)} \Tr(\cT \sigma_\nu(\tau_1)\cdots\sigma_\nu(\tau_k) \exp[\sum_\mu \int_0^\beta d\tau\, (z_\mu+\xi_\mu(\tau)) \sigma_\mu(\tau)])}{\zeta(z)}.
\end{align}
We evaluate only the case where $q_0=0$, which suffices to rule out glassiness when $p\to\infty$. The saddle point equations are
\begin{align}
    \Sigma(\tau_1 - \tau_2) &= J^2 G(\tau_1 - \tau_2)^{p-1}\label{eq:1rsb-eq1}\\
    G(\tau_1-\tau_2) &= \frac{\E_z\left[\zeta(z)^m \sum_\nu a_\nu(z;\tau_1,\tau_2)\right]}{\E_z\left[\zeta(z)^m\right]}\\
    q_1 &= \frac{\E_z\left[\zeta(z)^m \sum_\nu a_\nu(z;\tau_1) a_\nu(z;\tau_2)\right]}{\E_z\left[\zeta(z)^m\right]}\label{eq:1rsb-eq3}\\
    0 &= \frac{(\beta J)^2}{2} \lr{1 - \frac{1}{p}} q_1^{p} + \frac{1}{m^2}\log \E_z \zeta(z)^m - \frac{1}{m}\frac{\E_z \zeta(z)^m \log \zeta(z)}{\E_z \zeta(z)^m}\label{eq:1rsb-eq4}.
\end{align}
The final saddle point equation at $m=1$ corresponds to the TAP complexity, i.e.,
\begin{align}\label{eq:tapS}
    S &= \frac{(\beta J)^2}{2} \lr{1 - \frac{1}{p}} q_1^{p} + \log \E_z \zeta(z) - \frac{\E_z \zeta(z) \log \zeta(z)}{\E_z \zeta(z)}.
\end{align}
Although we do not provide a numerical solution in this manuscript, we also derive in \Cref{app:frsb} the full RSB saddle point equations. To define full RSB, we recall that the 1RSB ansatz of \eqref{eq:1rsbpic} assumes $G$ has a block diagonal form. In 2RSB, each block is similarly assumed to be block-diagonal; this introduces a new parameter $q_2$ satisfying $q_2 > q_1 > q_0$. (If $q_2=q_1$, then the model is simply 1RSB.) This hierarchical ansatz of block matrices inside block matrices can be generalized to $k$RSB for arbitrary integer $k$. Formally, the matrix is $G^\rsbk$ can be recursively constructed from $G:[-1,1]\to\mathbb{R}$ and two sequences $(q_t)_{0,\dots, k}$ and $(m_t)_{0,\dots,k}$. We set $m_0=n$ and define the $k$th stage of RSB according to
\begin{align}
    G^\rsbk_k &= (q_k - q_{k-1})\,U_{m_k}, \\
    G^\rsbk_p &= \operatorname{diag}_{\,m_t/m_{t+1}}\lr{G^\rsbk_{t+1}} + (q_t - q_{t-1})U_{m_t} \text{for all } t\in\{0,\dots,k-1\} \\
    G^\rsbk &\equiv G^\rsbk_0.
\end{align}
Here, $U_m$ is the $m\times m$ matrix of ones, $\operatorname{diag}_r(A)$ is the block-diagonal matrix with $r$ copies of $A$, $m_0=n$, and $q_{-1}=0$. The ansatz of \emph{full RSB} takes the $k\to\infty$ limit and thus replaces $q_0, \dots, q_k$ with the so-called Parisi order parameter $q(x)$ for $x \in [0,1]$, which monotonically increases. If $q(x)$ is constant, the model is replica symmetric; if it is a step function, the model has 1RSB; in general, $q(x)$ can be continuous. We derive the full RSB saddle point equations in \Cref{app:frsb} and report the result here. For single-qubit action
\begin{align}
    \cS_\infty(h) &= \cS_0[s(\tau)] - \int_0^1 d\tau_1\, d\tau_2\, \frac{(\beta J)^2}{2}\lr{G(\tau_1-\tau_2)^{p-1} - q(1)^{p-1}} s(\tau_1) \cdot s(\tau_2) - \beta \int_0^1 d\tau\, h \cdot s(\tau)
\end{align}
we have
\begin{align}
    G(\tau) = \int dh\, P(1,h) \langle \cT \, s(\tau) \cdot s(0) \rangle_{\cS_\infty(h)}, \qquad q(x) = \int dh \, P(x,h) u(x,h)^2
\end{align}
for $P, u$ obtained by solving
\begin{align}
    \frac{\partial P}{\partial x} &= \frac{J^2}{2} \frac{dq^{p-1}}{dx} \lr{\nabla^2 P - 2\beta x \nabla(u \cdot P)}, \qquad P(0, h) = \delta(h)\\
    \frac{\partial u}{\partial x} &= -\frac{J^2}{2} \frac{dq^{p-1}}{dx} \lr{\nabla^2 u + 2 \beta x (u\cdot \nabla) u}, \qquad u(1,h) = \int_0^1 d\tau\, \langle s(\tau) \rangle_{\cS_\infty(h)}.
\end{align}
The free energy is then
\begin{align}
    F &= \frac{\beta J^2}{2}\lr{1 - \frac{1}{p}} \int_0^1 d\tau\, \lr{G(\tau)^p - \int_0^1 dx\, q(x)^p} - \frac{1}{\beta}\phi(0,0)
\end{align}
for
\begin{align}
    \frac{\partial \phi}{\partial x} &= -\frac{J^2}{2} \frac{dq^{p-1}}{dx} \lr{\nabla^2 \phi + x(\nabla \phi)^2}, \qquad \phi(1,h) = \log \int \cD \Omega \exp[-\cS_\infty(h)].
\end{align}

\subsection{Glassiness of the 3-local Hamiltonian ensemble}\label{sec:p3}

\begin{figure}
    \centering
    \includegraphics[width=0.45\linewidth]{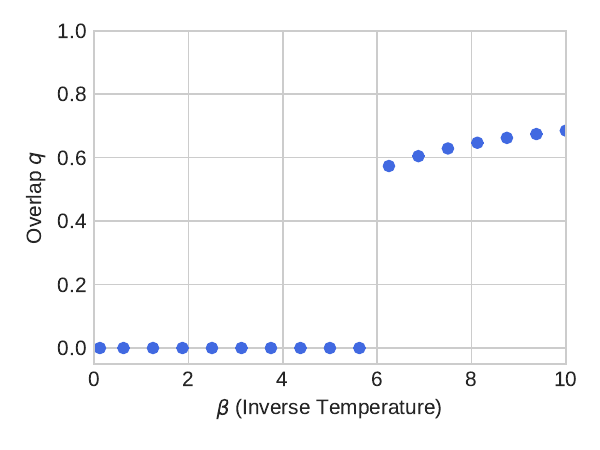}
    \caption{Numerical solutions to the replica saddle point equations indicate the presence of replica symmetry breaking in the $p=3$ model. Here, values of the overlap $q_0$ in the RS solution of the $p=3$ model are plotted over various values of inverse temperature $\beta$ with $J=1$ throughout. A phase transition in the off-diagonal value indicates that the RS ansatz is incorrect at some temperature $\beta \approx 6$, demonstrating a glass transition at constant temperature. We note that because the glass transition likely occurs before the off-diagonals of the RS ansatz become nonzero, the RS solver is fairly unstable and the precise temperature of the glass transition cannot be deduced from this plot alone; $\beta \approx 6$ only provides an approximate estimate of the transition.} 
    \label{fig:rs}
\end{figure}

In \Cref{fig:rs}, we numerically solve the RS saddle point equations~\eqref{eq:rssaddle} with a generic off-diagonal $q_0$ in the RS ansatz, using CTSEG~\cite{kavokine2025ctseg} and code adapted from the recent solution of the quantum Heisenberg model~\cite{kavokine2024exact}. This gives $G(\tau)$ and $q_0$ via a continuous-time quantum Monte Carlo algorithm (\Cref{app:rs}. By inspecting $q_0$, we find a clear phase transition at $\beta \approx 6$, indicating the presence of replica symmetry breaking. This is consistent with the breakdown of RS and the need for a 1RSB ansatz (which has two different off-diagonal values). We conclude that a glass transition occurs at some $\beta J \lesssim 6$. By numerically solving the 1RSB saddle point equations, one could in principle identify the precise value of $\beta J$ where the glass transition occurs; one could also solve the full RSB saddle point equations at even lower temperature to fully solve the physics of the ensemble. We leave such numerical solutions for future work.

We also give an analytical estimate of the glass transition temperature. We compute the stability of the RS solution under the constant trial function $G(\tau)=G$. This suggests the glass transition occurs at inverse temperature $\beta J \gtrsim 4$. Moreover, it closely corroborates our numerical results: if we set $q_0 = 0.6$, our analytical approximation suggests a transition at $\beta J \approx 5.9$, which closely matches the transition at $\beta J \approx 6$ observed in \Cref{fig:rs}. Our method resembles that of~\cite{bray1980replica} used to predict the glass transition temperature of the quantum Heisenberg model, which was recently numerically confirmed via an exact full RSB solution~\cite{kavokine2024exact}.

We start from the action $X$ of \eqref{eq:action} specialized to the RS ansatz with zero on the off-diagonals, which we denote by $X^{(0)}$, i.e.,
\begin{align}
    X^{(0)} &= \frac{s}{2}\int_0^\beta\int_0^\beta d\tau_1\,d\tau_2\, \lr{\frac{J^2G(\tau_1-\tau_2)^p}{p} - \Sigma(\tau_1-\tau_2)G(\tau_1-\tau_2)} + \log \mathbb{E}_h \int  \cD\Omega\, \exp[\sum_{a=1}^{s} \sum_\mu \int_0^\beta d\tau\, h_\mu^a(\tau) s_\mu(\Omega^a(\tau))].
\end{align}
Recall that the corresponding saddle point equations are
\begin{align}
    \Sigma(\tau) = J^2 G(\tau)^{p-1}, \qquad G(\tau) = \left\langle \sum_\mu \sigma_\mu(\tau) \sigma_\mu(0) \right\rangle_\Sigma.
\end{align}
We evaluate $F$ with the trial function $G(\tau, \tau') = G$ independent of $\tau, \tau'$ to obtain free energy
\begin{align}
    F \leq \frac{\beta J^2}{2}\lr{1 - \frac{1}{p}} G^p - \frac{1}{\beta} \log \zeta(G)
\end{align}
where we can exactly evaluate
\begin{align}
    \zeta(G) &= \lim_{s\to 0}\partial_s \E_h \int  \cD\Omega\, \exp[\sum_{a=1}^{s} \sum_\mu \int_0^\beta d\tau\, h_\mu^a(\tau) s_\mu(\Omega^a(\tau))]\\
    &= \Tr \cT \exp[\frac{\beta^2 J^2 G^{p-1}}{2} \lr{\int_0^1 d\tau \, \vec \sigma(\tau)}^2]\\
    &= 2\lr{1 + \frac{\beta^2 J^2 G^{p-1}}{4}} \exp[\frac{\beta^2 J^2 G^{p-1}}{8}]
\end{align}
using the identity~\cite{bray1980replica}
\begin{align}
    \Tr \cT \exp[\alpha\lr{\int_0^1 d\tau \, \vec \sigma(\tau)}^2] &= 2\lr{1 + \frac{\alpha}{2}} \exp[\frac{\alpha}{4}].
\end{align}
This gives variational bound
\begin{align}
    F \leq \frac{\beta J^2}{2}\lr{1 - \frac{1}{p}} G^p - \frac{1}{\beta} \log \left[2\lr{1 + \frac{\beta^2 J^2 G^{p-1}}{4}} \exp[\frac{\beta^2 J^2 G^{p-1}}{8}]\right].
\end{align}
We minimize the RHS at $p=3$ by choosing
\begin{align}\label{eq:gopt}
    G &= \frac{\beta^2 J^2 + \beta J \mathcal A^{1/3} + \mathcal A^{2/3} - 192}{12 \beta J \mathcal A^{1/3}} \text{ for } \mathcal A = \beta^3 J^3 + 24\sqrt{9\beta^4 J^4 + 9024 \beta^2 J^2 + 12288} + 2304\,\beta J.
\end{align}
Note that as $\beta\to\infty$, this gives
\begin{align}
    F \leq - \frac{\beta J^2}{384}
\end{align}
to leading order; since it is nonphysical for the free energy to diverge, this computation is consistent with the presence of a glass phase. However, it is known that this computation does not necessarily imply glassiness: a more careful treatment must first maximize the free energy with respect to replica parameters and then minimize with respect to conventional parameters~\cite{baldwin2023revisiting}.

To estimate the temperature of the glass transition, we estimate the appropriate value of $G$ and solve for $\beta J$ to obtain the temperature of the glass transition. We take the action
\begin{align}
    X &= \int_0^\beta\int_0^\beta d\tau_1\, d\tau_2\, \lr{\frac{J^2}{2p} \sum_{a,b=1}^{s}G_{ab}^p(\tau_1,\tau_2) - \frac{1}{2} \sum_{a,b=1}^{s} \Sigma_{ab}(\tau_1,\tau_2) G_{ab}(\tau_1, \tau_2)}\nonumber\\
    &\quad + \log \mathbb{E}_h \int \cD\Omega\, \exp[\sum_{a=1}^{s} \sum_\mu \int_0^\beta d\tau\, h_\mu^a(\tau) s_\mu(\Omega^a(\tau))].
\end{align}
with an RS ansatz that takes nonzero off-diagonals (i.e., $G_{ab}(\tau_1,\tau_2) = q_0, \Sigma_{ab}(\tau_1,\tau_2) = \hat q_0$ when $a \neq b$) to obtain
\begin{align}
    X &= \frac{s}{2}\int_0^\beta\int_0^\beta d\tau_1\,d\tau_2\, \lr{\frac{J^2G(\tau_1-\tau_2)^p}{p} - \Sigma(\tau_1-\tau_2)G(\tau_1-\tau_2)} + \frac{\beta^2(s-1)s}{2}\lr{\frac{J^2q_0^p}{p} - \hat q_0 q_0}\nonumber \\
    &\quad + \log \mathbb{E}_h \int  \cD\Omega\, \exp[\sum_{a=1}^{s} \sum_\mu \int_0^\beta d\tau\, h_\mu^a(\tau) s_\mu(\Omega^a(\tau))].
\end{align}
We expand the final term around the RS ansatz with vanishing off-diagonals. Writing it as
\begin{align}
    A(\Sigma) = \log \Tr \cT \exp[\frac{1}{2}\sum_{a,b=1}^s \sum_\mu \int_0^\beta d\tau_1\, d\tau_2\, \Sigma_{ab}(\tau_1,\tau_2) \sigma_\mu^a(\tau_1) \sigma_\mu^b(\tau_2)]
\end{align}
we expand around the $\hat q_0 = 0$ in $\Sigma^{(0)}$ to obtain
\begin{align}
    A(\Sigma) \approx A(\Sigma^{(0)}) + \frac{1}{2} \sum_{a \neq b} \sum_\mu \int_0^\beta d\tau_1\,d\tau_2\, \hat q_0 \langle \sigma_\mu^a(\tau_1) \sigma_\mu^b(\tau_2) \rangle_{\Sigma^{(0)}} + \frac{1}{8} \sum_{\substack{a \neq b\\c \neq d}} \sum_{\mu,\nu} \int_0^\beta d\tau_1\cdots d\tau_4 \hat q_0^2 \langle \sigma_\mu^a(\tau_1) \sigma_\mu^b(\tau_2); \sigma_\nu^c(\tau_3) \sigma_\nu^d(\tau_4) \rangle_{\Sigma^{(0)}}^\mathrm{conn}
\end{align}
for connected correlator $\langle A;B \rangle^\mathrm{conn} = \langle AB \rangle - \langle A \rangle \langle B \rangle$. Replica independence in $\Sigma^{(0)}$ gives for $a \neq b, c\neq d$ that
\begin{align}
    \langle \sigma_\mu^a(\tau_1) \sigma_\mu^b(\tau_2) \rangle_{\Sigma^{(0)}} &= 0\\
    \left\langle \sum_{\mu\nu} \sigma_\mu^a(\tau_1) \sigma_\mu^b(\tau_2);\sigma_\nu^c(\tau_3) \sigma_\nu^d(\tau_4) \right\rangle_{\Sigma^{(0)}}^\mathrm{conn} &= \delta_{ac}\delta_{bd} \sum_\mu \langle \sigma_\mu(\tau_1) \sigma_\mu(\tau_3) \rangle_{\Sigma^{(0)}} \langle \sigma_\mu(\tau_2) \sigma_\mu(\tau_4) \rangle_{\Sigma^{(0)}} \nonumber\\
    &\qquad + \delta_{ad}\delta_{bc} \sum_\mu \langle \sigma_\mu(\tau_1) \sigma_\mu(\tau_4) \rangle_{\Sigma^{(0)}} \langle \sigma_\mu(\tau_2) \sigma_\mu(\tau_3) \rangle_{\Sigma^{(0)}},
\end{align}
where we also used $\langle \sigma_\mu(\tau) \sigma_\nu(\tau')\rangle_{\Sigma^{(0)}} = \delta_{\mu\nu} \langle \sigma_\mu(\tau) \sigma_\mu(\tau') \rangle_{\Sigma^{(0)}}$. We now apply saddle point equations
\begin{align}
    \hat q_0 = J^2 q_0, \qquad G(\tau-\tau') = \sum_\mu \langle \sigma_\mu(\tau) \sigma_\mu(\tau') \rangle_{\Sigma^{(0)}}.
\end{align}
This gives
\begin{align}
    A(\Sigma) \approx A(\Sigma^{(0)}) + \frac{J^4 q_0^{2p-2} s(s-1)}{12} \lr{\int_0^\beta d\tau_1\,d\tau_2\, G(\tau_1-\tau_2)}^2
\end{align}
and thus
\begin{align}
    X &\approx X^{(0)} - \frac{(\beta J)^2 q_0^p (s-1)s}{2}\lr{1-\frac{1}{p}} + \frac{J^4 q_0^{2p-2} s(s-1)}{12} \lr{\int_0^\beta d\tau_1\,d\tau_2\, G(\tau_1-\tau_2)}^2.
\end{align}
Under our ansatz that $G(\tau) = G$ is constant, this gives
\begin{align}
    X &\approx X^{(0)} + \frac{(\beta J)^2 q_0^p s(s-1)}{2}\left[-\lr{1-\frac{1}{p}} + \frac{(\beta J)^2 q_0^{p-2}}{6} G^2\right].
\end{align}
At the glass transition, one would expect the RS solution to become marginally stable; i.e., the perturbation to $X^{(0)}$ should vanish. This gives
\begin{align}
    1 - \frac{1}{p} = \frac{(\beta J)^2 q_0^{p-2}}{6} G^2 \implies G = \frac{1}{\beta J} \sqrt{\frac{6(p-1)}{p q_0^{p-2}}}.
\end{align}
We plug this into \eqref{eq:gopt} at $p=3$ and set $q_0 < 1$ to obtain approximate bound $\beta J \gtrsim 4$. Note that if we take $q_0 \approx 0.6$ from \Cref{fig:rs} (which solves $G(\tau)$ via the RS saddle point equations rather than assuming $G(\tau) = G$), we obtain $\beta J \approx 5.9$, which closely agrees with the value of $\beta \approx 6$ shown in \Cref{fig:rs}. We also note that this yields $G \approx 0.4$ (and which decreases further with larger $\beta$), similar to the quenched saddle and dissimilar to the $G \to 1$ annealed saddle observed in~\cite{swingle2023bosonicmodelquantumholography}.

\subsection{Non-glassiness of the \texorpdfstring{$q$}{q}-local Hamiltonian ensemble for large \texorpdfstring{$q$}{q}}\label{sec:1rsb}

Here, we show that the 1RSB saddle point equations cannot be satisfied for any $\beta J$ satisfying $(\beta J)^2 q_1^{p-1} \to 0$ as $p \to \infty$. For $q_1$ bounded away from 0, this precludes any $\beta J = \exp[o(p)]$, which is the form of the result stated in the introduction. Since a glass transition would occur at $m=1$, it suffices to rule out a $q_1 \neq 0$ solution to the 1RSB saddle point equations at $m=1$. We recall that the final saddle point equation can be written as (using \eqref{eq:tapS})
\begin{align}
    S = \frac{(\beta J)^2}{2} \lr{1 - \frac{1}{p}} q_1^{p} + \log \E_z \zeta(z) - \frac{\E_z \zeta(z) \log \zeta(z)}{\E_z \zeta(z)} = 0.
\end{align}
For $Q = \cN(0, \hat q_1 \delta_{\mu\nu})$ and measure
\begin{align}
    dP(z) = \frac{\zeta(z)}{\E_Q \zeta} dQ(z),
\end{align}
we can rewrite this quantity as (see \Cref{lem:kl} in \Cref{app:1rsb})
\begin{align}
    S = \frac{(\beta J)^2 q_1^p}{2}\lr{1 - \frac{1}{p}} - D_\KL(P||Q).
\end{align}
Observe that if $q_1 = 0$, then $S$ cannot be positive; we will show that if $q_1 \neq 0$, then $S > 0$ for all $(\beta J)^2 q_1^{p-1} = o(1)$ in asymptotically large $p$, preventing the saddle point equations from being satisfied. To complete this argument, it suffices to upper-bound $D_\KL(P||Q)$.

The Gaussian log-Sobolev inequality gives an upper bound on the KL divergence
\begin{align}
    D_\KL(P||Q) \leq \frac{\hat q_1}{2} \E_P \norm{\nabla \log \zeta(z)}^2.
\end{align}
Using the saddle point equation at $m=1$
\begin{align}
    q_1 = \frac{\E_z \zeta(z) \sum_\nu a_\nu(z;\tau_1) a_\nu(z;\tau_2)}{\E_z \zeta(z)},
\end{align}
one can compute that
\begin{align}
    \E_P \norm{\nabla \log \zeta(z)}^2 = \E_P \sum_\nu \lr{\int_0^\beta a_\nu(z;\tau)d\tau}^2 = \beta^2 \frac{\E_z \zeta(z) \sum_\nu a_\nu(z;\tau_1) a_\nu(z;\tau_2)}{\E_z \zeta(z)} = \beta^2 q_1.
\end{align}
The saddle point equation
\begin{align}
    \hat q_1 = J^2 q_1^{p-1}
\end{align}
then gives upper bound
\begin{align}
    D_\KL(P||Q) \leq \frac{(\beta J)^2 q_1^p}{2}.
\end{align}
Unfortunately, this yields a bound $S \geq S \geq -\frac{(\beta J)^2 q_1^p}{2p}$ that is too loose. In \Cref{app:1rsb}, we instead use the sharpened log-Sobolev inequality~\cite{fathi2014quantitative}
\begin{align}
    D_\KL(P||Q) \leq \frac{(\beta J)^2q_1^q}{2} \cdot \frac{1 - \lambda + \lambda \log \lambda}{(1-\lambda)^2}
\end{align}
for a Poincar\'e constant $\lambda$ computable in terms of a two-point correlator. For sufficiently small $\beta$, we find that the factor $\frac{1 - \lambda + \lambda \log \lambda}{(1-\lambda)^2}$ yields a nontrivial improvement over Gaussian log-Sobolev due to strong convexity of the potential associated with $P$. This gives the following result, which is rigorous assuming the 1RSB saddle point equations~\eqref{eq:1rsb-eq1}-\eqref{eq:1rsb-eq3}. The proof is in \Cref{app:1rsb}.

\begin{theorem}[No nontrivial solution to 1RSB saddle point equations]\label{thm:no1rsb}
    Let $\lambda = 1 - (\beta J)^2 q_1^{p-1} \geq 1 - x_*$ for some $x_* < 1$. Then the quantity $S$ in \eqref{eq:tapS} satisfies
    \begin{align}
        S \geq \frac{(\beta J)^2 q_1^p}{2}\lr{1 - \frac{1}{p} - \frac{1-\lambda + \lambda \log \lambda}{(1-\lambda)^2}}.
    \end{align}
    In particular, if $(\beta J)^2 q_1^{p-1} \to 0$ as $p \to \infty$, then
    \begin{align}
        S \geq \frac{(\beta J)^2 q_1^p}{4}\lr{1 + o_p(1)}.
    \end{align}
\end{theorem}

Since the 1RSB saddle point \eqref{eq:1rsb-eq4} is $S=0$ at $m=1$, we conclude that the saddle point equations cannot cannot be satisfied at $m=1$ for sufficiently large $p$ unless $q_1 = 0$. Since glassiness (e.g., the transition from dynamical to static 1RSB) is expected to occur at $m=1$, this provides strong evidence that there is no glassy phase for any $p \geq p_*$. By naively combining the numerical results of \Cref{fig:glass} with \Cref{thm:no1rsb}, we can obtain an extremely rough estimate the value of $p_*$.

\begin{remark}[Non-glassy $p$]
    Fixing $\beta J \approx 4$ and $q_1 \approx 0.6$ from \Cref{fig:rs}, \Cref{thm:no1rsb} implies non-glassiness for all $p \gtrsim 8$.
\end{remark}

The prior finite-$n$ numerical results of~\cite{swingle2023bosonicmodelquantumholography} at the ground state ($\beta$ extensively growing with $n$) suggested non-glassiness for all $p\gtrsim 5$. Consistent with this prediction --- and indeed roughly with the value of non-glassy $p$ --- our analytical result instead addresses constant temperature. We also note that in the classical case, the large-$p$ Ising $p$-spin model has a static 1RSB phase at inverse temperature $\beta J = \sqrt{2 \log 2}$ up to exponentially small corrections in $p$~\cite{talagrand2000rigorous}; this is much smaller than the $\beta J = \exp[o(p)]$ regime in which our result holds. Hence, our result provides strong evidence that the $p$-local Pauli ensemble is non-glassy at any constant temperature for all $p > p_*$ even at a fairly small constant $p_*$.

\begin{acknowledgments}
    ERA is funded in part by the Walter Burke Institute for Theoretical Physics at Caltech. AZ is supported by the Hertz Foundation. The authors thank David Gamarnik and Brian Swingle for discussions on an early version of this work, and especially thank Michael Winer for many conversations about spin glass techniques and the physics of $p$-local Hamiltonians. BK thanks Nikita Kavokine for helpful feedback about running the numerical simulations. AZ thanks Damien Barbier for an introduction of the replica method, Vittorio Erba and Bruno Loureiro for the Les Houches workshop ``Towards a theory for typical-case algorithmic hardness'', and Giulio Biroli, Leticia Cugliandolo, Laura Foini, Aram Harrow, and Florent Krzakala for helpful discussions about quantum spin glasses.
\end{acknowledgments}

\bibliography{main}

\newpage
\appendix
\section{Path integrals for the \texorpdfstring{$q$}\ -local Hamiltonian ensemble}\label{app:path}

This appendix shows how to obtain the saddle point equations of the random $q$-local Hamiltonian (\Cref{def:pspin}). This computation already exists in the literature in~\cite{swingle2023bosonicmodelquantumholography}, but we attempt to provide a self-contained derivation of it in somewhat more pedagogical detail here. The later appendices provide new computations for RSB. Throughout the appendices, we will use the SYK-like notation of $N$ qubits, $q$ locality, and $n$ replicas.

\subsection{Partition function}
Let $D^N \Omega$ be the (normalized) path integral measure over spin configurations $\Omega_r(\tau)$ for $r \in [N]$ and $d\Omega_r(\tau) = d\theta_r(\tau)\,d\phi_r(\tau)\,\sin\theta_r(\tau)$ which is formally defined later. 
In this section, we rewrite the random $q$-local Hamiltonian partition function in the form
\begin{align}
    Z = \int \cD^N \Omega \,\exp\left[-\int_0^\beta d\tau \sum_I J_I s_{I_1}(\tau)\cdots s_{I_q}(\tau)\right].
\end{align}
More carefully, one could write this is in discretized form where $s_{r\mu}(\tau)$ is $s_\mu(\theta_r^{(\tau/\Delta\tau)}, \phi_r^{(\tau/\Delta\tau)})$, the exponential is trotterized, and
\begin{align}
    \cD^N \Omega &= \prod_{r=1}^N \prod_{i=1}^{\beta/\Delta \tau} \frac{d\Omega_r^{(i)}}{2\pi} \lr{\cos \frac{\theta_r^{(i)}}{2} \cos \frac{\theta_r^{(i+1)}}{2} + e^{-i(\phi_r^{(i)}-\phi_r^{(i+1)})}\sin \frac{\theta_r^{(i)}}{2} \sin \frac{\theta_r^{(i+1)}}{2}}
\end{align}
for
\begin{align}
    d\Omega_r^{(i)} = d\theta_r^{(i)} \, d\phi_r^{(i)} \, \sin\theta_r^{(i)}, \quad s_x(\theta,\phi) = 3\sin\theta\cos\phi, \quad s_y(\theta,\phi) = 3\sin\theta\sin\phi, \quad s_z(\theta,\phi) = 3\cos\theta.
\end{align}
For $\theta \in [0, \pi)$ and $\phi \in [0,2\pi)$, a spin coherent state is defined as
\begin{align}
    \ket{\Omega} = \begin{pmatrix}\cos \frac{\theta}{2} \\ e^{i\phi}\sin \frac{\theta}{2}\end{pmatrix}.
\end{align}
These have the convenient property that for any $\mu$,
\begin{align}
    \sigma_\mu = \frac{1}{2\pi}\int d\Omega\, \ketbra{\Omega} s_\mu.
\end{align}
Given $n$ parameters $\theta_1, \dots, \theta_n$ and $\phi_1, \dots, \phi_n$, each pair $\theta_i, \phi_i$ denotes a single qubit state $\ket{\Omega_i}$. 
We use shorthand notation $\ket{\Omega_1, \dots, \Omega_n}$ to then denote the product state $\ket{\Omega_1} \otimes \cdots \otimes \ket{\Omega_n}$.

To write the path integral in terms of the $s$ variables, we observe that for any time $\Delta \tau$
\begin{align}
    1-\Delta \tau \sum_I J_I \sigma_{I_1}\cdots \sigma_{I_q} = \int \prod_{r=1}^N \frac{d\Omega_r}{2\pi} \ketbra{\Omega_1, \dots, \Omega_n} \lr{1-\Delta \tau \sum_I J_I s_{I_1} \cdots s_{I_q}}.
\end{align}
We take the product for an integer $\beta/\Delta \tau$ times, giving
\begin{align}
    \Tr[\lr{1-\Delta \tau \sum_I J_I \sigma_{I_1}\cdots \sigma_{I_q}}^{\beta/\Delta \tau}] = \int \tilde \cD^N \Omega \prod_{i=1}^{\beta/\Delta \tau}\lr{1-\Delta \tau \sum_I J_I s_{I_1}^{(i)} \cdots s_{I_q}^{(i)}}
\end{align}
for
\begin{align}
    \tilde \cD^N \Omega = \Tr[\prod_{r=1}^N \prod_{i=1}^{\beta/\Delta \tau} \frac{d\Omega_r^{(i)}}{2\pi} \ketbra{\Omega_r^{(i)}}].
\end{align}
In the limit $\Delta \tau \to 0$, this becomes
\begin{align}
    \Tr(\exp\left[-\beta \sum_I J_I \sigma_{I_1}\cdots \sigma_{I_q} \right]) = \int \cD^N \Omega \exp\left[-\int_0^\beta d\tau \sum_I J_I s_{I_1}(\tau)\cdots s_{I_q}(\tau)\right]
\end{align}
for
\begin{align}
    \cD^N \Omega &= \lim_{\Delta \tau \to 0} \lr{\prod_{r=1}^N \prod_{i=1}^{\beta/\Delta \tau} \frac{d\Omega_r^{(i)}}{2\pi} \lr{\cos \frac{\theta_r^{(i)}}{2} \cos \frac{\theta_r^{(i+1)}}{2} + e^{-i(\phi_r^{(i)}-\phi_r^{(i+1)})}\sin \frac{\theta_r^{(i)}}{2} \sin \frac{\theta_r^{(i+1)}}{2}}},
\end{align}
where $D^N \Omega$ is the path integral measure over all spin configurations $\Omega_r(\tau)$ (and absorbs our previous normalization factors of $2\pi$).

\subsection{Expectation over disorder}

Here, we show how to rewrite the replicated partition function in $G,\Sigma$ variables by taking an expectation over disorder.
Define $G_{ab}(\tau_1, \tau_2)$ and associated Lagrange multiplier $\Sigma_{ab}(\tau_1, \tau_2)$, with measures $\cD G$ and $\cD \Sigma$ given in the continuum limit over times $\tau$. Let $\cD^N h$ be the Gaussian measure relating auxiliary fields $h^a_{r\mu}(\tau_1), h^b_{s\nu}(\tau_2)$ by covariance $\delta_{rs}\delta_{\mu\nu}\Sigma_{ab}(\tau_1, \tau_2)$. We will obtain
\begin{align}
    \EE{Z^m} &= \int \cD^N \Omega\, \cD G\, \cD\Sigma\, \cD^N h\, \exp\Bigg\{\int_0^\beta \int_0^\beta d\tau_1 \,d\tau_2 \Bigg(\frac{J^2 N}{2q} \sum_{a,b=1}^m G_{ab}^q(\tau_1, \tau_2) - \frac{N}{2}\sum_{a,b=1}^m \Sigma_{ab}(\tau_1, \tau_2) G_{ab}(\tau_1, \tau_2)\Bigg) \nonumber \\
    &\qquad \qquad \qquad \qquad \qquad \quad  + \sum_{r,\mu,a} \int_0^\beta d\tau\, h^a_{r\mu}(\tau) s_{r\mu}^a(\tau)\Bigg\}.
\end{align}
In discretized form, this can be rewritten as
\begin{align}
    \EE{Z^m} &= \int_0^\pi \int_0^{2\pi} \left[\prod_{a=1}^m\prod_{r=1}^N \prod_{i=1}^{\beta/\Delta \tau} d\theta_r^{(i,a)}\, d\phi_r^{(i,a)}\, \frac{\sin \theta_r^{(i,a)}}{2\pi} \lr{\cos \frac{\theta_r^{(i,a)}}{2} \cos \frac{\theta_r^{(i+1,a)}}{2} + e^{-i(\phi_r^{(i,a)}-\phi_r^{(i+1,a)})}\sin \frac{\theta_r^{(i,a)}}{2} \sin \frac{\theta_r^{(i+1,a)}}{2}}\right]\nonumber \\
    &\quad \int_{-\infty}^\infty \lr{\prod_{a,b=1}^m \prod_{i,j=1}^{\beta/\Delta \tau} dG_{ab}(i, j)} \int_{-i\infty}^{i\infty} \lr{\prod_{a,b=1}^m \prod_{i, j=1}^{\beta/\Delta \tau} d\Sigma_{ab}(i, j)} \lr{2\pi \det(\Sigma)}^{-1/2} \int_{-\infty}^\infty \lr{\prod_{a,b=1}^m\prod_{i,j=1}^{\beta/\Delta\tau}\prod_{r,\mu}dh_{r\mu}^{(i,a)}}\nonumber\\
    &\quad \exp\left[(\Delta \tau)^2\sum_{i,j=1}^{\beta/\Delta\tau}\lr{\frac{J^2N}{2q}\sum_{a,b=1}^m G_{ab}^q(i,j) - \frac{N}{2}\sum_{a,b=1}^m\Sigma_{ab}(i,j)G_{ab}(i,j)}\right]\nonumber\\
    &\quad \exp\left[-\frac{1}{2}\sum_{i,j=1}^{\beta/\Delta\tau}\sum_{a=1}^m\sum_{r,\mu}h_{r\mu}^{(i,a)}\Sigma_{ab}^{-1}(i,j)h_{r\mu}^{(j,b)} + (\Delta\tau)\sum_{i=1}^{\beta/\Delta\tau}\sum_{a=1}^m\sum_{r,\mu} h_{r\mu}^{(i,a)} s_\mu(\theta_r^{(i,a)},\phi_r^{(i,a)})\right],
\end{align}
where $\Sigma_{ab}^{-1}(i,j)$ refers to the inverted matrix indexed by $(i, a)$ and $(j, b)$, so $\sum_{c=1}^m\sum_{k=1}^{\beta/\Delta \tau} \Delta\tau \Sigma_{ac}(i,k) \Sigma^{-1}_{cb}(k, j) = \delta_{ab}\delta_{ij}$.

The starting point of this computation is the replicated path integral over variables $s_{I_j}$,
\begin{align}
    \EE{Z^m} = \int \cD^N \Omega \, \EE{\exp\left[-\int_0^\beta d\tau \sum_I \sum_{a=1}^m J_I s_{I_1}^a(\tau)\cdots s_{I_q}^a(\tau)\right]}
\end{align}
and apply identity $\EE{e^{\sigma X}} = e^{\sigma^2/2}$ for $X \sim \cN(0,1)$ to obtain
\begin{align}
    \EE{Z^m} &= \int \cD^N \Omega \exp\left[\frac{\EE{J_I^2}}{2q!} \sum_I \lr{\sum_{a=1}^m \int_0^\beta d\tau\, s_{I_1}^a(\tau)\cdots s_{I_q}^a(\tau)}^2\right]\\
    &= \int \cD^N \Omega \exp\left[\frac{J^2}{2q N^{q-1}} \int_0^\beta \int_0^\beta d\tau_1 \,d\tau_2 \sum_{a,b=1}^m \lr{\sum_{r=1}^N \sum_{\mu\in\{x,y,z\}} s_{r\mu}^a(\tau_1) s_{r\mu}^b(\tau_2)}^q\right],
\end{align}
where in the first line we obtain factor $2q!$ (the 2 comes from the identity and the $q!$ from the orderings of $I$), and in the second line we evaluate the expectation and rearrange the terms in the innermost parentheses.
Recall that explicitly, this is (as $\Delta \tau \to 0$) the quantity
\begin{align}
    \EE{Z^m} &= \int \left[\prod_{a=1}^m\prod_{r=1}^N \prod_{i=1}^{\beta/\Delta \tau} d\theta_r^{(i,a)}\, d\phi_r^{(i,a)}\, \frac{\sin \theta_r^{(i,a)}}{2\pi} \lr{\cos \frac{\theta_r^{(i,a)}}{2} \cos \frac{\theta_r^{(i+1,a)}}{2} + e^{-i(\phi_r^{(i,a)}-\phi_r^{(i+1,a)})}\sin \frac{\theta_r^{(i,a)}}{2} \sin \frac{\theta_r^{(i+1,a)}}{2}}\right]\nonumber \\
    &\qquad \times \exp\left[\frac{(\Delta \tau)^2NJ^2}{2q} \sum_{i,j=1}^{\beta/\Delta \tau} \sum_{a,b=1}^m \lr{\frac{1}{N}\sum_{r=1}^N s_\mu(\theta_r^{(i,a)},\phi^{(i,a)}) s_\mu(\theta_r^{(j,b)},\phi^{(j,b)})}^q\right].
\end{align}
We substitute in
\begin{align}
    G_{ab}(\tau_1, \tau_2) &= \frac{1}{N}\sum_{r,\mu} s_{r\mu}^a(\tau_1) s_{r\mu}^b(\tau_2)
\end{align}
to obtain
\begin{align}
    \EE{Z^m} &= \int \cD^N \Omega \,\exp\left[\frac{J^2 N}{2q} \int_0^\beta \int_0^\beta d\tau_1 \,d\tau_2 \sum_{a,b=1}^m G_{ab}^q(\tau_1, \tau_2)\right].
\end{align}
We insert as a prefactor the identity
\begin{align}
    1 = \int \cD G \, \prod_{a,b=1}^m \prod_{\tau_1, \tau_2} \delta\lr{G_{ab}(\tau_1, \tau_2) - \frac{1}{N}\sum_{r,\mu} s_{r\mu}^a(\tau_1) s_{r\mu}^b(\tau_2)}
\end{align}
for
\begin{align}
    \cD G = \prod_{a,b=1}^m \prod_{\tau_1, \tau_2} dG_{ab}(\tau_1, \tau_2).
\end{align}
We replace the delta function with a Lagrange multiplier via $\delta(x) = \frac{1}{2\pi i}\int_{-i\infty}^{i\infty} d\lambda\, \exp[-\lambda x]$ with $\lambda = \frac{N}{2}\Sigma_{ab}(\tau_1, \tau_2)$ to obtain
\begin{align}
    \EE{Z^m} &= \int \cD^N \Omega\, \cD G\, \cD\Sigma\, \exp\Bigg\{\int_0^\beta \int_0^\beta d\tau_1 \,d\tau_2 \Bigg(\frac{J^2 N}{2q} \sum_{a,b=1}^m G_{ab}^q(\tau_1, \tau_2) \nonumber \\
    &\qquad \qquad \qquad \qquad \qquad \quad - \frac{N}{2}\sum_{a,b=1}^m \Sigma_{ab}(\tau_1, \tau_2) \left[G_{ab}(\tau_1, \tau_2) - \frac{1}{N} \sum_{r,\mu} s_{r\mu}^a(\tau_1) s_{r\mu}^b(\tau_2)\right]\Bigg)\Bigg\}.
\end{align}
Finally, we introduce an auxiliary field to eliminate the quadratic term with the Hubbard-Stratonovich transformation $\exp[\frac{1}{2}x^T M x] = \lr{1/\sqrt{(2\pi)^n\det(M)}}\int d^n h \exp[-\frac{1}{2}h^T M^{-1} h + h^T x]$. Let $\cD^N h$ denote the Gaussian measure relating $h_{r\mu}^a(\tau_1), h_{s\nu}^b(\tau_2)$ by covariance $\delta_{rs}\delta_{\mu\nu}\Sigma_{ab}(\tau_1, \tau_2)$. Then
\begin{align}
    \EE{Z^m} &= \int \cD^N \Omega\, \cD G\, \cD\Sigma\, \cD^N h\, \exp\Bigg\{\int_0^\beta \int_0^\beta d\tau_1 \,d\tau_2 \Bigg(\frac{J^2 N}{2q} \sum_{a,b=1}^m G_{ab}^q(\tau_1, \tau_2) - \frac{N}{2}\sum_{a,b=1}^m \Sigma_{ab}(\tau_1, \tau_2) G_{ab}(\tau_1, \tau_2)\Bigg) \nonumber \\
    &\qquad \qquad \qquad \qquad \qquad \qquad \quad  + \sum_{r,\mu,a} \int_0^\beta d\tau\, h_{r\mu}^a(\tau) s_{r\mu}^a(\tau)\Bigg\}.
\end{align}

\subsection{Saddle point equations}
Here, we show how to obtain the saddle point equations
\begin{align}
    J^2 G_{ab}(\tau_1,\tau_2)^{q-1} = \Sigma_{ab}(\tau_1, \tau_2), \qquad G_{ab}(\tau_1, \tau_2) = \left\langle\sum_\mu \sigma_\mu^a(\tau_1)\sigma_\mu^b(\tau_2) \right\rangle_\Sigma.
\end{align}
We can write the second saddle point equation more explicitly as
\begin{align}
    G_{ab}(\tau_1,\tau_2) = \frac{\Tr\left\langle\cT\,\exp\left[\int_0^\beta d\tau \sum_{a=1}^m \sum_\mu h_\mu^{(a)}(\tau) \sigma_\mu^{(a)}\right] \sum_\mu \sigma_\mu^{(a)}(\tau_1)\sigma_\mu^{(b)}(\tau_2)\right\rangle_{\delta_{\mu\nu}\Sigma_{ab}(\tau_1,\tau_2)}}{\Tr\left\langle\cT\,\exp\left[\int_0^\beta d\tau \sum_{a=1}^m \sum_\mu h_\mu^{(a)}(\tau) \sigma_\mu^{(a)}(\tau)\right]\right\rangle_{\delta_{\mu\nu}\Sigma_{ab}(\tau_1,\tau_2)}},
\end{align}
where the expectation is over sampling Gaussian $h$ with mean zero and covariance $\langle h_\mu^{(a)}(\tau_1) h_\nu^{(b)}(\tau_2)\rangle = \delta_{\mu\nu}\Sigma_{ab}(\tau_1,\tau_2)$. Here, the times $\tau_1, \tau_2$ in $\sigma_\mu^{(a)}(\tau_1), \sigma_\mu^{(b)}(\tau_2)$ are just labels for the time-ordering $\cT$: for example, in the discretized form with $\tau_1 = i\Delta \tau$ and $\tau_2 = j \Delta\tau$ with $j > i$, the numerator is the trace of the expectation of
\begin{align}
    \sum_\nu \lr{\prod_{k=\beta/\Delta\tau}^{j} \exp\left[\Delta \tau \sum_{a,\mu} h_\mu^{(k,a)}\sigma_\mu^{(a)}\right]} \sigma_\nu^{(b)} \lr{\prod_{k=j-1}^i \exp\left[\Delta \tau \sum_{a,\mu} h_\mu^{(k,a)}\sigma_\mu^{(a)}\right]} \sigma_\nu^{(a)} \lr{\prod_{k=i-1}^{1} \exp\left[\Delta \tau \sum_{a,\mu} h_\mu^{(k,a)}\sigma_\mu^{(a)}\right]}.
\end{align}
We now derive these equations.
The derivative with respect to $G_{ab}(i,j)$ is trivial: setting it to zero gives the saddle point equation
\begin{align}
    J^2 G_{ab}^{q-1}(i,j) = \Sigma_{ab}(i,j).
\end{align}
The derivative with respect to $\Sigma_{ab}(i,j)$ is slightly trickier. We first observe that the $\theta,\phi$ dependence now only occurs in the final $s_\mu$ function. Hence, we can swap the order of integration and evaluate
\begin{align}
    &\int_0^\pi \int_0^{2\pi} \left[\prod_{a=1}^m\prod_{r=1}^N \prod_{i=1}^{\beta/\Delta \tau} d\theta_r^{(i,a)}\, d\phi_r^{(i,a)}\, \frac{\sin \theta_r^{(i,a)}}{2\pi}\bra{\Omega_r^{(i,a)}}\ket{\Omega_r^{(i+1,a)}}\right] \exp\left[(\Delta\tau)\sum_{i=1}^{\beta/\Delta\tau}\sum_{a=1}^m\sum_{r,\mu} h_{r\mu}^{(i,a)} s_\mu(\theta_r^{(i,a)},\phi_r^{(i,a)})\right]\\
    &= \Tr(\prod_{i=1}^{\beta/\Delta\tau}\exp\left[(\Delta\tau)\sum_{a=1}^m \sum_{r,\mu} h_{r\mu}^{(i,a)} \sigma_{r\mu}^{(i,a)}\right]),
\end{align}
where $\sigma_{r\mu}^{(i,a)}$ indicates a Pauli of type $\mu$ acting on a qubit indexed by $(r, a)$ at time slice $i$. Integrating over $h$ is equivalent taking the expectation when sampling $h_{r\mu}^{(i,a)}$ with covariance $\delta_{rs}\delta_{\mu\nu}\Sigma_{ab}(i,j)$. Hence, we have
\begin{align}
    \EE{Z^m} &= \int_{-\infty}^\infty \lr{\prod_{a,b=1}^m \prod_{i,j=1}^{\beta/\Delta \tau} dG_{ab}(i, j)} \int_{-i\infty}^{i\infty} \lr{\prod_{a,b=1}^m \prod_{i, j=1}^{\beta/\Delta \tau} d\Sigma_{ab}(i, j)} \nonumber\\
    &\quad \exp\left[(\Delta \tau)^2\sum_{i,j=1}^{\beta/\Delta\tau}\lr{\frac{J^2N}{2q}\sum_{a,b=1}^m G_{ab}^q(i,j) - \frac{N}{2}\sum_{a,b=1}^m\Sigma_{ab}(i,j)G_{ab}(i,j)}\right] \nonumber \\
    &\quad \left\langle\Tr(\prod_{i=1}^{\beta/\Delta\tau}\exp\left[(\Delta\tau)\sum_{a=1}^m \sum_{r,\mu} h_{r\mu}^{(i,a)} \sigma_{r\mu}^{(i,a)}\right])\right\rangle_{\delta_{rs}\delta_{\mu\nu}\Sigma_{ab}(i,j)}.
\end{align}
Since the $N$ sites indexed by $r$ are independent, we can further rewrite this in terms of single-qubit operators $\sigma_\mu$ and fields $h_\mu$:
\begin{align}
    \EE{Z^m} &= \int_{-\infty}^\infty \lr{\prod_{a,b=1}^m \prod_{i,j=1}^{\beta/\Delta \tau} dG_{ab}(i, j)} \int_{-i\infty}^{i\infty} \lr{\prod_{a,b=1}^m \prod_{i, j=1}^{\beta/\Delta \tau} d\Sigma_{ab}(i, j)} \nonumber\\
    &\quad \exp\Bigg\{(\Delta \tau)^2\sum_{i,j=1}^{\beta/\Delta\tau}\lr{\frac{J^2N}{2q}\sum_{a,b=1}^m G_{ab}^q(i,j) - \frac{N}{2}\sum_{a,b=1}^m\Sigma_{ab}(i,j)G_{ab}(i,j)} \nonumber \\
    &\qquad\quad + N\log \left\langle\Tr(\prod_{i=1}^{\beta/\Delta\tau}\exp\left[(\Delta\tau)\sum_{a=1}^m \sum_{\mu} h_\mu^{(i,a)} \sigma_\mu^{(i,a)}\right])\right\rangle_{\delta_{\mu\nu}\Sigma_{ab}(i,j)}\Bigg\}.
\end{align}
In the limit $\Delta\tau \to 0$, we write this as
\begin{align}
    \EE{Z^m} &= \int \cD G \, \cD \Sigma\, \exp\Bigg[\int_0^\beta\int_0^\beta d\tau_1\,d\tau_2\, \lr{\frac{J^2N}{2q} G_{ab}^q(\tau_1, \tau_2) - \frac{N}{2}\sum_{a,b=1}^m \Sigma_{ab}(\tau_1, \tau_2) G_{ab}(\tau_1, \tau_2)} \nonumber\\
    &\qquad + N \log \left\langle \Tr(\exp\left[\int_0^\beta d\tau \, \sum_{a=1}^m \sum_\mu h_\mu^{(a)}(\tau) \sigma_\mu^{(a)}(\tau)\right]) \right\rangle_{\delta_{\mu\nu}\Sigma_{ab}(\tau_1,\tau_2)}\Bigg].
\end{align}
We can undo the insertion of the auxiliary field $h$ from Hubbard-Stratonovich, giving
\begin{align}
    \EE{Z^m} &= \int \cD G \, \cD \Sigma\, \exp\Bigg[\int_0^\beta\int_0^\beta d\tau_1\,d\tau_2\, \lr{\frac{J^2N}{2q} G_{ab}^q(\tau_1, \tau_2) - \frac{N}{2}\sum_{a,b=1}^m \Sigma_{ab}(\tau_1, \tau_2) G_{ab}(\tau_1, \tau_2)} \nonumber\\
    &\qquad + N\log \Tr(\cT\,\exp\left[\frac{1}{2}\int_0^\beta \int_0^\beta d\tau_1\, d\tau_2\, \sum_{a,b=1}^m \sum_\mu \sigma_\mu^{(a)}(\tau_1)\Sigma_{ab}(\tau_1, \tau_2)\sigma_\mu^{(b)}(\tau_2)\right])\Bigg],
\end{align}
where we abuse notation and use the time-ordering operator $\cT$ to remind ourselves that the auxiliary field must be time-ordered.
Duhamel and cyclicity of trace gives derivative
\begin{align}
    \frac{\partial}{\partial \Sigma_{ab}(i,j)} \log \Tr(\exp[X]) = \frac{N(\Delta \tau)^2}{2}\frac{\Tr(\cT\,\exp\left[X\right]\sum_\mu \sigma_\mu^{(a)}(\tau_1)\sigma_\mu^{(b)}(\tau_2))}{\Tr(\cT\,\exp[X])}
\end{align}
for
\begin{align}
    X = \frac{(\Delta \tau)^2}{2}\sum_{i,j=1}^{\beta/\Delta\tau} \sum_{a,b=1}^m \sum_\mu \sigma_\mu^{(i,a)}\Sigma_{ab}(i, j)\sigma_\mu^{(j,b)}.
\end{align}
Hence, the second saddle point equation is
\begin{align}
    -\frac{N(\Delta\tau)^2}{2} G_{ab}(i,j) + \frac{N(\Delta \tau)^2}{2}\frac{\Tr(\cT\,\exp\left[X\right]\sum_\mu \sigma_\mu^{(i,a)}\sigma_\mu^{(j,b)})}{\Tr(\cT\,\exp[X])} = 0
\end{align}
giving
\begin{align}
    G_{ab}(i,j) = \frac{\Tr(\cT\,\exp\left[X\right]\sum_\mu \sigma_\mu^{(i,a)}\sigma_\mu^{(j,b)})}{\Tr(\cT\,\exp[X])}.
\end{align}
To write the time-ordering cleanly in the discretized form, we return to the previous representation with auxiliary fields $h$.

\section{RS solution of the \texorpdfstring{$q$}\ -local Hamiltonian ensemble}\label{app:rs}

We describe here the numerical procedure used to obtain \Cref{fig:rs}; the analytical details are presented in the main text. Our code is implemented using CTSEG~\cite{kavokine2025ctseg} and is adapted from the recent solution of the quantum Heisenberg model~\cite{kavokine2024exact}. We initialize $\Sigma(\tau) = J^2 G(\tau)^{q-1}$ using the Liouville solution of~\cite{swingle2023bosonicmodelquantumholography}, which approximates the large-$q$ solution of the saddle point equations by expanding $G(\tau)$ is around 3, i.e.,
\begin{align}
    G = 3\lr{1 + \frac{g(\tau)}{q} + \cdots}.
\end{align}
They resum the series as $3e^{g/q}$ to obtain approximation
\begin{align}
    G(\tau) = b \lr{\frac{\pi}{\beta \sin \frac{\pi \tau}{\beta}}}^{2/q}, \qquad b = \frac{1}{\pi}\lr{\frac{1}{2}-\frac{1}{q}}\tan \frac{\pi}{q}.
\end{align}
In~\cite{swingle2023bosonicmodelquantumholography}, it was observed that initializing numerical solvers from this solution converge to the quenched RS $G(\tau)$ as opposed to the annealed saddle $G(\tau)\to 1$. In the RS phase of \Cref{fig:rs}, we similarly find that $\Sigma(\beta/2) \sim 0.6$ (with $J=1$), consistent with seeing the quenched as opposed to the annealed saddle. We discretized $\Sigma(\tau)$ into $N_\tau = 2000$ timesteps and sampled the external field $N_h = 2000$ times; the remaining parameters are similar to those of~\cite{kavokine2024exact}. We note that the RS solver (with a generic off-diagonal element) is fairly unstable for the $p=3$ model; especially around the transition temperature shown in \Cref{fig:rs}, its output is sensitive to the initialization and choice of hyperparameters. This is likely because the glass transition occurs at significantly smaller $\beta$, and hence the RS solver is faced with multiple unstable saddles. In particular, the RS ansatz's off-diagonal term would only be expected to become nonzero when the size of each 1RSB cluster is large, although the glass transition occurs when they are exponentially small. The RS solution is thus likely unstable for a reasonable range of $\beta$ below the transition shown in \Cref{fig:rs}.

\section{1RSB solution of the \texorpdfstring{$q$}\ -local Hamiltonian ensemble}\label{app:1rsb}

We give the 1RSB saddle point equations of the random $q$-local Hamiltonian ensemble at arbitrary $q$. Unlike the finite-$N$ numerics of that probe glassiness at small values of $q$ in~\cite{swingle2023bosonicmodelquantumholography}, our computation takes the $N\to\infty$ limit analytically. The saddle point equations are given as integrals of single-qubit Hamiltonians, which can be directly evaluated numerically. In this appendix, we show that sufficiently large $q$ is not glassy.

\subsection{Setup}
We start from the $G,\Sigma$ action \eqref{eq:action} derived by~\cite{swingle2023bosonicmodelquantumholography} (and derived in a self-contained manner in \Cref{app:path})
\begin{align}
    X &= \int_0^\beta\int_0^\beta d\tau_1\, d\tau_2\, \lr{\frac{J^2}{2q} \sum_{a,b=1}^{n}G_{ab}^q(\tau_1,\tau_2) - \frac{1}{2} \sum_{a,b=1}^{n} \Sigma_{ab}(\tau_1,\tau_2) G_{ab}(\tau_1, \tau_2)}\nonumber\\
    &\quad + \log \mathbb{E}_h \int \cD\Omega\, \exp[\sum_{a=1}^{n} \sum_\mu \int_0^\beta d\tau\, h_\mu^a(\tau) s_\mu(\Omega^a(\tau))].
\end{align}
Recall that the expectation over $h$ is with respect to $h_\mu^a(\tau_1), h_\nu^b(\tau_2)$ drawn from a Gaussian with mean zero and covariance $\delta_{\mu\nu}\Sigma_{ab}(\tau_1,\tau_2)$. Taking derivatives with respect to $G$ and $\Sigma$, we obtain saddle point equations 
\begin{align}
    J^2 G_{ab}(\tau_1,\tau_2)^{q-1} = \Sigma_{ab}(\tau_1, \tau_2), \qquad G_{ab}(\tau_1, \tau_2) = \left\langle\sum_\mu \sigma_\mu^a(\tau_1)\sigma_\mu^b(\tau_2) \right\rangle_\Sigma.
\end{align}
Under the 1RSB ansatz with $q_0 = 0$, i.e.,
\begin{align}
    G_{ab}(\tau_1, \tau_2) &= G(\tau_1-\tau_2)\delta_{ab} + q_1\lr{1-\delta_{a_0b_0}}\delta_{a_1b_1}\\
    \Sigma_{ab}(\tau_1, \tau_2) &= \Sigma(\tau_1-\tau_2)\delta_{ab} + \hat q_1\lr{1-\delta_{a_0b_0}}\delta_{a_1b_1}\label{eq:sigma1rsb}
\end{align}
for $a_0 \in [m], a_1 \in [n/m]$ and $a=(a_0, a_1)$, the action becomes
\begin{align}\label{eq:Xnew}
    X &= \frac{n}{2}\int_0^\beta\int_0^\beta d\tau_1\,d\tau_2\, \lr{\frac{J^2G(\tau_1-\tau_2)^q}{q} - \Sigma(\tau_1-\tau_2)G(\tau_1-\tau_2)} + \frac{\beta^2(m-1)n}{2}\lr{\frac{J^2q_1^q}{q} - \hat q_1 q_1}\nonumber \\
    &\quad + \log \mathbb{E}_h \int  \cD\Omega\, \exp[\sum_{a=1}^{n} \sum_\mu \int_0^\beta d\tau\, h_\mu^a(\tau) s_\mu(\Omega^a(\tau))].
\end{align}
and the saddle point equations with respect to $G$ are
\begin{align}\label{eq:gsaddle}
    J^2 G(\tau_1-\tau_2)^{q-1} = \Sigma(\tau_1-\tau_2), \qquad  J^2 q_1^{q-1} = \hat q_1, \qquad q_1 = \left\langle \sum_\mu \sigma_\mu^1(\tau_1) \sigma_\mu^2(\tau_2) \right\rangle_{\Sigma'},
\end{align}
and with respect to $\Sigma$ are
\begin{align}\label{eq:dsigma}
    G(\tau_1-\tau_2) = \left\langle \sum_\mu \sigma_\mu^1(\tau_1) \sigma_\mu^1(\tau_2) \right\rangle_{\Sigma'}, \qquad q_1 = \left\langle \sum_\mu \sigma_\mu^1(\tau_1) \sigma_\mu^2(\tau_2) \right\rangle_{\Sigma'}
\end{align}
where $\Sigma'$ is the covariance matrix in one 1RSB block, i.e., for $a, b \in [m]$,
\begin{align}
    \Sigma'_{ab}(\tau_1,\tau_2) = \Sigma(\tau_1-\tau_2)\delta_{ab} + \hat q_1 (1-\delta_{ab}).
\end{align}
Below, we will also evaluate the saddle point equation in $\partial_m X$. For all these three saddle point equations, we need to take an analytic continuation to real $m$ after sending $n \to 0$.

\subsection{Saddle point equations}
The following notation will be useful. In each 1RSB block, we decompose the fields $h_\mu^a(\tau)$ (in the $\log \E_h$ term obtained by Hubbard-Stratonovich in \eqref{eq:Xnew}) into the sum of two independent random variables
\begin{align}\label{eq:randvars}
    h_\mu^a(\tau) = z_\mu + \xi_\mu^a(\tau), \qquad z_\mu \sim \cN(0, \delta_{\mu\nu}\hat q_1), \qquad \xi_\mu^a(\tau) \sim \cN(0, \delta_{\mu\nu}(\Sigma(\tau_1-\tau_2) - \hat q_1)),
\end{align}
which recover the proper covariance:
\begin{align}
    \E_{h\sim\Sigma'}[h_\mu^a(\tau_1) h_\nu^b(\tau_2)] = \EE{z_\mu z_\nu} + \EE{\xi_\mu^a(\tau_1) \xi_\nu^b(\tau_2)} = \delta_{\mu\nu}\lr{\hat q_1 + \delta_{ab} (\Sigma(\tau_1-\tau_2) - \hat q_1)} = \delta_{\mu\nu}\Sigma'_{ab}(\tau_1,\tau_2).
\end{align}
We also define the quantity
\begin{align}\label{eq:zeta}
    \zeta(z) \coloneqq \E_{\xi_\mu(\tau)} \int \cD\Omega\, \exp[\sum_\mu \int_0^\beta d\tau\, (z_\mu + \xi_\mu(\tau)) s_\mu(\Omega(\tau))]
\end{align} 
where here the exponent only contains Pauli matrices acting on a single qubit. This notation allows us to rewrite the log expectation term from Hubbard-Stratonovich in a more convenient form.

\begin{lemma}[Auxiliary field expectation]\label{lem:expzeta}
    \begin{align}
        \mathbb{E}_{h\sim\Sigma'} \int \cD\Omega\, \exp[\sum_{a=1}^m \sum_\mu \int_0^\beta d\tau\, h_\mu^a(\tau) s_\mu(\Omega^a(\tau))] = \E_z\left[\zeta(z)^m\right].
    \end{align}
\end{lemma}
\begin{proof}
    We evaluate
    \begin{align}
        &\mathbb{E}_{h\sim\Sigma'} \int \cD\Omega\, \exp[\sum_{a=1}^m \sum_\mu \int_0^\beta d\tau\, h_\mu^a(\tau) s_\mu(\Omega^a(\tau))]\nonumber\\
        &= \E_z \left[\E_{\xi^1, \dots, \xi^m}\left[ \int \cD\Omega\, \prod_{a=1}^m \exp[\sum_\mu \int_0^\beta d\tau\, (z_\mu + \xi_\mu^a(\tau)) s_\mu(\Omega^a(\tau))] \,\Bigg|\, z\right]\right]\\
        &= \E_z \left[\lr{\E_{\xi_\mu(\tau)}\left[ \int \cD\Omega\, \exp[\sum_\mu \int_0^\beta d\tau\, (z_\mu + \xi_\mu(\tau)) s_\mu(\Omega(\tau))] \,\Bigg|\, z\right]}^m\right]\\
        &= \E_z\left[\zeta(z)^m\right].
    \end{align}
\end{proof}

The saddle point equations from $\partial_G X = 0$, i.e.,
\begin{align}
    J^2 G(\tau_1-\tau_2)^{q-1} = \Sigma(\tau_1-\tau_2), \qquad  J^2 q_1^{q-1} = \hat q_1,
\end{align}
are easy to work with independently of $m$. We first rewrite $\partial_\Sigma X = 0$ saddle point equation for the off-diagonal terms. Although we derive it from an integer number of copies in \eqref{eq:dsigma}, we apply it to any real $m$ in the replica method.
\begin{lemma}[Off-diagonal $\partial_\Sigma X$ saddle point equation]
    The off-diagonal equation in \eqref{eq:dsigma} is equivalent to
    \begin{align}
        q_1 = \frac{\E_z\left[\zeta(z)^m \sum_\nu a_\nu(z;\tau_1) a_\nu(z;\tau_2)\right]}{\E_z\left[\zeta(z)^m\right]}
    \end{align}
    for $\zeta$ defined in \eqref{eq:zeta}, random variables $z, \xi$ defined in \eqref{eq:randvars}, and
    \begin{align}
        a_\nu(z;\tau) \coloneqq \frac{\E_\xi \Tr(\cT \sigma_\nu(\tau) \exp[\sum_\mu \int_0^\beta d\tau'\, (z_\mu+\xi_\mu(\tau')) \sigma_\mu(\tau')])}{\zeta(z)}.
    \end{align}
\end{lemma}
\begin{proof}
    We start with
    \begin{align}
        \left\langle \sum_\mu \sigma_\mu^1(\tau_1) \sigma_\mu^2(\tau_2) \right\rangle_{\Sigma'} =  \frac{\E_h \Tr(\cT\, \sum_\nu \sigma_\nu^1(\tau_1) \sigma_\nu^2(\tau_2) \exp[\sum_a \sum_\mu \int_0^\beta d\tau\, h_\mu^a(\tau) \sigma_\mu^a(\tau)])}{\mathbb{E}_h \Tr(\cT\, \exp[\sum_a \sum_\mu \int_0^\beta d\tau\, h_\mu^a(\tau) \sigma_\mu^a(\tau)])}
    \end{align}
    and observe that the denominator is $\E_z[\zeta(z)^m]$ by \Cref{lem:expzeta}. The numerator is computed via
    \begin{align}
        &\E_h \Tr(\cT\, \sum_\nu \sigma_\nu^1(\tau_1) \sigma_\nu^2(\tau_2) \exp[\sum_a \sum_\mu \int_0^\beta d\tau\, h_\mu^a(\tau) \sigma_\mu^a(\tau)]) \nonumber\\
        &= \E_z \Bigg[\sum_\nu \E_{\xi^1} \left[\Tr(\cT\, \sigma_\nu^1(\tau_1) \exp[\sum_\mu \int_0^\beta d\tau\, (z_\mu + \xi_\mu^1(\tau)) \sigma_\mu^1(\tau)])\,\Bigg|\,z\right]\nonumber\\
        &\qquad \times \E_{\xi^2} \left[\Tr(\cT\, \sigma_\nu^2(\tau_1) \exp[\sum_\mu \int_0^\beta d\tau\, (z_\mu+\xi_\mu^2(\tau)) \sigma_\mu^2(\tau)])\,\Bigg|\,z\right]\nonumber\\
        &\qquad \times \prod_{a=3}^m \E_{\xi^a} \left[\Tr(\cT\, \exp[\sum_\mu \int_0^\beta d\tau\, (z_\mu+\xi_\mu^a(\tau)) \sigma_\mu^a(\tau)])\,\Bigg|\,z\right]\Bigg]\\
        &= \E_z\left[\zeta(z)^m \sum_\nu a_\nu(z;\tau_1) a_\nu(z;\tau_2)\right]
    \end{align}
    for $a_\nu$ as defined in the lemma statement.
\end{proof}

\begin{remark}[Comparison to Ising $q$-spin]
    The Ising $q$-spin model, under the standard normalization with coefficient variance $J^2 q!/2N^{q-1}$, has saddle point equation~\cite{montanari2003nature,parisi2012}
    \begin{align}
        q_1 =  \frac{\EE{2^m \cosh^m \lr{\beta J\sqrt{qq_1^{q-1}} Z} \tanh^2 \lr{\beta J \sqrt{qq_1^{q-1}} Z}}}{\EE{2^m \cosh^m \lr{\beta J\sqrt{q q_1^{q-1}} Z}}}
    \end{align}
    for standard normal $Z$.
\end{remark}

In the quantum setting, the diagonal term $G(\tau), \Sigma(\tau)$ has nontrivial time dependence. We also require its saddle point equation.

\begin{lemma}[Diagonal $\partial_\Sigma X$ saddle point equation]
    For any real $m$, the diagonal equation in \eqref{eq:dsigma} is analytically continued to
    \begin{align}
        G(\tau_1-\tau_2) = \frac{\E_z\left[\zeta(z)^m \sum_\nu a_\nu(z;\tau_1,\tau_2)\right]}{\E_z\left[\zeta(z)^m\right]}
    \end{align}
    for $\zeta$ defined in \eqref{eq:zeta}, random variables $z, \xi$ defined in \eqref{eq:randvars}, and
    \begin{align}
        a_\nu(z;\tau_1, \tau_2) = \frac{\E_\xi \Tr(\cT \sigma_\nu(\tau_1) \sigma_\nu(\tau_2) \exp[\sum_\mu \int_0^\beta d\tau\, (z_\mu+\xi_\mu(\tau)) \sigma_\mu(\tau)])}{\zeta(z)}.
    \end{align}
\end{lemma}
\begin{proof}
    We start with
    \begin{align}
        \left\langle \sum_\mu \sigma_\mu^1(\tau_1) \sigma_\mu^1(\tau_2) \right\rangle_{\Sigma'} =  \frac{\E_h \Tr(\cT\, \sum_\nu \sigma_\nu^1(\tau_1) \sigma_\nu^1(\tau_2) \exp[\sum_a \sum_\mu \int_0^\beta d\tau\, h_\mu^a(\tau) \sigma_\mu^a(\tau)])}{\mathbb{E}_h \Tr(\cT\, \exp[\sum_a \sum_\mu \int_0^\beta d\tau\, h_\mu^a(\tau) \sigma_\mu^a(\tau)])}
    \end{align}
    and observe that the denominator is $\E_z[\zeta(z)^m]$ by \Cref{lem:expzeta}. The numerator is computed via
    \begin{align}
        &\E_h \Tr(\cT\, \sum_\nu \sigma_\nu^1(\tau_1) \sigma_\nu^1(\tau_2) \exp[\sum_a \sum_\mu \int_0^\beta d\tau\, h_\mu^a(\tau) \sigma_\mu^a(\tau)]) \nonumber\\
        &= \E_z \Bigg[\sum_\nu \E_{\xi^1} \left[\Tr(\cT\, \sigma_\nu^1(\tau_1) \sigma_\nu^1(\tau_2) \exp[\sum_\mu \int_0^\beta d\tau\, (z_\mu + \xi_\mu^1(\tau)) \sigma_\mu^1(\tau)])\,\Bigg|\,z\right]\nonumber\\
        &\qquad \times \prod_{a=3}^m \E_{\xi^a} \left[\Tr(\cT\, \exp[\sum_\mu \int_0^\beta d\tau\, (z_\mu+\xi_\mu^a(\tau)) \sigma_\mu^a(\tau)])\,\Bigg|\,z\right]\Bigg]\\
        &= \E_z\left[\zeta(z)^m \sum_\nu a_\nu(z;\tau_1,\tau_2)\right]
    \end{align}
    for $a_\nu$ as defined in the lemma statement.
\end{proof}

We next need to identify the appropriate block size $m$ by evaluating the saddle point equation $\partial_m X = 0$; equivalently, this sets the TAP complexity $S=0$. (For convenience, we simplify our final answer below using $\hat q_1 = J^2 q_1^{q-1}$ from the $\partial_G X = 0$ saddle point equation.)

\begin{lemma}[$\partial_m X = 0$ saddle point equation]\label{lem:tapcomplexity}
    The saddle point equation $\partial_m X = 0$ is equivalent to
    \begin{align}
        \frac{(\beta J)^2}{2} \lr{1 - \frac{1}{q}} q_1^{q} + \frac{1}{m^2}\log \E_z \zeta(z)^m - \frac{1}{m}\frac{\E_z \zeta(z)^m \log \zeta(z)^m}{\E_z \zeta(z)^m} = 0
    \end{align}
    for $\zeta$ defined in \eqref{eq:zeta} and random variable $z = (z_\mu)_{\mu \in \{X,Y,Z\}} \sim \cN(0, \delta_{\mu\nu}\hat q_1)$ defined in \eqref{eq:randvars}.
\end{lemma}
\begin{proof}
Since we can decompose
\begin{align}
    \mathbb{E}_{h\sim\Sigma} \int  \cD\Omega\, \exp[\sum_{a=1}^{n} \sum_\mu \int_0^\beta d\tau\, h_\mu^a(\tau) s_\mu(\Omega^a(\tau))] = \lr{\mathbb{E}_{h\sim\Sigma'} \int \cD\Omega\, \prod_{a=1}^m \exp[\sum_\mu \int_0^\beta d\tau\, h_\mu^a(\tau) s_\mu(\Omega^a(\tau))]}^{n/m},
\end{align}
the only nontrivial quantity to compute is $\partial_m \log \zeta_0$ for (by \Cref{lem:expzeta})
\begin{align}
    \zeta_0 = \mathbb{E}_{h\sim\Sigma'} \int \cD\Omega\, \prod_{a=1}^m \exp[\sum_\mu \int_0^\beta d\tau\, h_\mu^a(\tau) s_\mu(\Omega^a(\tau))] = \E_z\left[\zeta(z)^m\right]
\end{align}
since
\begin{align}
    \partial_m \log \mathbb{E}_{h\sim\Sigma} \int  \cD\Omega\, \exp[\sum_{a=1}^{mn} \sum_\mu \int_0^\beta d\tau\, h_\mu^a(\tau) s_\mu(\Omega^a(\tau))] = \partial_m \frac{n}{m} \log \zeta_0 = -\frac{n}{m^2} \log \zeta_0 + \frac{n}{m} \partial_m \log \zeta_0.
\end{align}
The derivative
\begin{align}
    \partial_m \zeta_0 = \E_z \left[\zeta(z)^m \log \zeta(z)\right]
\end{align}
yields
\begin{align}
    \lim_{n\to 0}\partial_n \partial_m \log \mathbb{E}_h \int  \cD\Omega\, \exp[\sum_{a=1}^{n} \sum_\mu \int_0^\beta d\tau\, h_\mu^a(\tau) s_\mu(\Omega^a(\tau))] = -\frac{1}{m^2} \log \E_z \zeta(z)^m + \frac{1}{m}\frac{\E_z \zeta(z)^m \log \zeta(z)}{\E_z \zeta(z)^m}.
\end{align}
To conclude the computation, we see that the action satisfies
\begin{align}
    \lim_{n\to 0}\partial_n \partial_m X = \partial_m \frac{\beta^2 (m-1)}{2}\lr{\frac{J^2 q_1^q}{q} - \hat q_1 q_1} + \lim_{n\to 0}\partial_n \partial_m \log \mathbb{E}_h \int  \cD\Omega\, \exp[\sum_{a=1}^{n} \sum_\mu \int_0^\beta d\tau\, h_\mu^a(\tau) s_\mu(\Omega^a(\tau))].
\end{align}
\end{proof}

\begin{remark}[Comparison to Ising $q$-spin]
    The Ising $q$-spin model, under the standard normalization with coefficient variance $J^2 q!/2N^{q-1}$, has $\partial_m X = 0$ saddle point equation~\cite{montanari2003nature,parisi2012}
    \begin{align}
        0 = \frac{(\beta J)^2}{2} q_1^q(q-1) + \frac{\log\EE{2^m \cosh^m \lr{\beta \sqrt{q q_1^{q-1}} Z}}}{m^2} - \frac{1}{m}\frac{\EE{2^m\cosh^m\lr{\beta J \sqrt{q q_1^{q-1}} Z} \log \lr{2^m\cosh^m\lr{\beta J \sqrt{q q_1^{q-1}} Z}}}}{\EE{2^m\cosh^m\lr{\beta J \sqrt{q q_1^{q-1}}}}}
    \end{align}
    for standard normal $Z$.
\end{remark}

\subsection{Solving the saddle point equations}
We summarize our saddle point equations that must be solved. We will solve them numerically at $q=3$ due to the lack of a closed form for the integrals over Pauli matrices (similar to the numerical solution of~\cite{kavokine2024exact} for the quantum Heisenberg model). For random variables
\begin{align}
    z_\mu \sim \cN(0, \delta_{\mu\nu}\hat q_1), \qquad \xi_\mu(\tau) \sim \cN(0, \delta_{\mu\nu}(\Sigma(\tau_1-\tau_2) - \hat q_1)),
\end{align}
we defined quantities
\begin{align}
    \zeta(z) &= \E_{\xi_\mu(\tau)} \Tr( \cT \exp[\sum_\mu \int_0^\beta d\tau\, (z_\mu + \xi_\mu(\tau)) \sigma_\mu(\tau)] )\\
    a_\nu(z;\tau_1, \dots, \tau_k) &= \frac{\E_{\xi_\mu(\tau)} \Tr(\cT \sigma_\nu(\tau_1)\cdots\sigma_\nu(\tau_k) \exp[\sum_\mu \int_0^\beta d\tau\, (z_\mu+\xi_\mu(\tau)) \sigma_\mu(\tau)])}{\zeta(z)}.
\end{align}
The saddle point equations are
\begin{align}
    \Sigma(\tau_1 - \tau_2) &= J^2 G(\tau_1 - \tau_2)^{q-1}\\
    \hat q_1 &= J^2 q_1^{q-1}\\
    G(\tau_1-\tau_2) &= \frac{\E_z\left[\zeta(z)^m \sum_\nu a_\nu(z;\tau_1,\tau_2)\right]}{\E_z\left[\zeta(z)^m\right]}\\
    q_1 &= \frac{\E_z\left[\zeta(z)^m \sum_\nu a_\nu(z;\tau_1) a_\nu(z;\tau_2)\right]}{\E_z\left[\zeta(z)^m\right]}\\
    0 &= \frac{(\beta J)^2}{2} \lr{1 - \frac{1}{q}} q_1^{q} + \frac{1}{m^2}\log \E_z \zeta(z)^m - \frac{1}{m}\frac{\E_z \zeta(z)^m \log \zeta(z)}{\E_z \zeta(z)^m} \label{eq:msaddle}.
\end{align}
Due to the spin-1/2 operators preventing a closed form for the expectation values involved, these saddle point must be solved numerically (cf.~\cite{swingle2023bosonicmodelquantumholography,kavokine2024exact}). In the $q\to\infty$ limit, we show analytically below that the 1RSB saddle point equations cannot be satisfied at temperatures satisfying $(\beta J)^2 q_1^{q-1}\to 0$ as $q \to \infty$. Since we permit unstable solutions, this gives evidence that neither 1RSB nor full RSB phases exist. Although our analysis does not address $\beta J$ exponentially large in $q$, we note that~\cite{swingle2023bosonicmodelquantumholography} provided finite-$N$ numerical results on the ground state ($\beta\to\infty$) and found strong evidence of non-glassiness for all $q \geq 5$. Moreover, if we take the numerical results of of \Cref{fig:rs} at $q=3$ as appropriate values of $\beta J$ and $q_1$ for general $q$, our analytical result suggests non-glassiness for at least all $q \geq 7$, which is in fairly close agreement with the ground state numerics of~\cite{swingle2023bosonicmodelquantumholography}.

Before proceeding with our analysis, we note that the TAP complexity~\cite{monasson1995structural,franz1998effective} is precisely given by \eqref{eq:msaddle} at $m=1$, i.e.,
\begin{align}\label{eq:tapcom}
    S = \frac{(\beta J)^2}{2} \lr{1 - \frac{1}{q}} q_1^{q} + \log \E_z \zeta(z) - \frac{\E_z \zeta(z) \log \zeta(z)}{\E_z \zeta(z)}.
\end{align}
Under the ansatz resembling our glass decomposition,
\begin{align}
    \bm \rho = \sum_a e^{-\beta N f_a} \bm \rho_a
\end{align}
for free energies $f_a$, the TAP free energy assumes that these free energies concentrate at some saddle point $f_*$ in the large-$n$ limit, i.e.,
\begin{align}
    Z = \sum_a e^{-\beta N f_a} = \int \dd{f} \sum_a \delta(f-f_a) e^{-\beta N f} = \int \dd{f} e^{N[S(f) - \beta f]} = e^{N[S(f_*)-\beta f_*]},
\end{align}
where we introduced a complexity measure that counts the number of TAP states
\begin{align}\label{eq:scomp}
    S(f) = \frac{1}{N} \log \sum_a \delta(f-f_a).
\end{align}
Solving for the TAP free energy $\Phi$ over a replicated system $Z_m = \sum_a e^{-\beta m N f_a(\beta)} = e^{-N\beta \Phi}$ gives at $m=1$ precisely the TAP complexity reported in \eqref{eq:tapcom}. However, this does not immediately imply glassiness; one must check the so-called \emph{marginality condition}, which requires that the replicon eigenvalue vanishes~\cite{almeida1978}, and which implies aging in out-of-equilibrium dynamics~\cite{cugliandolo1993analytical,marinari1994replica,giamarchi1996variational,biroli2002out}. Since dynamical and static 1RSB commonly appear together, we focus below on static 1RSB instead of computing the marginality condition. We will see that $S > 0$ for any $q_1 \neq 0$, implying that no static RSB phase ever occurs, due to the $S=0$ saddle point equation~\eqref{eq:msaddle}.

\begin{lemma}[TAP complexity from KL divergence]\label{lem:kl}
    For $Q = \cN(0, \hat q_1 \delta_{\mu\nu})$ and measure
    \begin{align}
        dP(z) = \frac{\zeta(z)}{\E_Q \zeta} dQ(z),
    \end{align}
    we can write the TAP complexity as
    \begin{align}
        S = \frac{(\beta J)^2 q_1^q}{2}\lr{1 - \frac{1}{q}} - D_\KL(P||Q).
    \end{align}
\end{lemma}
\begin{proof}
    The definition of KL divergence gives
    \begin{align}
        D_\KL(P||Q) = \E_P \log \frac{dP}{dQ} = \E_P \log \zeta - \log \E_Q \zeta.
    \end{align}
    We change measure using
    \begin{align}
        \E_P \log \zeta = \int (\log \zeta) \frac{dP}{dQ} dQ = \frac{\E_Q \zeta \log \zeta}{\E_Q \zeta}
    \end{align}
    to obtain
    \begin{align}
        D_\KL(P||Q) = \frac{\E_Q \zeta \log \zeta}{\E_Q \zeta} - \log \E_Q \zeta.
    \end{align}
    From the definition of TAP complexity \eqref{eq:tapcom}, we have
    \begin{align}
        S = \frac{(\beta J)^2}{2} \lr{1 - \frac{1}{q}} q_1^{q} + \log \E_Q \zeta - \frac{\E_Q \zeta \log \zeta}{\E_Q \zeta} = \frac{(\beta J)^2}{2} \lr{1 - \frac{1}{q}} q_1^{q} - D_\KL(P||Q).
    \end{align}
\end{proof}

We use~\cite{fathi2014quantitative}, which provides a sharper bound than the classical log-Sobolev inequality.

\begin{lemma}[Sharpened log-Sobolev bound]\label{lem:ls}
    For
    \begin{align}
        \langle \sigma_{\nu_1}(\tau_1)\cdots \sigma_{\nu_j}(\tau_j) \rangle = \frac{\E_{\xi_\mu(\tau)\sim\cN(0,(\Sigma(\tau_1-\tau_2)-\hat q_1)\delta_{\mu\nu})} \Tr(\cT \exp[\sum_\mu \int_0^\beta d\tau\, (z + \xi_\mu(\tau)) \sigma_\mu(\tau)] \prod_{r}\sigma_{\nu_r}(\tau_r))}{\zeta(z)},
    \end{align}
    define
    \begin{align}
        \eta \coloneqq \frac{1}{3} \sup_x \sum_\nu \int_0^\beta d\tau_1\,d\tau_2\, \lr{\langle \sigma_\nu(\tau_1) \sigma_\nu(\tau_2) \rangle - \langle \sigma_\nu(\tau_1)\rangle\langle \sigma_\nu(\tau_2) \rangle}\Bigg|_{z=\sqrt{\hat q_1}x}.
    \end{align}
    Then for $\lambda = 1 - \hat q_1 \eta$, we have
    \begin{align}
        D_\KL(P||Q) \leq \frac{(\beta J)^2q_1^q}{2} \cdot \frac{1 - \lambda + \lambda \log \lambda}{(1-\lambda)^2}.
    \end{align}
\end{lemma}
\begin{proof}
    Let $\gamma$ denote the standard Gaussian density; define $d\nu = f d\gamma$ for
    \begin{align}
        f(x) = \frac{\zeta(\sqrt{\hat q_1} x)}{\E_Q \zeta}
    \end{align}
    so $D_\KL(P||Q) = D_\KL(\nu||\gamma)$. Since we can rewrite
    \begin{align}
        d\nu = (2\pi)^{-3/2} e^{-V(x)}dx
    \end{align}
    for
    \begin{align}
        V(x) = \frac{1}{2}\norm{x}^2 - \log \frac{\zeta(\sqrt{\hat q_1} x)}{\E_q \zeta},
    \end{align}
    then
    \begin{align}
        \nabla^2 V(x) \succeq \lambda I_3
    \end{align}
    implies a Poincar\'e inequality with constant $\lambda$ due to Bakry-\'Emery. Computing the Hessian, we have first derivative
    \begin{align}\label{eq:fder}
        \partial_{z_\nu} \log \zeta(z) = \int_0^\beta \langle \sigma_\nu(\tau)\rangle d\tau = \int_0^\beta a_\nu(z;\tau) d\tau
    \end{align}
    and second derivative
    \begin{align}
        \frac{\partial^2}{\partial z_\mu \partial z_\nu} \log \zeta(z)\Bigg|_{z=\sqrt{\hat q_1} x} = \int_0^\beta d\tau_1\,d\tau_2\, \lr{\langle \sigma_\mu(\tau_1)\sigma_\nu(\tau_2)\rangle - \langle \sigma_\mu(\tau_1)\rangle\langle\sigma_\nu(\tau_2)\rangle}\Bigg|_{z=\sqrt{\hat q_1} x}.
    \end{align}
    Hence, we have
    \begin{align}
        \nabla^2_x V(x) = I - \hat q_1 \nabla^2_z \log \zeta(z)\Bigg|_{z=\sqrt{\hat q_1} x} \geq (1-\hat q_1 \eta)I_3.
    \end{align}
    By Theorem 1 of~\cite{fathi2014quantitative}, setting $\lambda = 1 - \hat q_1 \eta/3$ implies
    \begin{align}
        D_\KL(P||Q) \leq \frac{1 - \lambda + \lambda \log \lambda}{(1-\lambda)^2}\frac{1}{2}\int \frac{\norm{\nabla f}^2}{f} d\gamma
    \end{align}
    since $\zeta(z)$ is even.
    We compute
    \begin{align}
        \int \frac{\norm{\nabla f}^2}{f} d\gamma = \int f \norm{\nabla \log f}^2 d\gamma = \int \norm{\nabla \log f}^2 d\nu = \hat q_1 \E_P \norm{\nabla \log \zeta}^2.
    \end{align}
    Since \eqref{eq:fder} gives
    \begin{align}
        \E_P \norm{\nabla \log \zeta(z)}^2 = \E_P \sum_\nu \lr{\int_0^\beta a_\nu(z;\tau)d\tau}^2 = \beta^2 \frac{\E_z \zeta(z) \sum_\nu a_\nu(z;\tau_1) a_\nu(z;\tau_2)}{\E_z \zeta(z)} = \beta^2 q_1
    \end{align}
    by the 1RSB saddle point equation
    \begin{align}
        q_1 &= \frac{\E_z\left[\zeta(z)^m \sum_\nu a_\nu(z;\tau_1) a_\nu(z;\tau_2)\right]}{\E_z\left[\zeta(z)^m\right]}
    \end{align}
    at $m=1$, and since
    \begin{align}
        \hat q_1 = J^2 q_1^{q-1},
    \end{align}
    we conclude that
    \begin{align}
        \int \frac{\norm{\nabla f}^2}{f} d\gamma = (\beta J)^2 q_1^q.
    \end{align}
    This gives the claimed inequality
    \begin{align}
        D_\KL(P||Q) \leq \frac{(\beta J)^2q_1^q}{2} \cdot \frac{1 - \lambda + \lambda \log \lambda}{(1-\lambda)^2}.
    \end{align}
\end{proof}

\begin{corollary}
    Let $\lambda = 1 - (\beta J)^2 q_1^{q-1} \geq 1 - x_*$ for some $x_* < 1$. Then the TAP complexity satisfies
    \begin{align}
        S \geq \frac{(\beta J)^2 q_1^q}{2}\lr{1 - \frac{1}{q}-  \frac{1-\lambda + \lambda \log \lambda}{(1-\lambda)^2}} \geq \frac{(\beta J)^2 q_1^q}{4}\lr{1 - \frac{2}{q} - \lr{\frac{1}{3} + \frac{2x_*}{9(1-x_*)}}(\beta J)^2 q_1^{q-1}}.
    \end{align}
\end{corollary}
\begin{proof}
    We bound
    \begin{align}
        \eta = \frac{1}{3}\sum_\nu \int_0^\beta d\tau_1\,d\tau_2\, \lr{\langle \sigma_\nu(\tau_1) \sigma_\nu(\tau_2) \rangle - \langle \sigma_\nu(\tau_1)\rangle\langle \sigma_\nu(\tau_2) \rangle} \leq \beta^2
    \end{align}
    using $|\langle \sigma_\nu(\tau_1) \sigma_\nu(\tau_2) \rangle^\mathrm{conn}| \leq 1$. We apply \Cref{lem:ls} with
    \begin{align}
        \lambda \geq 1 - \beta^2 \hat q_1 = 1 - (\beta J)^2 q_1^{q-1}.
    \end{align}
    Taking $x = (\beta J)^2 q_1^{q-1}$ and assuming $x < x_* < 1$, we obtain
    \begin{align}
        D_\KL(P||Q) \leq \frac{x q_1}{2} \cdot \frac{1-\lambda + \lambda \log \lambda}{(1-\lambda)^2} \leq \frac{x q_1}{4} \left[1 + \lr{\frac{1}{3} + \frac{2x_*}{9(1-x_*)}}x\right].
    \end{align}
    Applying \Cref{lem:kl} gives the claimed result.
\end{proof}

For $\beta J$ not exponentially growing in $q$, we thus find that $S > 0$ is lower-bounded by a constant (depending on $q$), ruling out the $S=0$ condition for the 1RSB saddle point equations at $m=1$. Note that we did not require stability; hence, this also provides evidence against full RSB.

\section{Full RSB solution of the \texorpdfstring{$q$}\ -local Hamiltonian ensemble}\label{app:frsb}
We show here the full RSB solution of the ensemble, obtaining the Parisi order parameter $q(x)$ at arbitrary $\beta$ and $\kappa$. (To avoid confusion between the Parisi order parameter and the locality of the Hamiltonian, we use $\kappa$ for locality in this section.) The computation proceeds by applying the replica trick to write the free energy in terms of a generic choice of overlap $G_{ab}(\tau_1,\tau_2)$ (and its Lagrange multiplier $\Sigma$), imposing the full RSB ansatz, then writing the resulting saddle point equations.

\subsection{Replica trick}
We start from
\begin{align}
    \EE{Z^n} &= \int \cD^N \Omega\, \cD G\, \cD\Sigma\, \exp\Bigg\{\int_0^\beta \int_0^\beta d\tau_1 \,d\tau_2 \Bigg(\frac{J^2 N}{2\kappa} \sum_{a,b=1}^n G_{ab}^\kappa(\tau_1, \tau_2) \nonumber \\
    &\qquad \qquad \qquad \qquad \qquad \quad - \frac{N}{2}\sum_{a,b=1}^n \Sigma_{ab}(\tau_1, \tau_2) \left[G_{ab}(\tau_1, \tau_2) - \frac{1}{N} \sum_{r,\mu} s_{r\mu}^a(\tau_1) s_{r\mu}^b(\tau_2)\right]\Bigg)\Bigg\}.
\end{align}
We reparameterize $G, \Sigma$ to have imaginary times $\tau_1,\tau_2 \in [0, 1]$ instead of $[0,\beta]$, so the temperature dependence is made explicit; we rewrite
\begin{align}
    Z &= \int \cD G \, \cD\Sigma\, \exp[NX]
\end{align}
for
\begin{align}
    X &= \frac{\beta^2}{2} \sum_{a,b=1}^n \int_0^1 d\tau_1\,d\tau_2 \lr{\frac{J^2 G_{ab}^\kappa(\tau_1, \tau_2)}{\kappa} - \Sigma_{ab}(\tau_1, \tau_2) G_{ab}(\tau_1, \tau_2)} \nonumber\\
    &\qquad + \log \int \cD^N \Omega \exp\left[\frac{\beta^2}{2} \sum_{a,b=1}^n \int_0^1 d\tau_1\,d\tau_2\, \Sigma_{ab}(\tau_1,\tau_2) \lr{s^a(\tau_1) \cdot s^b(\tau_2)}\right]
\end{align}
where we used shorthand
\begin{align}
    s^a(\tau_1) \cdot s^b(\tau_2) &= \frac{1}{N} \sum_{r=1}^N \sum_\mu s_{r\mu}^a(\tau_1) s_{r\mu}^b(\tau_2).
\end{align}
As $N\to\infty$, the free energy converges to $-X/\beta$ evaluated at $G,\Sigma$ satisfying the saddle point equations
\begin{align}
    G_{ab}(\tau_1,\tau_2) &= \left\langle \cT \, \sum_\mu \sigma^a_\mu(\tau_1) \sigma^b_\mu(\tau_2) \right\rangle_{\cS_\loc}, \qquad \Sigma_{ab}(\tau_1,\tau_2) = J^2 G_{ab}^{\kappa-1}(\tau_1,\tau_2)
\end{align}
for action
\begin{align}
    \cS_\loc = \sum_{a=1}^n \cS_0[\sigma^a(\tau)] - \frac{(\beta J)^2}{2} \int_0^1 d\tau_1\,d\tau_2\, \sum_{a,b=1}^n \Sigma_{ab}(\tau_1,\tau_2) \sum_\mu \sigma^a_\mu(\tau_1) \sigma^b_\mu(\tau_2),
\end{align}
where $\cS_0$ is a single-site action that we leave unspecified.
Substituting $G_{ab}(\tau_1,\tau_2) = G(\tau_1-\tau_2)$ and $G_{a\neq b}(\tau_1,\tau_2) = G_{ab}$, we obtain free energy per site
\begin{align}\label{eq:lilf}
    f &= \frac{\beta J^2}{2}\lr{1 - \frac{1}{\kappa}} \int_0^1 d\tau_1\,d\tau_2\, G^\kappa(\tau_1-\tau_2) + \lim_{n\to 0} \frac{1}{\beta n}\left[(\beta J)^2\lr{1 - \frac{1}{\kappa}}\sum_{a < b} G_{ab}^\kappa - \log Z_\loc\right]
\end{align}
by the replica trick, where
\begin{align}
    G_{ab} &= \left\langle \cT \, \sum_\mu \sigma^a_\mu(\tau_1) \sigma^b_\mu(\tau_2) \right\rangle_{\cS_\loc}, \qquad
    G(\tau) = \left\langle \cT \, \sum_\mu \sigma^a_\mu(\tau_1) \sigma^a_\mu(\tau_2) \right\rangle_{\cS_\loc}, \qquad
    Z_\loc = \int \cD \Omega \exp[-\cS_\loc]\label{eq:g2points}\\
    \cS_\loc &= \sum_{a=1}^n \cS_0[\sigma^a(\tau)] - \frac{(\beta J)^2}{2} \int_0^1 d\tau_1\,d\tau_2\, G^{\kappa-1}(\tau_1-\tau_2) \sum_\mu \sigma^a_\mu(\tau_1) \sigma^a_\mu(\tau_2) - (\beta J)^2 \sum_{a < b} G^{\kappa-1}_{ab} \int_0^1 d\tau_1\,d\tau_2\, \sum_\mu \sigma^a_\mu(\tau_1) \sigma^b_\mu(\tau_2).\label{eq:sloc}
\end{align}

\subsection{Full RSB free energy}
We recursively construct the full RSB ansatz and follow a similar solution as that in \cite{kavokine2024exact} for the quantum Heisenberg model. We define a self-replica function $G:[-1,1]\to\mathbb{R}$ and two sequences $(q_p)_{0,\dots, k}$ and $(m_p)_{0,\dots,k}$. We set $m_0=n$ and define the $k$th stage of RSB according to
\begin{align}
G^\rsbk_k &= (q_k - q_{k-1})\,U_{m_k}, \\
G^\rsbk_p &= \operatorname{diag}_{\,m_p/m_{p+1}}\lr{G^\rsbk_{p+1}} + (q_p - q_{p-1})U_{m_p} \text{for all } p\in\{0,\dots,k-1\} \\
G^\rsbk &\equiv G^\rsbk_0 .
\end{align}
Here, $U_m$ is the $m\times m$ matrix of ones, $\operatorname{diag}_r(A)$ is the block-diagonal matrix with $r$ copies of $A$, $m_0=n$, and $q_{-1}=0$. Introducing notation
\begin{align}
    \bar \sigma^a_\mu = \int_0^1 d\tau \sigma_\mu^a(\tau), \qquad \sigma^a(\tau_1) \cdot \sigma^b(\tau_2) = \sum_\mu \sigma^a_\mu(\tau_1) \sigma^b_\mu(\tau_2)
\end{align}
we have at the $k$th stage of RSB that
\begin{align}
    \cS_\loc^\rsbk &= \sum_{a=1}^n \lr{\cS_0[\sigma^a(\tau)] - \frac{(\beta J)^2}{2} \int_0^1 d\tau_1\,d\tau_2\, \lr{G^{\kappa-1}(\tau_1-\tau_2) - J^2 q_k^{\kappa-1}} \sigma^a(\tau_1) \cdot \sigma^a(\tau_2)} - \frac{(\beta J)^2}{2} \sum_{a,b} \lr{G^\rsbk}^{\kappa-1}_{ab} \bar \sigma^a\cdot \bar \sigma^b.
\end{align}
Introducing an auxiliary field $h$ (with three components, and which we will ultimately set to $h=0$) and writing out the path integral, the local partition function at step $p$ is
\begin{align}
    Z_p^\rsbk(h) &= \int \cD \Omega \exp\Bigg[-\sum_{a=1}^n\lr{ \cS_0[s^a(\tau)] - \frac{(\beta J)^2}{2} \int_0^1 d\tau_1\,d\tau_2\, \lr{G^{\kappa-1}(\tau_1-\tau_2) - J^2 q_k^{\kappa-1}} s^a(\tau_1) \cdot s^a(\tau_2)}\nonumber\\
    &\qquad\qquad\qquad\quad + \frac{(\beta J)^2}{2} \sum_{a,b} \lr{G^\rsbk_p}^{\kappa-1}_{ab} \bar s^a\cdot \bar s^b + \beta h \cdot \sum_{a=1}^{m_k} \bar s^a\Bigg].
\end{align}
Evaluating at steps $p$ and $p+1$ then applying a Hubbard-Stratonovich transformation, one finds the recurrence relation
\begin{align}
    Z_p^\rsbk(h) &= \int dh_p \, \mathbb{G}\lr{h_p | J^2(q_p^{\kappa-1} - q_{p-1}^{\kappa-1})} \lr{Z_{p+1}^\rsbk(h+h_p)}^{m_p/m_{p+1}},
\end{align}
where $\mathbb G(y|v)$ places a Gaussian distribution on $y$ with mean zero and variance $v$.
The replica computation for full RSB then takes limits $n \to 0$ and $k \to \infty$ with $m_k=1$. The variable $q_p$ becomes the Parisi order parameter $q(x)$ for $x \in [0,1]$, and $m_p$ becomes $x$. Given Laplacian operator $\nabla^2  := \sum_i \frac{\partial^2 }{\partial h_i^2}$, the heat-kernel identity
\begin{align}
    \left[\exp(\frac{v}{2}\nabla^2)\cdot f\right](h) = \int dh' \mathbb G(h'|v) f(h+h')
\end{align}
for any sufficiently regular function $f$ implies that (taking $Z_p^\rsbk(h) \to \zeta(x,h))$
\begin{align}
    \zeta(x-dx, h) = \exp(\frac{J^2}{2}dq^{\kappa-1}(x)\nabla^2) \zeta(x,h)^{1-dx/x},
\end{align}
i.e.,
\begin{align}
    \frac{\partial \zeta}{\partial x} = -\frac{J^2}{2} \frac{d q^{\kappa-1}}{dx} \nabla^2 \zeta + \frac{1}{x}\zeta \log \zeta
\end{align}
with boundary condition $\zeta(1,h) = \int \cD \Omega \exp[-\cS_\infty(h)]$ for single-site action
\begin{align}\label{eq:sinfty}
    \cS_\infty(h) &= \cS_0[s(\tau)] - \int_0^1 d\tau_1\, d\tau_2\, \frac{(\beta J)^2}{2}\lr{G^{\kappa-1}(\tau_1-\tau_2) - q^{\kappa-1}(1)} s(\tau_1) \cdot s(\tau_2) - \beta h \cdot \bar s.
\end{align}
Using~\cite{gabay1982symmetry}
\begin{align}
    \lim_{n\to 0} \sum_{a < b} G_{ab}^\kappa = -\frac{1}{2}\int_0^1 dx\, q(x)^\kappa
\end{align}
and
\begin{align}
    \lim_{n\to0}\frac{1}{n} \log Z_\loc = \int dh \, \mathbb{G}(h|J^2 q^{\kappa-1}(0)) \lim_{x\to 0} \frac{1}{x} \log \zeta(x, h),
\end{align}
we plug the above quantities into \eqref{eq:lilf} (and set $h=0$) to obtain free energy
\begin{align}
    f &= \frac{\beta J^2}{2}\lr{1 - \frac{1}{\kappa}} \int_0^1 d\tau_1\,d\tau_2\, \lr{G(\tau_1-\tau_2)^\kappa - \int_0^1 dx\, q(x)^\kappa} - \frac{1}{\beta}\phi(0,0)\\
    &= \frac{\beta J^2}{2}\lr{1 - \frac{1}{\kappa}} \int_0^1 d\tau\, \lr{G(\tau)^\kappa - \int_0^1 dx\, q(x)^\kappa} - \frac{1}{\beta}\phi(0,0),
\end{align}
where we applied the periodicity of $G$ and introduced $\phi(x,h) = (1/x)\log\zeta(x,h)$ which satisfies
\begin{align}
    \frac{\partial \phi}{\partial x} &= -\frac{J^2}{2} \frac{dq^{\kappa-1}}{dx} \lr{\nabla^2 \phi + x(\nabla \phi)^2}\\
    \phi(1,h) &= \log \int \cD \Omega \exp[-\cS_\infty(h)]
\end{align}
for $\cS_\infty(h)$ defined in \eqref{eq:sinfty}. To evaluate the quantities in \eqref{eq:g2points}, we compute the expectation of the observable $G(\tau_1,\tau_2) = \left\langle s^a(\tau_1) \cdot s^a(\tau_2) \right\rangle$ on the thermal state similarly to above. In the full RSB ansatz, we introduce the quantity $G_p^\rsbk(h)$ which evaluates $G(\tau_1,\tau_2)$ with the matrix $Q^\rsbk$ at step $p$ (with an external field). Then the $a=b$ case is
\begin{align}\label{eq:saa}
    \langle s^a(\tau_1) \cdot s^b(\tau_2) \rangle_p^\rsbk(h) &= \int dh_p\, \mathbb G\lr{h_p|J^2(q_p^{\kappa-1}-q_{p-1}^{\kappa-1})} \frac{\langle s^a(\tau_1) \cdot s^b(\tau_2) \rangle_{p+1}^\rsbk(h+h_p) Z_{p+1}^\rsbk(h+h_p)^{m_p/m_{p+1}}}{Z_p^\rsbk(h)}.
\end{align}
For $a \neq b$, we similarly have
\begin{align}\label{eq:sab}
    \langle \bar s^a \cdot \bar s^b \rangle_p^\rsbk(h) &= \int dh_p\, \mathbb G\lr{h_p|J^2(q_p^{\kappa-1}-q_{p-1}^{\kappa-1})} \frac{\langle \bar s^a \cdot \bar s^b \rangle_{p+1}^\rsbk(h+h_p) Z_{p+1}^\rsbk(h+h_p)^{m_p/m_{p+1}}}{Z_p^\rsbk(h)}.
\end{align}
As $n\to 0$ and $k\to\infty$, $G_p^\rsbk(h)$ becomes $\rho(x,h)$ and satisfies (temporarily suppressing $\tau_1,\tau_2$)
\begin{align}
    \zeta(x-dx,h)\rho(x-dx,h) = \exp[\frac{J^2}{2} dq(x)^{\kappa-1}\nabla^2]\rho(x,h)\zeta(x,h)^{1-dx/x},
\end{align}
i.e., for variable substitution $\beta u(x,h) = (1/x)\nabla \log \zeta = \nabla \phi(x, h)$,
\begin{align}
    \frac{\partial \rho}{\partial x} = -\frac{J^2}{2} \frac{dq^{\kappa-1}}{dx} \lr{\nabla^2 \rho + 2 \beta x (u\cdot \nabla) \rho}
\end{align}
with boundary condition $\rho(1,h)_{\tau_1,\tau_2} = \langle s(\tau_1)\cdot s(\tau_2)\rangle_{\cS_\infty(h)}$. Similarly, $u$ satisfies
\begin{align}
    \frac{\partial u}{\partial x} = -\frac{J^2}{2} \frac{dq^{\kappa-1}}{dx} \lr{\nabla^2 u + 2 \beta x (u\cdot \nabla) u}
\end{align}
with boundary condition $u(1,h) = \langle \bar s\rangle_{\cS_\infty(h)}$, as seen by evaluating $\nabla \phi(1,h)$ given our boundary condition for $\phi(1,h)$. A similar computation for the $a \neq b$ case gives the PDE
\begin{align}
    \frac{\partial \nu}{\partial x} = -\frac{J^2}{2} \frac{dq^{\kappa-1}}{dx} \lr{\nabla^2 \nu + 2 \beta x (u\cdot \nabla) \nu}
\end{align}
for $\nu(y, x, h)$ with boundary condition $\nu(x, x, h) = u^2(x,h)$.
Plugging these answers into \eqref{eq:saa} and \eqref{eq:sab} (using \eqref{eq:g2points}), we obtain equations
\begin{align}
    G(\tau_1-\tau_2) &= \int dh\, \mathbb G\lr{h|J^2 q^{\kappa-1}(0)} \rho(0,h)_{\tau_1,\tau_2} = \rho(0, 0)_{\tau_1,\tau_2}\\
    q(x) &= \int dh \, \mathbb G\lr{h|J^2 q^{\kappa-1}(0)} \nu(0, x, h) = \nu(0, x, 0)
\end{align}
since we anticipate $q(0) = 0$, producing a $\delta$ function at $h=0$. We reparameterize these PDEs as in~\cite{kavokine2024exact} to rewrite $\rho,\nu$ in terms of a probability distribution $P(h)$.

We summarize the full RSB solution as follows. For single-site action
\begin{align}
    \cS_\infty(h) &= \cS_0[s(\tau)] - \int_0^1 d\tau_1\, d\tau_2\, \frac{(\beta J)^2}{2}\lr{G(\tau_1-\tau_2)^{\kappa-1} - q(1)^{\kappa-1}} s(\tau_1) \cdot s(\tau_2) - \beta \int_0^1 d\tau\, h \cdot s(\tau)
\end{align}
we have
\begin{align}
    G(\tau) = \int dh\, P(1,h) \langle \cT \, s(\tau) \cdot s(0) \rangle_{\cS_\infty(h)}, \qquad q(x) = \int dh \, P(x,h) u(x,h)^2
\end{align}
for $P, u$ obtained by solving
\begin{align}
    \frac{\partial P}{\partial x} &= \frac{J^2}{2} \frac{dq^{\kappa-1}}{dx} \lr{\nabla^2 P - 2\beta x \nabla(u \cdot P)}, \qquad P(0, h) = \delta(h)\\
    \frac{\partial u}{\partial x} &= -\frac{J^2}{2} \frac{dq^{\kappa-1}}{dx} \lr{\nabla^2 u + 2 \beta x (u\cdot \nabla) u}, \qquad u(1,h) = \int_0^1 d\tau\, \langle s(\tau) \rangle_{\cS_\infty(h)}.
\end{align}
The free energy is then
\begin{align}
    F &= \frac{\beta J^2}{2}\lr{1 - \frac{1}{\kappa}} \int_0^1 d\tau\, \lr{G(\tau)^\kappa - \int_0^1 dx\, q(x)^\kappa} - \frac{1}{\beta}\phi(0,0)
\end{align}
for
\begin{align}
    \frac{\partial \phi}{\partial x} &= -\frac{J^2}{2} \frac{dq^{\kappa-1}}{dx} \lr{\nabla^2 \phi + x(\nabla \phi)^2}, \qquad \phi(1,h) = \log \int \cD \Omega \exp[-\cS_\infty(h)].
\end{align}

\end{document}